\documentclass[a4paper,11pt]{article}
\usepackage[pdftex]{graphicx}
\usepackage[T1]{fontenc}
\usepackage{lmodern}
\usepackage{slashed}
\usepackage[utf8]{inputenc}
\usepackage[english]{babel}
\usepackage{microtype}
\usepackage{cite}
\usepackage{amsmath}
\usepackage{amssymb}
\usepackage{amsfonts}
\usepackage[nottoc,notlot,notlof]{tocbibind}
\usepackage{upgreek}
\usepackage{mathtools}
\numberwithin{equation}{section}
\allowdisplaybreaks
\usepackage[all]{xy}
\usepackage{color}
\usepackage{graphicx}
\graphicspath{{images/}}
\usepackage{geometry}
\usepackage[toc,page]{appendix}
\usepackage{hyperref}
\usepackage[normalem]{ulem}

\newcommand{\diff}{\mathrm{d}}

\newcommand{\Hom}{\mathrm{Hom}}

\newcommand{\Spin}{\mathrm{Spin}}

\newcommand{\Cl}{C\ell}

\newcommand{\stab}{\mathfrak{stab}}

\newcommand{\fso}{\mathfrak{so}}

\newcommand{\fm}{\mathfrak{m}}

\newcommand{\fg}{\mathfrak{g}}
\newcommand{\fp}{\mathfrak{p}}
\newcommand{\fh}{\mathfrak{h}}

\newcommand{\ZZ}{\mathbb{Z}}

\usepackage{bm}
\usepackage{ragged2e}
\usepackage{appendix}
\usepackage{slashed}
\usepackage{bbold}
\usepackage{cancel}
\usepackage{multirow}
\usepackage{array,multirow}
\usepackage{multicol}
\usepackage{multirow}
\usepackage{array}
    \newcolumntype{P}[1]{>{\centering\arraybackslash}p{#1}}
    \newcolumntype{M}[1]{>{\centering\arraybackslash}m{#1}}

\usepackage{amsthm}
\newtheorem{theorem}{Theorem}[section]
\newtheorem*{theorem*}{Theorem}
\newtheorem{proposition}{Proposition}[section]
\newtheorem*{proposition*}{Proposition}
\newtheorem{definition}{Definition}[section]
\newtheorem{remark}{Remark}[section]
\newtheorem{lemma}{Lemma}[section]
\newtheorem{corollary}{Corollary}[section]

\DeclareMathAlphabet{\mathpzc}{OT1}{pzc}{m}{it}

\definecolor{blue-violet}{rgb}{0.54, 0.17, 0.89}
\definecolor{PineGreen}{cmyk}{0.92, 0, 0.59, 0.25}
\definecolor{YellowOrange}{cmyk}{0, 0.42, 1, 0}
\definecolor{orange}{rgb}{0.95, 0.5, 0.1}

\usepackage[dvipsnames]{xcolor}

\usepackage{caption} 
\usepackage{float}

\begin{document}

\begin{titlepage}
\begin{flushright}
\par\end{flushright}
\vskip 0.5cm
\begin{center}
\textbf{\Huge \bf Fermionic Spencer Cohomologies}\\
\vskip .2cm
\textbf{\Huge \bf of $D=11$ Supergravity}

\vskip 1cm 

\large {\bf C.~A.~Cremonini}$^{~a,}$\footnote{carlo.alberto.cremonini@gmail.com}, 
\large {\bf P.~A.~Grassi}$^{~b, c, d,}$\footnote{pietro.grassi@uniupo.it},
\large {\bf R.~Noris}$^{~e,}$\footnote{noris@fzu.cz},
\large {\bf L.~Ravera}$^{~c,f, g,}$\footnote{lucrezia.ravera@polito.it}
\large {\bf A.~Santi}$^{~h,}$\footnote{asanti.math@gmail.com}

\vskip .5cm {
\small
{$^{(a)}$ \it Arnold-Sommerfeld-Center for Theoretical Physics, Ludwig-Maximilians-Universit\"at M\"unchen, Theresienstr. 37, D-80333 Munich, Germany}\\
{$^{(b)}$ \it DiSIT, Universit\`a del Piemonte Orientale, Viale T. Michel 11, 15121 Alessandria, Italy}\\
{$^{(c)}$ \it INFN, Sezione di Torino, Via P. Giuria 1, 10125 Torino, Italy} \\
 {$^{(d)}$ \it Theoretical Physics Department,}  {\it CERN,}  {\it CH-1211 Geneva 23, Switzerland}\\
 {$^{(e)}$ \it CEICO, Institute of Physics of the Czech Academy of Sciences, \\
Na Slovance 2, 182 21 Prague 8, Czech Republic}\\
{$^{(f)}$ \it Politecnico di Torino, Corso Duca degli Abruzzi, 24, 10129 Torino, Italy}\\
{$^{(g)}$ \it Grupo de Investigaci\'{o}n en Física Te\'{o}rica, Universidad Cat\'{o}lica De La Sant\'{i}sima Concepci\'{o}n, Chile}\\
{$^{(h)}$ \it Universit\`a degli Studi di Roma Tor Vergata,
Via della Ricerca Scientifica 1, 00133 Roma, Italy}\\
}
\end{center}

\begin{abstract} 
We combine the theory of Cartan-Tanaka prolongations with the Molien-Weyl integral formula and Hilbert-Poincaré series to compute 
the Spencer cohomology groups of the $D=11$ 
Poincaré superalgebra $\mathfrak p$, relevant for superspace formulations of $11$-dimensional supergravity in terms of nonholonomic superstructures. This includes novel fermionic Spencer groups, providing with new cohomology classes of $\mathbb Z$-grading $1$ and form number $2$.
Using the Hilbert-Poincaré series and the Euler characteristic, we also explore Spencer cohomology contributions in higher form numbers. We then propose a new general definition of filtered deformations of graded Lie superalgebras along  first-order fermionic directions and investigate such deformations of $\mathfrak p$ that are maximally supersymmetric. In particular, we establish a no-go type theorem for maximally supersymmetric filtered subdeformations of $\mathfrak p$
along timelike (i.e., generic) first-order fermionic directions.
\end{abstract}

\vfill{}
\vspace{1.5cm}
\end{titlepage}

\setcounter{footnote}{0}
\tableofcontents

\section{Introduction}\label{sec1}

In the last few years, an interesting approach to investigate the bosonic backgrounds of supergravity theories  emerged \cite{Santi:2010kb,Santi:2011mc,Figueroa-OFarrill:2015rfh,Figueroa-OFarrill:2015tan,deMedeiros:2016srz,Figueroa-OFarrill:2016khp,deMedeiros:2018ooy,Santi:2019kpx}: in particular, the Killing superalgebras of such backgrounds, which are Lie superalgebras generated by the respective Killing spinors and Killing vectors, were put into correspondence with \emph{filtered deformations} of graded subalgebras of the Poincar\'e superalgebras.

For instance, the following result in the context of $D=11$ supergravity was proved:
\begin{theorem*}\!\!\!\!\cite{Figueroa-OFarrill:2015rfh}
    The Killing superalgebra of an 11-dimensional bosonic supergravity background is a filtered subdeformation of the $D=11$ Poincar\'e superalgebra.
\end{theorem*}
The converse was then established in \cite{Figueroa-OFarrill:2016khp} for highly supersymmetric backgrounds, thus giving rise to a bijective correspondence between the latter and a certain type of filtered deformations.
\begin{theorem*}\!\!\!\!\cite{Figueroa-OFarrill:2016khp}\hfill
\begin{itemize}
    \item[(i)]   Let $(M,g,F)$ be an $11$-dimensional Lorentzian spin manifold
    with a closed $4$-form $F\in\Omega^4(M)$.  If the real vector space of Killing spinors of $(M,g,F)$ has dimension $>16$ (i.e., strictly more than half the rank of the associated Majorana spinor bundle), then the bosonic field equations of $11$-dimensional supergravity are automatically satisfied.
    \item[(ii)] Any realizable (see \cite{Figueroa-OFarrill:2016khp} for details on this notion, which is of a cohomological nature) highly supersymmetric filtered subdeformation of the $D=11$ Poincaré superalgebra is a subalgebra of the Killing superalgebra of a highly supersymmetric $D=11$ bosonic supergravity background.
\end{itemize}
\end{theorem*}
For applications of filtered deformations to supergravity and rigid supersymmetric field theories in other dimensions, see \cite{deMedeiros:2018ooy, deMedeiros:2016srz, MR4316462,MR4506436}.\\

One of the objectives of the present paper is to investigate the subdeformations of Poincar\'e superalgebras from a more comprehensive perspective. The main algebraic tool to classify filtered deformations is the \emph{(generalized) Spencer cohomology}, which we briefly describe here for the $D=11$ Poincar\'e superalgebra (see \cite{Cheng-Kac,Figueroa-OFarrill:2015rfh,Figueroa-OFarrill:2016khp} for a general introduction). 
Let $V$ be a real vector space of dimension $D=11$ endowed with a Lorentzian inner product $\eta$ of signature ``mostly minus'' and associated Clifford algebra $\Cl(V)$.  We let $S$ be an irreducible $\Cl(V)$-module and denote by $\mathfrak g:=\fso(V)$
the Lorentz Lie algebra, corresponding to the connected spin group $G:=\Spin^\circ(V)$. The $D=11$ {\it Poincar\'e superalgebra} $\mathfrak{p}=\mathfrak p_{\bar 0}\oplus\mathfrak p_{\bar 1}$ is the $\mathbb Z$-graded Lie
superalgebra
\begin{equation}
\label{eq:Z-grading}
    \mathfrak{p} = \mathfrak{p}_0 \oplus \mathfrak{p}_{-1} \oplus \mathfrak{p}_{-2} = \mathfrak{so} \left( V \right) \oplus S \oplus V \,,
\end{equation}
allowing for the following non-trivial brackets:
\begin{itemize}
\item $[-,-]: \fp_0 \times \fp_i \to \fp_i$, which consists of the
  adjoint action of $\fso(V)$ on itself and its natural actions on $V$
  and $S$;
\item $[-,-]: \fp_{-1} \times \fp_{-1} \to \fp_{-2}$, which is the Dirac current 
$\kappa:S\otimes S\to V$ 
of a spinor $s\in S$. It is an $\fso(V)$-equivariant symmetric map (unique up to scalings, since $S$ is $\fso(V)$-irreducible). 
\end{itemize}
In particular, the even Lie subalgebra $\fp_{\bar 0}=\fso(V) \oplus V$ is
the Poincaré algebra. We note that \eqref{eq:Z-grading} is compatible with the $\mathbb{Z}_2$-grading, in that $\mathfrak p_{\bar 0}=\mathfrak{p}_{0}\oplus\mathfrak{p}_{-2}$, $\mathfrak{p}_{\bar 1}=\mathfrak{p}_{-1}$, and with the Lie superalgebra structure, since $[\mathfrak p_{i},\mathfrak p_j]\subset \mathfrak p_{i+j}$ for all $i,j\in\mathbb Z$, where we set $\mathfrak{p}_{k}=0$ if $k\neq -2,-1,0$ for convenience. We also recall that $\eta(\kappa(s,s),v) = \left<s, v\cdot s\right>$ for all $s\in S$, $v \in V$, where 
$\left<-,-\right>$ is the
$\fso(V)$-invariant symplectic structure on $S$ and $\cdot$ refers to the Clifford action. The space $\mathfrak p_{-1}=S$ collects the odd spinorial translations and $\mathfrak p_{-2}=V$ the even translations. 

Denoting by $\mathfrak{m} = \mathfrak{p}_{-1} \oplus \mathfrak{p}_{-2}$ the ($2$-step nilpotent) supertranslation ideal of $\mathfrak p$, we define the {\it Spencer cochains} as the cochains of $\mathfrak{m}$ with values in the whole Poincar\'e superalgebra:
\begin{equation}
\label{eq:Spencer-cohomology}
    C^{\bullet} \left( \mathfrak{m} , \mathfrak{p} \right) = \left(\bigwedge^\bullet \mathfrak{m}\right)^* \!\otimes \mathfrak{p} \,,
\end{equation}
where the symbol of exterior algebra has to be understood here in the supersense (cochains are skew-symmetric as usual, except that they are symmetric on any pair of entries from $\mathfrak m_{\bar 1}=S$).
In other words, we are considering the Chevalley-Eilenberg cochains with values in the module $\mathfrak{p}$, where the action is defined via the adjoint representation \cite{CE}. We also remark that the cohomology of relative cochains  
$C^{\bullet} \left( \mathfrak{p} , \mathfrak{p}_0 ; \mathfrak{p} \right)=\operatorname{Hom}_{\mathfrak p_0}(\Lambda^\bullet (\mathfrak p/\mathfrak p_0);\mathfrak p)\cong C^{\bullet} \left( \mathfrak{m} , \mathfrak{p} \right)^{\mathfrak{p}_0}$
selects the Lorentz-invariant subcomplex of  \eqref{eq:Spencer-cohomology} given by the basic cochains; however, we do not restrict our analysis  to the trivial 
$\mathfrak{so} \left( V \right)$-modules in this paper. 

The $\mathbb Z_2$-grading of $\mathfrak m$ is extended to $\mathfrak m^*$ by duality ($\mathfrak p_{k}^*$ has the same $\mathbb Z_2$-grading of $\mathfrak p_{k}$) and then additively to any tensor product. Similarly, the $\mathbb Z$-grading of $\mathfrak p$ is extended to the whole $C^{\bullet} \left( \mathfrak{m} , \mathfrak{p} \right)$ by declaring $\deg(\mathfrak{p}_{k}^*)=-\deg(\mathfrak{p}_{k})$ and then additively to tensor products. The space of Spencer cochains can thus be decomposed accordingly:
\begin{equation}
\label{eq:decompositionI}
    C^\bullet \left( \mathfrak{m} , \mathfrak{p} \right) = \bigoplus_{d\in \mathbb{Z}} C^{d,\bullet} \left( \mathfrak{m} , \mathfrak{p} \right) \,.
\end{equation}
Our first main result concerns the Spencer cohomology groups
$H^{d,q}(\mathfrak{m},\mathfrak{p}):=Z^{d,q}(\fm,\fp)/B^{d,q}(\fm,\fp)$
(as usual, $Z^{\bullet}(\fm,\fp)$ and $B^{\bullet}(\fm,\fp)$ are the space of Spencer cocycles and coboundaries, respectively) 
with {\it form number} $q=2$ and any {\it positive $\mathbb Z$-grading} $d$. The ``even'' groups $H^{d,2}(\mathfrak{m},\mathfrak{p})$, $d=2,4$, were already determined in \cite{Figueroa-OFarrill:2015rfh, Figueroa-OFarrill:2016khp}, and the content of $H^{2,2}(\mathfrak{m},\mathfrak{p})$ appropriately identified with the {\it bosonic} content (besides the metric) of Nahm's $D=11$ supergravity multiplet \cite{Nahm}. Even more remarkably, the explicit expressions of the normalized cocycles in $H^{2,2}(\mathfrak{m},\mathfrak{p})$ precisely match the Killing spinor equations of $D=11$ supergravity; in other words, they supply the supersymmetry variation of the gravitino at the first-order \cite{Cremmer:1978km} with a cohomological origin. For the relevance of the Spencer groups in the context of geometric formulations of supergravity theories as target spaces of curvatures of nonholonomic superstructures, we  refer the reader directly to the introduction of \cite{Figueroa-OFarrill:2015rfh} and \cite[\S 5.3]{KST-TG}. All these facts motivated the need of a similar understanding for the ``odd'' groups $H^{d,2}(\mathfrak{m},\mathfrak{p})$, where only the cases $d=1,3$ are potentially non-trivial by $\mathbb Z$-grading reasons.
\begin{theorem}
\label{thm:1}
Let $\mathfrak p=\mathfrak{p}_0 \oplus \mathfrak{p}_{-1} \oplus \mathfrak{p}_{-2} = \mathfrak{so} \left( V \right) \oplus S \oplus V$ be the $D=11$ Poincaré superalgebra and $C^\bullet (\mathfrak{m},\mathfrak{p})$ the associated Spencer complex. Then:
\begin{enumerate}
	\item $H^{1,2}(\mathfrak{m},\mathfrak{p})\cong S$ as an $\mathfrak{so} \left( V \right)$-representation, in particular it is odd, non-trivial, irreducible;
	\item $H^{2,2}(\mathfrak{m},\mathfrak{p})\cong \Lambda^4 V$ as an $\mathfrak{so} \left( V \right)$-representation, again non-trivial and irreducible but even;
	\item $H^{3,2}(\mathfrak{m},\mathfrak{p})=H^{4,2}(\mathfrak{m},\mathfrak{p})=0$.
\end{enumerate}	
All the remaining Spencer cohomology groups $H^{d,2}(\mathfrak{m},\mathfrak{p})$ with $d\geq 5$ vanish by $\mathbb Z$-grading reasons.
\end{theorem}
The content of Theorem \ref{thm:1} will first be motivated through the Molien-Weyl integral formula, and then rigorously established via the approach of Cartan-Tanaka prolongations and combinatorial identities. In this latter approach, the realisation \eqref{eq:decompositionI} is convenient, as it makes the prolongation structure more transparent and allows for representation-theoretic techniques, whereas another complex, naturally isomorphic to the Spencer one, is more suited to the Molien-Weyl techniques.

We call {\it statistics} the parity of the sum of the form number $q$ and $\mathbb Z$-grading $d$, and refer to elements with an even (odd) statistics as {\it bosonic} ({\it fermionic}). In particular, elements of $V$ and $S^*$ are bosonic, while elements of $S$ and $V^*$ are fermionic. 
These are the conventions mostly used in the BRST community, to be compared with the supergravity conventions introduced earlier:
\vskip0.2cm
{\small
\begin{table}[H]
\begin{centering}
\makebox[\textwidth]{%
\begin{tabular}{|c|c|c|c|c|}
\hline
& $\mathbb Z$-grading & $\mathbb Z_2$-grading & Form Number & Statistics\\
\hline
\hline
$V$ & $-2$ & {\rm even}  & $+0$ & {\rm bosonic} \\
\hline
\;$V^*$ & $+2$ & {\rm even}  & $+1$ & {\rm fermionic}\\
\hline
$S$ & $-1$ & {\rm odd}  & $+0$ & {\rm fermionic}\\
\hline
\;$S^*$ & $+1$ & {\rm odd}  & $+1$   &{\rm bosonic}\\
\hline
$\mathfrak{so}(V)$ & $+0$ & {\rm even}  & $+0$  & {\rm bosonic}\\
\hline 
\end{tabular}}
\end{centering}
\caption[]{\label{Conventions} Grading conventions.} \vskip14pt
\end{table}}
\vskip-0.4cm\par\noindent
See also e.g. \cite{PVN-Z2} for a related discussion.

We let $e^a$ ($a=0,\ldots,10$) and $\psi^\alpha$ ($\alpha=1,\ldots,32$) be the differential forms of flat superspace $M=P/G$ with $P$ the (universal cover of the) Poincaré supergroup 
and $G = \mathrm{Spin}^\circ(V)$ 
as
before. They describe the usual vielbein and its supersymmetric partner, the gravitino, and can be conveniently collected into the supervielbein $\mathpzc{E}^I=\left\lbrace e^a , \psi^\alpha \right\rbrace$. The associated $\mathbb{Z}$-grading is enforced by assigning $e^a \to u$ and $\psi^\alpha \to t$, as described in \S\ref{sec2}. Similarly, we denote by $\mathpzc{X}_{\tilde{I}} = \left\lbrace L_{ab} , X_a , q_\alpha \right\rbrace$ the generators of $\mathfrak{p}$, so that a constant differential form $\omega$ with values in $\mathfrak p$ reads as 
$
\omega = \omega^{\tilde{I}}_{IJ} \mathpzc{X}_{\tilde{I}}\otimes \mathpzc{E}^I \wedge \mathpzc{E}^J
$,
where the coefficients $\omega^{\tilde{I}}_{IJ}$ are constant and Einstein summation convention has been tacitly used (this will be the case throughout the whole paper). The commutation relations of the elements of $\mathpzc{E}^I=\left\lbrace e^a , \psi^\alpha \right\rbrace$ follow from the supercommutativity rules and the statistics of $V^*$ and $S^*$. 
The space of such differential forms with its natural differential $d$ coincides with
the Spencer complex. 


We will use the Molien-Weyl formula \cite{DerksenKemper,procesi}  to compute the isotypic components\footnote{We recall that an isotypic component of a given module is the direct sum of all its irreducible submodules of fixed isomorphism class.} of the complex and its Euler characteristic, for any form number and for several isotypic components of irreducible $\mathfrak{so}(V)$-modules. The formula is implemented on machines using the residue formula for meromorphic forms. This a priori information on the cohomology can be elegantly encoded in the  Hilbert-Poincar\'e series: 
given an irreducible representation $U$ of $G$ (thought classically: bosonic and of zero $\mathbb Z$-grading), we define the {\it multiparameter Hilbert-Poincar\'e $U$-series} of $C^\bullet(\mathfrak m, \mathfrak p)$ by
\begin{equation}
\begin{aligned}
\label{eq:sdimHPUseriesSpencerIntroduction}
P_U(C^\bullet(\mathfrak m, \mathfrak p),u,t)&:=
\!\!\!\!\!
\sum_{m,n\geq -1}
\!\!\!\!\!\operatorname{sdim}(C^\bullet(\mathfrak m, \mathfrak p)_{m,n}\otimes U^*)^G\,u^{m}\,t^{n}\;,
\end{aligned}
\end{equation}
where ``$\operatorname{sdim}$'' is the \emph{superdimension} of a vector superspace and 
parameters $t$ and $u$ are as above.
Collapsing \eqref{eq:sdimHPUseriesSpencerIntroduction} with 
$u=t^{2}$ gives a series 
$P_U(C^\bullet(\mathfrak m, \mathfrak p),u=t^2,t)=
\sum_{d\in\mathbb Z}
\operatorname{sdim}(C^{d,\bullet}(\mathfrak m, \mathfrak p)\otimes U^*)^G\,t^{d}$
in $t$ 
that we 
call {\it collapsed Hilbert-Poincar\'e $U$-series}. 

{Theorem \ref{thm:3} of \S\ref{sec2} gives a fairly general version of the Molien-Weyl formula for graded Lie superalgebras.} Here we only summarize the results in the case of the $D=11$ Poincaré superalgebra.
\begin{theorem}
\label{thm:2}
The following Table \ref{SummarizingHPMW} comprises the collapsed Hilbert-Poincar\'e $U$-series for the space of Spencer cochains of the $D=11$ Poincaré superalgebra, for various choices of irreducible representations $U$ of $G:=\Spin^\circ(V)$: \vskip0.2cm
\renewcommand{\arraystretch}{1.5}
{\small
\begin{table}[H]
\begin{centering}
\makebox[\textwidth]{%
\begin{tabular}{|c|c|c|c|}
\hline
Representation $U$ & Dynkin label of $U$ & $\dim(U)$ & Collapsed Hilbert-Poincaré $U$-series \\
\hline
\hline
 $S$ & $[0,0,0,0,1]$ & $32$ & $-t+t^3-t^5-t^7+t^9$ \\
\hline
 $(V\otimes S)_o$ & [1,0,0,0,1] & $320$ & $t^{-1 } +t^3 +t^5 -t^7 + t^9
$ 
\\
\hline
$(\Lambda ^2 V\otimes S)_o$ & [0,1,0,0,1] & 1408  & $ - t^3  + t^5 - t^7 + t^{11}$ \\
\hline
$(\Lambda ^3 V\otimes S)_o$ & [0,0,1,0,1] & 3520 & $ t^3-t^7 + t^9$ \\
\hline 
$(\Lambda ^4 V\otimes S)_o$ & [0,0,0,1,1] & 5280 & $ t^3 + t^5 - t^7 + t^9$ \\
\hline
$(\odot^2 V\otimes S)_o$ & [2,0,0,0,1] & 1760 & $ - t^7  + t^9$ \\
\hline 
$(\odot^3 S)_o$ & [0,0,0,0,3] & 4224 & $ -t^3 + t^5 - t^7 - t^9 + t^{11} $ \\
\hline
$(V\otimes\Lambda^2 V\otimes S)_o$ & $[1,1,0,0,1]$ & 10240 & $ t + t^5 - t^9 + t^{11}$\\
\hline 
 \end{tabular}}
\end{centering}
\caption[]{\label{SummarizingHPMW}Collapsed Hilbert-Poincar\'e $U$-series for $D=11$ Poincaré superalgebra.} \vskip14pt
\end{table}}
\vskip-0.4cm\par\noindent
Here the lower index $_o$ of a representation selects its irreducible component of the highest weight.
The collapsed Hilbert-Poincaré $U$-series is equal to $\sum_{d\in\mathbb Z}(-1)^d\chi(H^{d,\bullet}(\mathfrak m, \mathfrak p)\otimes U^*)^G\,t^{d}$, 
where $\chi$ is the usual Euler characteristic w.r.t. the form number at fixed $\mathbb Z$-grading $d$. In other words, it counts the Euler characteristic of the $U$-isotypic component of the Spencer cohomology.
\end{theorem}	

The irreducible representations in Table \ref{SummarizingHPMW} precisely exhaust the irreducible representations that appear in the space of Spencer cochains for the form number $q=2$ and positive odd $\mathbb Z$-grading $d$, relevant in view of Theorem \ref{thm:1}. 
For completeness, we also considered the case of the trivial representation
$U=\mathbb C$ and determined its collapsed Hilbert-Poincaré $U$-series, which is $-1 - t^6 + t^8$. Theorem \ref{thm:2} above can thus be regarded as an extension of some of the recent results on {\it Lorentz-invariant} {\it scalar} cocycles obtained in \cite{CGNR} for ``\emph{Free Differential Algebras}'' \cite{sullivan,CDF,FDAdual1, gm3,gm13, DFd11} (FDAs, or, more precisely, \emph{super semifree differential graded-commutative algebras}). The analysis of the obtained collapsed Hilbert-Poincaré series is shown to agree with the content of Theorem \ref{thm:1}.
Moreover, representatives for the cohomology classes of $H^{1,2}(\mathfrak{m},\mathfrak{p})$ obtained through the different approaches are shown to be equivalent, see Propositions \ref{MWrepresentative} and \ref{prop:cohomology1p}, and the discussion in Remark \ref{cohomologyclassequivalence}. 
\\

The final objective of the present work concerns filtered subdeformations of $\fp=\fp_{-2}\oplus\fp_{-1}\oplus\fp_{0}$.
By parity consistency, deformations of Lie superalgebras are usually understood as {\it even}, cf. \cite{Cheng-Kac}, in particular the corresponding infinitesimal deformations are cocycles that represent an {\it even} cohomology class. Odd filtered deformations have been scarcely considered in the literature, often under the simplifying assumption that they integrate to a full deformation via a $1$-dimensional space of parameters -- indeed, using a single odd parameter $\mathpzc t$,  we have $\mathpzc t^2=0$ by nilpotency and any odd infinitesimal deformation is trivially unobstructed to a first-order full deformation. 
In \S\ref{sec:deformation-cohomology} we 
generalize Fialowski's notion of deformations over general commutative algebras \cite{Fia} to embrace {\it Lie superalgebras} and {\it filtrations}. More precisely, we consider filtered deformations parametrized by finite-dimensional exterior algebras $\Lambda^\bullet W$, see Definition \eqref{def:filtereddeformation} for more details. 
(The collection of all such deformations would give rise to the entire filtered deformation functor, which should not be confused with the functor of points of the Lie superalgebra.)


We then consider maximally supersymmetric filtered subdeformations of $\fp$ whose infinitesimal deformation is odd of $\mathbb Z$-grading $1$. We first study the Jacobi identities in complete generality and then restrict to the case of  filtered subdeformations with {\it timelike nilpotent infinitesimal deformation}. The former is a genericity type assumption (if the odd infinitesimal deformation is generic, then it is timelike), while the latter is a technical assumption of cohomological nature, see Definition \ref{def:6}, Remark \ref{rem:10}.
Building on the previous sections, we finally establish in \S\ref{sec:deformation-cohomology} the following main theorem, which is a result of no-go type.
\begin{theorem}
\label{thm:1.3}
Let $F$ be a maximally supersymmetric filtered subdeformation of the $D=11$ Poincar\'e superalgebra $\fp$. If the infinitesimal odd deformation of $F$ is generic and nilpotent, then $F$ is isomorphic (as a filtered Lie superalgebra) to a first-order odd filtered subdeformation of $\fp$. The result holds for any choice of the finite-dimensional auxiliary vector space $W$.
\end{theorem}
See Theorem \ref{thm:6} for a more detailed statement. In particular, the locally homogeneous Lorentzian manifolds underlying such filtered subdeformations are flat (the subspace  $\fh_{\bar 0}\otimes\Lambda^\bullet_{\bar 0}W$ of the even part $(\fh^\Lambda)_{\bar 0}$ is closed under Lie brackets and graded, so the Riemann curvature vanishes). This result indicates that odd maximally supersymmetric filtered subdeformations of $\fp$ with non-flat underlying Lorentzian manifold, if they exist, are obtained by either relaxing the genericity assumption -- hence, by considering infinitesimal odd deformations that are in the lightlike orbit of  the projectivized action of  $G:=\Spin^\circ(V)$ on $\mathbb{P}(S)$ -- or the nilpotency type assumption (or both). This study and the associated interpretation of the resulting superalgebras $F$ in terms of supergravity backgrounds (with the additional spinor field in the spectrum of degrees of freedom) 
will be the context of a separate work. Lastly, let us notice that in the case of a dynamical gravitino, a finite-dimensional $W$ is no longer suitable. Indeed, the spacetime components of the gravitino field are supposed to be anticommuting for any position in spacetime, which forbids a finite-dimensional $\Lambda^\bullet W$ as underlying algebra. However, in the present paper we only consider  background solutions and not dynamical fields. 

\bigskip

The paper is organized as follows. In \S\ref{sec2} we streamline the necessary mathematical background on the Molien-Weyl formula including Theorem \ref{thm:3}, a fairly general version that holds for any finite-dimensional Lie superalgebra with a consistent $\mathbb Z$-grading. In \S\ref{sec6} we discuss the applications of \S\ref{sec2} to the case of the $D=11$ Poincaré superalgebra, determine the associated Hilbert-Poincaré series of Theorem \ref{thm:2} and anticipate the content of Theorem \ref{thm:1} in the component approach. We then devote \S\ref{sec:spencer-complex} to the complete proof (see Theorems \ref{thm:persoilconto}-\ref{thm:H32}) and to some preliminary results that will be useful in \S\ref{sec:deformation-cohomology}. Finally, we devote \S\ref{sec:deformation-cohomology} to the study of the maximally supersymmetric filtered subdeformations of the $D=11$ Poincaré superalgebra with non-trivial infinitesimal odd deformation and establish Theorem \ref{thm:1.3}. We give our conclusions and future developments in \S\ref{Conclusions}.

\section{Molien-Weyl formula}\label{sec2}
\subsection{Invariant supersymmetric polynomials}\label{sec2.1}
Let $G$ be a finite group and $W$ a finite-dimensional linear representation of $G$, over the field $\mathbb C$ of complex numbers. If $\mathbb C[W]:=\odot^\bullet W^*$ is the space of polynomials on $W$ and $\mathbb C[W]^G=(\odot^\bullet W^*)^G$ the space of $G$-invariant polynomials on $W$, then
the Hilbert-Poincar\'e series 
\begin{align}
P(\mathbb C[W]^G,t):=\sum_{n\geq 0} \dim(\mathbb C[W]^G_n)\,t^n
\end{align} 
is the generating function for $\mathbb C[W]^G$ endowed with its standard $\mathbb Z$-grading  $\mathbb C[W]^G=\oplus_{n \geq 0} \mathbb C[W]^G_n$ with components $\mathbb C[W]^G_n=(\odot^n W^*)^G$ (in other words, it provides with
the dimension of the space of invariant polynomials at order $n$). Note that usually there are no bounds on the powers of $t$. 
The Hilbert-Poincar\'e series can then be computed by means of the Molien-Weyl formula
\begin{align}\label{MW_A}
P(\mathbb C[W]^G,t) = \frac{1}{|G|} \sum_{g \in G} \frac{1}{\det_W(1 - t g)}\,,
\end{align}
see, e.g., \cite{DerksenKemper,procesi,Pouliot:1998yv,Weyl}. 
\vskip0.2cm\par
We now consider a representation of $G$ on a
vector superspace $W = W_{\bar 0} \oplus W_{\bar 1}$, where $W_{\bar 0}$ is the bosonic subspace and $W_{\bar 1}$ is the fermionic one, and the space $\mathbb C[W]:=\odot^\bullet W^*$ of supersymmetric polynomials on $W$, understood w.r.t. statistics. In other words, the symbol of symmetric algebra on $W^*=(W^*)_{\bar 0}\oplus (W^*)_{\bar 1}=W_{\bar 1}^*\oplus W_{\bar 0}^*$ is meant in the ``supersense'',  using the supercommutativity rules applied to the even statistics of $W_{\bar 1}^*$ and the odd statistics of $W_{\bar 0}^*$. It is $\mathbb Z$-bigraded
\begin{equation}
\label{eq:bidecomposition}
\mathbb C[W]=\bigoplus_{m,n\geq 0}\mathbb C[W]_{m,n}
\end{equation}
with graded components $\mathbb C[W]_{m,n}=\Lambda^m W_{\bar 0}^*\otimes \odot ^n W_{\bar 1}^*$, where the symmetric and exterior algebra symbols are meant in the classical sense. Note that $\mathbb C[W]_{m,n}$ is bosonic (fermionic) when $m$ is even (odd) and that $G$ preserves the decomposition \eqref{eq:bidecomposition} and thus also the statistics, since it is a classical group. We shall introduce two different parameters $u$ and $t$ to parametrize elements in the subspaces $W_{\bar 0}^*$ and $W_{\bar 1}^*$, respectively, and give the following:
\begin{definition}
\label{def:1}
The Hilbert-Poincar\'e series of the space $\mathbb C[W]^G$ of $G$-invariant supersymmetric polynomials on $W$ is defined as
\begin{equation}
\label{eq:sdimHPseriesdefinition}
\begin{aligned}
P(\mathbb C[W]^G,u,t):&=\sum_{m,n\geq 0}\operatorname{sdim}(\mathbb C[W]^G_{m,n})\,u^m\,t^n
\;,
\end{aligned}
\end{equation}
where $\operatorname{sdim}(U)=\dim(U_{\bar 0})-\dim(U_{\bar 1})$ is the superdimension of a vector superspace $U=U_{\bar 0}\oplus U_{\bar 1}$.
\end{definition}
It is immediate to see from the definition of superdimension that
\begin{equation}
\label{eq:sdimHPseries}
\begin{aligned}
P(\mathbb C[W]^G,u,t)
&=\sum_{m,n\geq 0}(-1)^m\dim(\mathbb C[W]^G_{m,n})\,u^m\,t^n
\;.
\end{aligned}
\end{equation}
Moreover, the Molien-Weyl formula extends to the identity
\begin{align}\label{MW_B}
P(\mathbb C[W]^G, u, t) = \frac{1}{|G|} \sum_{g \in G} \frac{\det_{W_{\bar 0}} (1 - u g)}{\det_{W_{\bar 1}} (1 - t g)} \,,
\end{align}
where the unusual minus sign in front of the parameter $u$ in $\det_{W_{\bar 0}} (1 - u g)$ is due to the fact that we are considering superdimensions and that $W_{\bar 0}^*$ is fermionic. Put it differently, $W_{\bar 0}^*$ is actually parametrized by $-u$ and the Molien-Weyl formula \eqref{MW_B} is obtained from the classical arguments involving dimensions followed by replacing $u$ with $-u$.\\

Interestingly, the computation of the sums in \eqref{MW_A} and \eqref{MW_B} can be restricted to the conjugacy classes of the finite group $G$. 
If we consider $G$ to be a complex linearly reductive connected group, instead of a  finite group, the sum in the Molien-Weyl formula is replaced by an integral:
\begin{align}\label{MW_C}
P(\mathbb C[W]^G, u, t) = \int_{K} \frac{\det_{W_{\bar 0}} (1 - ug)}{\det_{W_{\bar 1}} (1 - tg)} \diff \mu \,,
 \end{align}
where $K$ is the maximal compact subgroup of $G$ with corresponding normalized Haar measure $\diff \mu$ (that is, $\int _{K} \diff \mu=1$). The integral \eqref{MW_C} may be simplified further. We let $D\cong (\mathbb C^\times)^r$ be a maximal complex torus of $G$ and $T\cong (S^1)^r$ a Cartan subgroup, namely a maximal compact subgroup of $D$ with normalized Haar measure $\diff\nu$. Here and in the following $r={\rm rk}\, G$ denotes the rank of $G$. Since the integrand is invariant under conjugation, Weyl integration formula \cite{procesi,Weyl} allows to integrate only on the Cartan subgroup and rewrite the above formula as
\begin{align}\label{MW_D0}
P(\mathbb C[W]^G, u, t) = \int_{T} \frac{\det_{W_{\bar 0}} (1 - u g)}{\det_{W_{\bar 1}} (1 - tg)} \phi(g)  \diff \nu\,,
\end{align}
where the integration is performed over the Cartan subgroup $T$ and $\phi:T\rightarrow\mathbb R$ denotes the weight function of Weyl -- in other words, the integration measure $\diff \mu$ restricted to $T$ coincides with $\phi\diff\nu$. As we shall see, it is possible to express Weyl weight function in terms of the positive roots of $G$. 
The integration  can be performed by introducing a set of complex coordinates $z=(z_1,\ldots,z_r)$ over the Cartan subgroup $T\cong (S^1)^r$, with each coordinate $z_i$ defined on the $i^{th}$-copy of unit circle $S^1$. Notably, any $z\in T$ acts diagonally on $W=W_{\bar 0}\oplus W_{\bar 1}$ and $\det_{W_{\bar 0}} (1 - u g(z))=
\prod_{i=1}^{\dim W_{\bar 0}}(1-u\, m_i(z))$, $\det_{W_{\bar 1}} (1 - t g(z))=
\prod_{j=1}^{\dim W_{\bar 1}}(1-t\, n_j(z))$, for certain Laurent monomials $m_i(z)$ and $n_j(z)$ in $z_1,\ldots,z_r$, which we will soon make explicit.

From now on we assume for simplicity that $G$ is semisimple. By using Weyl character formula and the residue theorem as in for instance \cite[\S 4.6]{DerksenKemper}, one obtains 
\begin{align}\label{MW_D}
P(\mathbb C[W]^G, u, t) = \oint_{|z_1|=1}\!\!\!\!\!\!\cdots \oint_{|z_r|=1} \frac{\det_{W_{\bar 0}} (1 - u g(z))}{\det_{W_{\bar 1}} (1 - t g(z))} 
\phi(z)
\diff\nu\,,
\end{align}
\begin{equation}
\label{MW_E}
\begin{aligned}
\displaystyle\diff\nu=\prod_{i=1}^{r} 
\frac{\diff z_i}{2 \pi i z_i}\;,\quad
\phi(z)= \prod_{\alpha \in \Delta^+} \left( 1 - z^\alpha \right) =\prod_{\alpha \in \Delta^+} \left( 1 - \prod_{i=1}^{r} z_i^{n_i(\alpha)} \right) \,, 
\end{aligned}
\end{equation}
\vskip0.2cm\par\noindent
where $\Delta^+$ is the finite set of positive roots of the Lie algebra $\mathfrak g$ of $G$, and the $n_i(\alpha)\in\mathbb Z$ are the components of the root $\alpha\in\Delta^+$ expressed as linear combination of the fundamental weights of $\mathfrak g$. 

Concerning the integrand, we use the following notion: for a (finite-dimensional) classical representation $U$ of $G$, we consider its character $\chi_U:T\rightarrow \mathbb C$ as $\chi_U(z) = Tr_U\left[g(z)\right]$, $g=g(z)\in T$. It can be expressed as the sum of Laurent monomials 
\begin{align}\label{MW_FA}
\chi_U(z) = \sum_{\lambda \in \Delta_U} z^\lambda
\end{align}
where $\Delta_U$ is the set of weights of the representation $U$ counted with their multiplicity (of course, one has to take into account both the vanishing and nonvanishing weights), and $z^\lambda=\prod_{i=1}^{r} z_i^{n_i(\lambda)}$, where the $n_i(\lambda)\in\mathbb Z$ are the components of the weight $\lambda$ expressed as linear combination of the fundamental weights.
One then constructs the \emph{bosonic plethystic exponential} as follows (the subscript ``$B$'' indicates that we are dealing with the bosonic part): 
\begin{align}
\label{MW_G}
PE_B[\chi_U(z)t] :&= Exp \left[\sum_{\lambda \in \Delta_U}  \sum_{n=1}^\infty 
\frac 1n t^n z^{n\lambda} \right]= 
\frac{1}{\prod_{\lambda \in \Delta_U} (1 - t z^\lambda)} = \frac{1}{\det_{U} (1 - t g(z))}\,,
\end{align}
corresponding to the denominator of \eqref{MW_D} (except the measure $\displaystyle\diff\nu$ and Weyl weight function $\phi(z)$). This formula is suitable for a conventional representations. The fermionic version has to be modified to the following \emph{fermionic plethystic exponential}  (see for example \cite{Feng:2007ur, Hanany:2008sb}):
\begin{equation}
\begin{aligned}
\label{MW_J}
{PE}_F[\chi_U(z) u]:&= Exp \left[- \sum_{\lambda\in \Delta_U}  
\sum_{n=1}^\infty \frac 1n u^n z^{n\lambda} \right]= 
\prod_{\lambda \in \Delta_U} (1 - u z^{\lambda})= \det\!{}_{U} (1 - u g(z)) \,,
\end{aligned}
\end{equation}
where a minus sign in the definition of the bosonic plethystic exponential \eqref{MW_G} has been inserted. This corresponds to the term in the numerator of \eqref{MW_D}.
\vskip0.2cm\par
In a more general setting, one can consider a decomposable representation $W = \oplus_{K=1}^{M+N} W_K$ into sum
of  representations of definite statistics, 
and thus use different parameters $u=(u_1,\ldots,u_M)$ and $t=(t_1, \ldots, t_N)$ to parametrize them. Definition \ref{def:1} extends immediately as follows.
\begin{definition}
The Hilbert-Poincar\'e series of the space $\mathbb C[W]^G$ of $G$-invariant supersymmetric polynomials on $W$ is the series in the parameters $u=(u_1,\ldots,u_M)$ and $t=(t_1, \ldots, t_N)$ given by
\vskip0.2cm\par\noindent
\begin{align}
\label{MW_N}
P(\mathbb C[W]^G,u,t) = \!\sum_{{m=(m_I)_{I=1}^M\atop n=(n_J)_{J=1}^N}\atop m_I,n_J\geq 0}\operatorname{sdim}(\mathbb C[W]^G_{m,n})\,u_1^{m_1}\cdots u_M^{m_M}\,t_1^{n_1}\cdots t_N^{n_N}\;.
\end{align}
\end{definition}

{Of course a formula analogue to \eqref{eq:sdimHPseries} still holds in this context (but we will not write it down). The integrand of \eqref{MW_D} generalizes directly to}
\begin{align}\label{MW_K}
PE[u,t](z) = 
\prod_{I=1}^M PE\left[\chi_{W_I}(z) u_I\right]\prod_{J=1}^N PE\left[\chi_{W_{M+J}}(z) t_J\right] \,,
\end{align}
where we omitted
the $B$, $F$ subscripts. In the following, we will also omit such subscripts, since the statistics of the variables under consideration will always be understood from the context. 
Putting all together, we arrive at
\begin{proposition}
\label{prop1}
Let $G$ be a complex linearly reductive connected semisimple group of rank $r$ and $W = \oplus_{K=1}^{M+N} W_K$ a  representation of $G$ that decomposes into representations of definite statistics. Then 
$P(\mathbb C[W]^G, u,t) = 
\oint_{|z_1|=1}\!\!\!\!\!\!\cdots \oint_{|z_r|=1} 
PE[u,t](z)\phi(z)\diff\nu
$.
\end{proposition}
By knowing the powers in $u$ and $t$ and with some additional work, the explicit form of invariant polynomials can often be inferred -- we will illustrate this strategy with examples in \S\ref{sec6}. In fact, the Molien-Weyl formula  allows to restrict the number of possible invariant polynomials to check, simplifying their quest.
\subsection{Spencer complex and Euler characteristic}
\label{sec:2.2}
Ultimately we want to consider supersymmetric polynomials on $\mathfrak m$ with values in the Poincaré superalgebra $\mathfrak p$, so certain modifications to \S\ref{sec2.1} are required:
our supersymmetric polynomials are not scalar-valued and, secondly, we are not restricting our analysis to the Lorentz-invariant ones. 

To this aim, it convenient to first recall \cite[Remark 3.4.3, pag. 85]{DerksenKemper}, which does not invoke any superstructure at all and is recasted here in a form suitable for our purposes:
\begin{proposition}
\label{prop2}
Let $G$ be a complex linearly reductive connected semisimple group of rank $r$ and $W$ and $U$ two  representations of $G$. Then the Hilbert-Poincar\'e series of the space $(\mathbb C[W]\otimes U)^G$ of $G$-invariant symmetric polynomials on $W$ with values in $U$ is given by
\vskip0.3cm\par\noindent
\begin{equation}
\label{MW_MvaluesinW}
\begin{aligned}
P( (\mathbb C[W]\otimes U)^G, t)& = \oint_{|z_1|=1}\!\!\!\!\!\!\cdots \oint_{|z_r|=1} 
Tr_{U}\left[g^{-1}(z)\right]\frac{1}{\det_{W} (1 - t g(z))}\phi(z)\diff\nu\\
&=\oint_{|z_1|=1}\!\!\!\!\!\!\cdots \oint_{|z_r|=1} 
\chi_U(z^{-1})PE[t](z)\phi(z)\diff\nu\;.
\end{aligned}
\end{equation}
\end{proposition}
A number of observations are in order. First, if the representation $U$ of $G$ is self-dual (which is always the case if $G$ is a classical group of type $B_r$, $C_r$ or an exceptional group but type $E_6$), then the character contribution 
$\chi_U(z^{-1})$ equals $\chi_U(z)$. Secondly, the vector space $U$ in Prop. \ref{prop2} is trivially parametrized, namely it carries no dependence w.r.t. $t$: of course, this can be changed at will, but provided the range of the powers of the parameter in the definition of the Hilbert-Poincaré series is adjusted accordingly. Finally, we make a simple but relevant observation for $U$ irreducible: $P((\mathbb C[W]\otimes U)^G, t)$ measures the dimension of the space of $G$-invariant polynomials on $W$ with values in $U$ and this is, at the same time, the multiplicity of the representation $U^*$ in $\mathbb C[W]$.
\vskip0.2cm\par
The general setting of Theorem \ref{thm:3} discussed later on is as follows: $\mathfrak p=\mathfrak p_{\bar 0}\oplus\mathfrak p_{\bar 1}$ is a finite-dimensional Lie superalgebra endowed with a $\mathbb Z$-grading $\mathfrak p=\oplus_{d\in\mathbb Z}\mathfrak p_d$ that is consistent, namely $\mathfrak p_{\bar 0}=\oplus_{d\in\mathbb Z}\mathfrak p_{2d}$ and $\mathfrak p_{\bar 1}=\oplus_{d\in\mathbb Z}\mathfrak p_{2d+1}$, $\mathfrak m=\oplus_{d<0}\mathfrak p_d$ is the negatively-graded part of $\mathfrak p$, $W=W_{\bar 0}\oplus W_{\bar 1}$ is a finite-dimensional representation of $\mathfrak m$ endowed with a consistent $\mathbb Z$-grading, i.e., $W=\oplus_{d\in\mathbb Z}W_d$ with $W_{\bar 0}=\oplus_{d\in\mathbb Z}W_{2d}$, $W_{\bar 1}=\oplus_{d\in\mathbb Z}W_{2d+1}$, and $\mathfrak p_{d'}\cdot W_{d''}\subset W_{d'+d''}$ for all $d',d''\in\mathbb Z$, $d'<0$. 
For all $d>0$, we parametrize elements with $\mathbb Z$-grading $2d$ by the parameter $u_d$ 
and elements with $\mathbb Z$-grading $2d-1$ by the parameter $t_d$. We remark that this prescription covers all the
subspaces $\mathfrak p_{-2d}^*$, $W_{2d}$, and, respectively, the subspaces $\mathfrak p_{-2d+1}^*$, $W_{2d-1}$; in particular it fully parametrizes $\mathfrak m^*$. In order to fully parametrize also $W$, we parametrize elements in $W_0$ by $u_0:=1$ and, for all $d> 0$, elements in $W_{-2d}$ and $W_{-2d+1}$ by $\tfrac{1}{u_{d}}$ and, respectively, $\tfrac{1}{t_{d}}$. 
We let $M$ be the number of parameters in $u=(u_1,\ldots,u_M)$ and $N$ the number of parameters in $t=(t_1,\ldots,t_N)$. Finally, we denote by $G$ the connected and simply connected Lie group with Lie algebra $\mathfrak g:=\mathfrak p_0$.

Now, consider the space 
$\mathbb C[\mathfrak m]:=\odot^\bullet \mathfrak m^*=\Lambda^\bullet \mathfrak m_{\bar 0}^*\otimes \odot ^\bullet \mathfrak m_{\bar 1}^*$
of supersymmetric polynomials on $\mathfrak m$ with the natural $\mathbb Z$-multigrading
$
\mathbb C[\mathfrak m]=\bigoplus^{m=(m_I)_{I=1}^M}_{{n=(n_J)_{J=1}^N}\atop m_I,n_J\geq 0}\mathbb C[\mathfrak m]_{m,n}
$
and the space $C^\bullet(\mathfrak m, W):=\mathbb C[\mathfrak m]\otimes W$ of Spencer cochains 
with the $\mathbb Z$-multigrading
$C^\bullet(\mathfrak m, W)=\bigoplus^{m=(m_I)_{I=1}^M\atop n=(n_J)_{J=1}^N}_{m_I,n_J\geq -1} C^\bullet(\mathfrak m, W)_{m,n}$.
(Note that here the indices start from $-1$, because $W$ is allowed to have non-trivial negative components.)
We fix an auxiliary irreducible representation $U$ of $G$ thought classically, i.e., bosonic and of zero $\mathbb Z$-grading, and define
the multiparameter Hilbert-Poincar\'e $U$-series as follows.
\begin{definition}
\label{def:3}
The Hilbert-Poincar\'e $U$-series of the space $C^\bullet(\mathfrak m, W)$
is
\vskip0.3cm\par\noindent
\begin{equation}
\begin{aligned}
\label{eq:sdimHPUseriesSpencer}
P_U(C^\bullet(\mathfrak m, W),u,t):&=P((C^\bullet(\mathfrak m, W)\otimes U^*)^G,u,t)\\
&=\!\!\!\!\!
\sum_{{m=(m_I)_{I=1}^M\atop n=(n_J)_{J=1}^N}\atop m_I,n_J\geq -1}
\!\!\!\!\!\operatorname{sdim}(C^\bullet(\mathfrak m, W)_{m,n}\otimes U^*)^G\,u_1^{m_1}\cdots u_M^{m_M}\,t_1^{n_1}\cdots t_N^{n_N}\;,
\end{aligned}
\end{equation}
and collapsing it
with the relations $u_I=t^{2I}$ and $t_J=t^{2J-1}$, gives the collapsed Hilbert-Poincar\'e $U$-series $P_U(C^\bullet(\mathfrak m, W),u=t^2,t)=
\sum_{d\in\mathbb Z}
\operatorname{sdim}(C^{d,\bullet}(\mathfrak m, W)\otimes U^*)^G\,t^{d}$ in $t$.
\end{definition}
\begin{remark}
\label{rem:1}{\rm
Choosing $U=\mathbb C$ to be the trivial representation, one gets the Hilbert-Poincar\'e series $P_{\mathbb C}(C^\bullet(\mathfrak m, W),u,t)=P(C^\bullet(\mathfrak m, W)^G,u,t)$
of the space of $G$-invariant Spencer cochains.
}
\end{remark}
\begin{remark}
\label{rem:2}
{\rm
We note that $C^\bullet(\mathfrak m, W)_{m,n}$ and also $C^{d,\bullet}(\mathfrak m, W)$ are not of a definite statistics in general, i.e., each splits into the non-trivial direct sum of its bosonic and fermionic components. The series thus counts the ``supermultiplicity'' of the representation $U$ in the space of Spencer cochains: bosonic (fermionic) representations contribute positively (negatively) to the multiplicity and there might be cancellations. The straightforward analogue to formula \eqref{eq:sdimHPseries} does not hold in this context and it will be replaced by part $2$ of Theorem \ref{thm:3} below.
}
\end{remark}
We only need to introduce one last notion.
\begin{definition}
\label{def:weightedcharacter}
The weighted character of the representation $W=W_{\bar 0}\oplus W_{\bar 1}$ of $G$ with consistent $\mathbb Z$-grading $W=\oplus_{d\in\mathbb Z}W_d$ is the formal series $F[u,t]:T\rightarrow\mathbb C$ in $u=(u_1,\ldots,u_M)$ and $t=(t_1,\ldots,t_N)$ given by
\begin{equation}
\label{eq:weightedcharcter}
F[u,t]=\sum_{d>0}u_d\,\chi_{W_{2d}}-\sum_{d>0}t_d\,\chi_{W_{2d-1}}+\chi_{W_0}-\sum_{d>0}\tfrac{1}{t_d}\,\chi_{W_{-2d+1}}+\sum_{d>0}\tfrac{1}{u_d}\,\chi_{W_{-2d}}\,,
\end{equation}
where $\chi_{U}:T\rightarrow \mathbb C$ is the character of the (underlying classical) representation $U$.
\end{definition}
Note the additional minus sign for characters of fermionic representations $W_{2d-1}$ and $W_{-2d+1}$, which is due to the fact that we are dealing with superdimensions.
\begin{theorem}
\label{thm:3}
Let $\mathfrak p=\oplus_{d\in\mathbb Z}\mathfrak p_d$ be a Lie superalgebra with a consistent $\mathbb Z$-grading
and $G$ the connected and simply connected Lie group with Lie algebra $\mathfrak g=\mathfrak p_0$. Let
$W=\oplus_{d\in\mathbb Z}W_d$ be a representation of the negatively-graded part $\mathfrak m=\oplus_{d<0}\mathfrak p_d$  of $\mathfrak p$ endowed with a consistent $\mathbb Z$-grading and $C^\bullet(\mathfrak m, W)$ the space  of Spencer cochains. If $G$ is semisimple, then, for any irreducible classical representation $U$ of $G$ (i.e., bosonic and with zero $\mathbb Z$-grading), we have:
\begin{enumerate}
	\item The Hilbert-Poincar\'e $U$-series of $C^\bullet(\mathfrak m, W)$ can be computed via the Molien-Weyl formula
\begin{equation}
\label{eq:MWthm}
P_U(C^\bullet(\mathfrak m, W),u,t)=\oint_{|z_1|=1}\!\!\!\!\!\!\cdots \oint_{|z_r|=1} 
\chi_U(z)\, F[u,t](z^{-1})\, PE[u,t](z)\, \phi(z)\,\diff\nu\;,
\end{equation}
where the measure $\diff\nu$ and the Weyl weight function $\phi$ are as in \eqref{MW_E}, 
$\chi_U:T\rightarrow \mathbb C$ is the character of $U$,  $F[u,t]:T\rightarrow\mathbb C$ the weighted character of $W$ as in Definition \ref{def:weightedcharacter} and $PE(u,t)$ the plethystic exponential of $W$ as in \S \ref{sec2.1};
	\item The Hilbert-Poincar\'e $U$-series 
    collapsed at $u=t^2$ coincides with
\vskip0.2cm\par\noindent
\begin{align*}
P_U(C^\bullet(\mathfrak m, W),u=t^2,t)&=\sum_{d\in\mathbb Z}(-1)^d\chi(C^{d,\bullet}(\mathfrak m, W)\otimes U^*)^G\,t^{d}\\
&=\sum_{d\in\mathbb Z}(-1)^d\chi(H^{d,\bullet}(\mathfrak m, W)\otimes U^*)^G\,t^{d}
\;,
\end{align*}
where $\chi$ is the usual Euler characteristic w.r.t. the form number, and it can be computed using the Molien-Weyl formula \eqref{eq:MWthm} with $u=t^2$.
\end{enumerate}
\end{theorem}
\begin{remark}
{\rm The Euler characteristics computed in part $2$ of Theorem \ref{thm:3} is well-defined, since it is at a fixed $\mathbb Z$-grading: the vector spaces involved are finite-dimensional and in finite number.}
\end{remark}
\begin{proof}
The first claim is  the combination of the results of Proposition \ref{prop1} and Proposition \ref{prop2}, together with the simple observation on characters that $\chi_{U^*}(z^{-1})=\chi_U(z)$.

For the second claim, 
the definition of collapsed Hilbert-Poincaré $U$-series reads as
\begin{equation*}
\label{eq:sdimHPseriesII} 
\begin{aligned}
P_U(C^\bullet(\mathfrak m, W),u=t^2,t)&=
\sum_{d\in\mathbb Z}
\operatorname{sdim}(C^{d,\bullet}(\mathfrak m, W)\otimes U^*)^G\,t^{d}\\
&=\sum_{d\in\mathbb Z}\Big(\operatorname{dim}(C^{d,\bullet}_{\bar 0}(\mathfrak m, W)\otimes U^*)^G-\operatorname{dim}(C^{d,\bullet}_{\bar 1}(\mathfrak m, W)\otimes U^*)^G\,\Big)\,t^{d}\\
&=\sum_{d\in\mathbb Z}\sum_{q\geq 0}(-1)^{q+d}\operatorname{dim}(C^{d,q}(\mathfrak m, W)\otimes U^*)^G\,t^{d}\\
&=\sum_{d\in\mathbb Z}(-1)^d\chi(C^{d,\bullet}(\mathfrak m, W)\otimes U^*)^G\,t^{d}
\;,
\end{aligned}
\end{equation*}
where we used that the statistics is the parity of the sum of the form number $q$ and $\mathbb Z$-grading $d$. 

Since $\chi(C^{d,\bullet}(\mathfrak m, W)\otimes U^*)^G$ is nothing but the Euler characteristic of the $U$-isotypic component of the Spencer complex
at a fixed $\mathbb Z$-grading, it coincides with the Euler characteristic of its cohomology, by $G$-equivariance, complete reducibility, and 
the standard telescopic argument. 
\end{proof}

\section{Hilbert-Poincaré series of the $D=11$ Poincaré superalgebra}\label{sec6}

In this section, we apply the Molien-Weyl formula to $D=11$ flat superspace $M=P/G$. Let us first introduce the necessary ingredients and then focus on the analysis of the emerging results.

\subsection{The $D=11$ ingredients}
\label{subsec:MWD=11ingredients}
The flat superspace $M=P/G$ can be dually described in terms of the supervielbein $\mathpzc{E}^I=\{e^a,\psi^\alpha\}$, satisfying the Maurer-Cartan equations\footnote{Here, $a,b,\ldots=0,1,\ldots,10$ are vector indices, $\alpha=1,\ldots,32$ spinorial indices, and $\psi^\alpha$ is a Majorana gravitino. For simplicity, in the following we will frequently omit writing the spinorial index.}
\begin{eqnarray}\label{MC11d}
{\diff} e^a =  {\frac{i}{2}} \bar{\psi} \Gamma^a \psi \,,
~~~~~
{\diff} \psi^\alpha =0 \,,
\end{eqnarray}
where ${\diff}$ is the usual exterior differential, the bar symbol is the symplectic duality (or Dirac conjugation) on the $D=11$ Majorana spinor representation $S$ and the capital Greek letter $\Gamma$ denotes the $D=11$ Dirac matrices\footnote{The convention is that $\{\Gamma^a,\Gamma^b\}=2\eta^{ab}$, where $\eta$ is the flat metric on $V\cong \mathbb R^{1,10}$ with mostly minus signature.}. (Recall that our analysis in \S\ref{sec2} works at the complexified level, so there is no loss in generality here.) The symplectic duality is explicitly given by $\bar\psi=\psi^T C$, with ${}^T$ the transpose and $C$ the $D=11$ charge conjugation matrix, which satisfies
$\Gamma_a^T=-C\Gamma_a C^{-1}$, $C^T=C^{-1}=-C$.
The Spencer differential also acts on the generators $\mathpzc X_{\tilde I}=\{ L_{ab},X_a, q_\alpha\}$ of $\fp$  via
\begin{eqnarray}
\label{elevA}
{\diff} X_a =0\,, ~~~~~~
\diff q_\alpha=\frac i2 (\bar\psi\Gamma^a)_\alpha X_a\,, ~~~~~~
{\diff} L_{ab} = X_{[a} e_{b]} - {\frac{1}{2}} \bar q\,\Gamma_{ab} \psi \,,
\end{eqnarray}
which complement the Maurer-Cartan equations \eqref{MC11d}.

The semisimple connected group $G$ is of type $B_5$, with Cartan subgroup $T\cong (S^1)^5$ and all representations completely reducible and self-dual. The key ingredients needed for the Molien-Weyl formula are the characters $\chi_{V}:T\rightarrow\mathbb C$, $\chi_{S}:T\rightarrow\mathbb C$ and $\chi_{\mathfrak{so}(V)}:T\rightarrow\mathbb C$ of the representations $V$, $S$ and the adjoint representation $\fso(V)$, respectively. Their explicit expressions \eqref{11dAa}-\eqref{11dAc}, together with the plethystic exponentials and Haar measure $\diff \mu|_{T}$, can be found in Appendix \ref{appchar}.

Since we will consider cochains with values in $\mathfrak p$, we have to consider the weighted character
\begin{eqnarray}
\label{elevC}
F[u,t](z) = \frac1u \chi_{V}(z) - \frac{1}{t} \chi_{S}(z) 
+ \chi_{\mathfrak{so}(V)}(z)\;,
\end{eqnarray}
where we assign parameters
$[e^a] = {u}$, $[\psi^\alpha]=t$, and $[X_a]={u^{-1}}$, $[q_\alpha]=t^{-1}$, $[L_{ab}]=1$ as in \S \ref{sec:2.2}. 
The multiparameter Hilbert-Poincar\'e series for the space of $G$-invariant Spencer cochains is defined as in Definition \ref{def:3} and Remark \ref{rem:1} and it can be computed via the Molien-Weyl formula 
\begin{eqnarray}
\label{elevB}
P(C^{\bullet}(\mathfrak{m},\mathfrak p)^G,u,t) = \oint_{|z_1|=1}\!\!\!\!\!\!\cdots \oint_{|z_5|=1}  F[u,t] 
PE[\chi_V u] PE[\chi_S t] \diff\mu|_T \,,
\end{eqnarray}
due to part $1$ of Theorem \ref{thm:3} and the fact that all the representations of the group $G$ are self-dual. 

\subsection{A warm up: the trivial representation}
Before focusing on Spencer cochains in non-trivial irreducible representations, let us set the stage by discussing the Hilbert-Poincar\'e series for Lorentz-invariant cochains.
\begin{proposition}
The explicit expression of the Hilbert-Poincar\'e series of the space of  $G$-invariant Spencer cochains of the $D=11$ Poincaré superalgebra is given by 
\begin{eqnarray}\label{elevD}
P(C^{\bullet}(\mathfrak{m},\mathfrak p)^G,u,t) &=&
\frac{1}{\left(1- t^4\right) u}
{(u-1) \left[-t^8 u^7+t^6 \left(-\left(u^9-u^8+u^6\right)\right) \right.}
\nonumber \\
&+&t^4 \left(u^{10}+u^9-u^8+u^7-u^6+2 u^5-2 u^4+u\right)
\nonumber \\
&+&t^2 \left(-2 u^{10}+u^9-u^8+u^7-3 u^6+u^5-u^4+u^3-2 u^2-1\right)
\nonumber \\
&+&{\left. u \left(u^{10}+u^8+u^6+u^4+u^2+u+2\right)\right]} \,.
\end{eqnarray}
The collapsed  Hilbert-Poincar\'e series is thus given by $P(C^{\bullet}(\mathfrak{m},\mathfrak p)^G,u=t^2,t)=-1 - t^6 + t^8 $, with the following non-trivial Euler characteristics of the Spencer cohomology:
\begin{align}
\label{eq:ECinvariant}
\chi(H^{0,\bullet}(\mathfrak m, \mathfrak p))^G&=-1\;,\qquad 
\chi(H^{6,\bullet}(\mathfrak m, \mathfrak p))^G=-1\;, \qquad 
\chi(H^{8,\bullet}(\mathfrak m, \mathfrak p))^G=+1\;.
\end{align}
\end{proposition}
\begin{proof}
The first two claims follow from the computation of \eqref{elevB} using the residue formula and setting $u=t^2$. The last claim follows from the fact that the collapsed Hilbert-Poincaré series coincide with $\sum_{d\in\mathbb Z}(-1)^d\chi(H^{d,\bullet}(\mathfrak m, \mathfrak p))^G\,t^{d}$, which was
established in Theorem \ref{thm:3}.
\end{proof}
Equation \eqref{elevD} describes the possible Lorentz-invariant objects constructed in terms of the ingredients discussed above, i.e., the Lorentz-invariant Spencer cochains on $\mathfrak m$ with values in $\mathfrak p$. 
For instance, the factor $1/(1-t^4)$ stands for the powers of the invariant cochain $\bar\psi\Gamma_{abcde} \psi
\bar\psi\Gamma^{abcde} \psi$ (which is a scalar commuting cochain, and therefore it can appear with any power).

The non-trivial Euler characteristics \eqref{eq:ECinvariant} immediately enforce the existence of at least three non-trivial Lorentz-invariant cohomology classes -- since $G=\Spin^\circ(V)$ acts completely reducibly, one can always choose a Lorentz-invariant representative for each of the Lorentz-invariant classes. The interpretation is the following:
\begin{align}
\label{elevE}
-1 : \qquad & \omega^{(1)} = X_a e^a - \bar q \psi \,, \nonumber \\
- t^6 : \qquad & \omega^{(5)} = \left(X_a e^a - \bar q \psi\right) \wedge  
(\bar\psi \Gamma_{cd} \psi e^c e^d) \,, \nonumber \\
+ t^8 : \qquad & \omega^{(6)} = {i} \bar\psi \Gamma_{abcde} \psi e^a e^b e^c e^d X^e {+ \frac{3}{2}}
(\bar\psi \Gamma_{ab} \psi e^a e^b) \wedge  (\bar \psi \Gamma_{cd} \psi L^{cd} {- 10 i \, \bar{q}\Gamma_c \psi e^c} ) \,.
\end{align}
Then, we see that 
\begin{align}
\label{evelE}
 \omega^{(5)}  &= {2} \omega^{(1)} \wedge \omega^{(4)}_2 \,,\\
  \omega^{(6)} &= {i} \bar\psi \Gamma_{abcde} \psi e^a e^b e^c e^d X^e {+ 3}
\omega^{(4)}_2 \wedge  (\bar \psi \Gamma_{cd} \psi L^{cd} {- 10 i \, \bar{q}\Gamma_c \psi e^c}) \,,
\end{align}
where the Lorentz-invariant scalar 4-form $\omega^{(4)}_2= {\frac{1}{2}} \bar\psi \Gamma_{ab} \psi e^a e^b$ had been already discovered in \cite{CGNR} and  its closedness therein established.
It is immediate to check that also $\omega^{(1)}$ is closed and that 
\begin{align}
\label{evelF}
\diff \omega^{(6)} = & {-2}  \bar\psi \Gamma_{abcde} \psi \bar\psi \Gamma^a \psi e^b e^c e^d X^e \nonumber \\
& {+ 3}
\omega^{(4)}_2 \wedge \left[  \bar \psi \Gamma_{cd} \psi \left( X^c e^d - {\frac{1}{2}} \bar q \Gamma^{cd} \psi \right) {- 5  \left( \bar{\psi} \Gamma_{cd} \psi X^c e^d + \bar{q} \Gamma_c \psi \bar{\psi} \Gamma^c \psi \right) }\right] = 0 \,,
\end{align}
thanks to the Fierz identity
\begin{equation}
    A \, \Gamma_a\psi \bar\psi \Gamma^a\psi + B \, \Gamma_{ab}\psi \bar\psi \Gamma^{ab}\psi + C \, \Gamma_{abcde} \psi  \bar\psi \Gamma^{abcde} \psi =0 \,, \qquad A - 10 \, B - 6 \cdot 5!\, C = 0 \,.
\end{equation}
 \subsection{Analysis of the Hilbert-Poincaré series}\label{sec3.3}

We discuss the Hilbert-Poincar\'e series for different irreducible representations of the simple Lie algebra of type $B_5$. By defining a formal linear combination
$\Lambda=\sum_{[a_1,\ldots,a_5]\in\mathbb N^5} 
\lambda_{(a_1, \dots, a_5)} \chi_{[a_1,\dots,a_5]}$
of the characters $\chi_{[a_1,\ldots,a_5]}:T\rightarrow\mathbb C$ with coefficients $\lambda_{(a_1, \dots, a_5)}$ either $0$ or $1$, 
one can keep track of the different representations, and 
 the Molien-Weyl formula reads
\begin{eqnarray}
\label{tenB}
P_{\Lambda}(C^{\bullet}(\mathfrak{m},\mathfrak p),u,t) = \oint_{|z_1|=1}\!\!\!\!\!\!\cdots \oint_{|z_5|=1} \Lambda(z)\, F[u,t](z)\,
PE[\chi_V u](z)\, PE[\chi_S t](z)\, \diff\mu|_T \,.
\end{eqnarray}
Since we are only interested in cochains with form number $q=2$, we can restrict the number of possible $\lambda_{(a_1, \dots, a_5)}$:
\begin{lemma}
\label{lem:1}
The irreducible modules $U$ that appear in  $C^{d,2}(\mathfrak m,\mathfrak p)$ for positive odd $\mathbb Z$-gradings $d$ are exactly those considered in Table \ref{SummarizingHPMW} of Theorem \ref{thm:2}.
\end{lemma}
We will not display the relevant characters since they are rather cumbersome, but include them into the Mathematica Notebook attached to the arXiv version of the paper. The complete computation of the Hilbert-Poincar\'e series \eqref{tenB} is then obtained by automatic evaluation on a computer, after a long analysis of the poles and residues, and the final result is Theorem \ref{thm:2}.

\begin{remark}{\rm
It is interesting to note that all coefficients in Table \ref{SummarizingHPMW} are either $0, \pm 1$ and that all the series truncate at the maximal power $t^{11}$. This truncation seems to suggest that the subcomplex of the Spencer complex that is formed by cochains that vanish when all entries are in $V$ is exact, so that the cohomology of the Spencer complex does not change by quotienting it by this subcomplex. If this were the case, then the Spencer cohomology would be isomorphic to the cohomology of the quotient complex 
$\Lambda^\bullet V^*\otimes\fp$ and vanish eventually. However, this subcomplex is not exact for some $\mathbb Z$-gradings in general and there is also a $\mathbb Z$-grading mismatch since $\deg(V)=-2$ (and not $\deg(V)=-1$). The reason behind this truncation phenomenon does not appear therefore obvious and it might deserve a more conceptual understanding.
}
\end{remark}

Let us consider the special case of $U=S$. In this case, we set all $\lambda_{(a_1, \dots, a_5)}=0$ except for $\lambda_{(0,0,0,0,1)}=1$ and get
$P_{S}(C^{\bullet}(\mathfrak{m},\mathfrak p),u=t^2,t) = 
-t+t^3-t^5-t^7+t^9$.
By part $2$ of Theorem \ref{thm:3}, we have non-trivial Euler characteristics of the $S$-isotypic component of the Spencer cohomology at various gradings:
\begin{align}
\label{eq:ECisotypicS}
\chi(H^{1,\bullet}(\mathfrak m, \mathfrak p)\otimes S^*)^G&=+1\;,\qquad
\chi(H^{3,\bullet}(\mathfrak m, \mathfrak p)\otimes S^*)^G=-1\;,\qquad
\chi(H^{5,\bullet}(\mathfrak m, \mathfrak p)\otimes S^*)^G=+1\;,\nonumber\\
\chi(H^{7,\bullet}(\mathfrak m, \mathfrak p)\otimes S^*)^G&=+1\;,\qquad
\chi(H^{9,\bullet}(\mathfrak m, \mathfrak p)\otimes S^*)^G=-1\;.
\end{align}
The first Euler characteristic in \eqref{eq:ECisotypicS} gives a contribution in $H^{1,2}(\mathfrak m,\mathfrak p)$. 
To describe this class, we introduce the following forms 
\begin{align}
\label{tenD}
\omega_1 &= \Gamma_a q \bar\psi \Gamma^a \psi\;, \qquad
\omega_2 = \Gamma_{ab} q \bar\psi \Gamma^{ab} \psi\;, \qquad
\omega_3 = \Gamma_{a_1\dots a_5} q \bar\psi \Gamma^{a_1\dots a_5} \psi\;, \nonumber\\
\omega_4 &= \psi X_a e^a\;, \, \quad \qquad
\omega_5 =  \Gamma_{ab}\psi X^a e^b\;.\nonumber
\end{align}
Among these five quantities, the closure condition will kill two of them, and out of the remaining three, two of them are exact as variations of the forms
$\chi_1 = \Gamma_a q e^a$ and
$\chi_2 = L_{ab} \Gamma^{ab} \psi$:
 \begin{align}
    \diff \chi_1&=\frac{i}{2}(\omega_1+\omega_4-\omega_5),\qquad \diff\chi_2=-\omega_5-\frac{35}{32}\omega_1+\frac{19}{64}\omega_2+\frac{1}{768}\omega_3\,.    
\end{align}
\begin{proposition}\label{MWrepresentative}
The combination
\begin{align}
\label{eq:representative-12}
    \tilde\omega= a\omega_1+b\omega_2+c\omega_3\;,\qquad b=-\frac{3a}{22},\, c=-\frac{a}{1320}\;,
\end{align}
is a non-trivial cohomology representative inside the group $H^{1,2}(\fm,\fp)$.
\end{proposition}
\begin{proof}
    Closedness is a consequence of the Fierz Identity
\begin{align}
    a (\Gamma_a\Gamma_c)^\delta{}_{(\gamma}C\Gamma^a_{\alpha\beta)}+b(\Gamma_{ab}\Gamma_c)^\delta{}_{(\gamma}C\Gamma^{ab}_{\alpha\beta)}+c(\Gamma_{a_1\ldots a_5}\Gamma_c)^\delta{}_{(\gamma}C\Gamma^{a_1\ldots a_5}_{\alpha\beta)}=0\;,
\end{align}
while the fact that $\tilde\omega$ is non-exact follows since the differential of every combination of $\chi_1$,$\chi_2$ contains one $X^a$, whereas $\omega_{1},\omega_2,\omega_3$ do not.
\end{proof}

The above discussion shows that $H^{1,2}(\mathfrak m,\mathfrak p)$ contains {\it at least} the irreducible component $S$. If the differential were generic, the only other irreducible representation to consider would be $U=(V\otimes\Lambda^2 V\otimes S)_o$. However, a simple check reveals that such a module is not present in form number $q=2$. Thus, we may anticipate
$
    H^{1,2}(\mathfrak m,\mathfrak p) \cong S
$.
As emphasised above, this result holds if the differential were generic: this is a delicate argument that will be confirmed in Theorem \ref{thm:persoilconto} via representation-theoretic techniques.
\vskip0.1cm\par

A similar discussion can be performed for the group $H^{3,2}(\mathfrak m,\mathfrak p)$. However, in order to do so, we need an expression of the Hilbert-Poincar\'e series which is more transparent w.r.t. the form number. To 
retrieve information about the latter, we may add an auxiliary parameter $p$ (in physical terms, called a ``fugacity'') to the plethystic exponentials
$PE[\chi_V u]\rightarrow PE[\chi_V u p ]$, 
$PE[\chi_S t] \rightarrow PE[\chi_S t p]$,
so that the Molien-Weyl formula reads 
\begin{eqnarray}
\label{tenBfugacity}
P_{U}(C^{\bullet}(\mathfrak{m},\mathfrak p),u,t,p) = \oint_{|z_1|=1}\!\!\!\!\!\!\cdots \oint_{|z_5|=1} \chi_U(z)\, F[u,t](z)\,
PE[\chi_V up](z)\, PE[\chi_S tp](z)\, \diff\mu|_T \,.
\end{eqnarray}
We may then evaluate \eqref{tenBfugacity} at $u=t^2$ and select the contributions of the collapsed Hilbert-Poincaré $U$-series at a fixed $\mathbb Z$-grading. In our case of interest according to Lemma \ref{lem:1}, we set
\begin{equation}
\begin{aligned}
\Lambda&=\lambda_{(0,0,0,0,1)}\chi_{[0,0,0,0,1]}+\lambda_{(1,0,0,0,1)}\chi_{[1,0,0,0,1]}+\lambda_{(0,1,0,0,1)}\chi_{[0,1,0,0,1]}+\lambda_{(0,0,1,0,1)}\chi_{[0,0,1,0,1]}\\
&+\lambda_{(2,0,0,0,1)}\chi_{[2,0,0,0,1]}+\lambda_{(0,0,0,1,1)}\chi_{[0,0,0,1,1]}+\lambda_{(0,0,0,0,3)}\chi_{[0,0,0,0,3]}+\lambda_{(1,1,0,0,1)}\chi_{[1,1,0,0,1]}
\end{aligned}
\end{equation}
and select the contributions of the series  $P_{\Lambda}(C^{\bullet}(\mathfrak{m},\mathfrak p),u=t^2,t,p)$ that are linear and cubic in $t$. 
\begin{proposition}
\label{prop:CHPSd13}
We have

{\small{
\begin{align}
\label{tEA}
P_{\Lambda}(C^{1,\bullet}(\mathfrak{m},\mathfrak p),u=t^2,t,p)
&=\Big(2 p - 5 p^2 + 2  p^3\Big) \lambda_{(00001)}+
\Big(2p - 5 p^2 + 3  p^3\Big) \lambda_{(10001)}\nonumber\\
&+\Big(p - 3 p^2 + 2  p^3\Big)  \lambda_{(01001)}+
\Big(-p^2 + p^3\Big) \lambda_{(00101)}  
+\Big(-p^2 + p^3\Big)\lambda_{(20001)}\nonumber\\
& +
\Big(-p^2 + p^3\Big)  \lambda_{(00011)}+
\Big(-p^2 + p^3\Big) \lambda_{(00003)}+ p^3 \lambda_{(11001)}
\end{align}}}
and

{\small{
\begin{align}
\label{tEABB}
P_{\Lambda}(C^{3,\bullet}(\mathfrak{m},\mathfrak p),u=t^2,t,p)
&=\Big(-3 p^2 + 11 p^3 - 9 p^4 + 2 p^5\Big) \lambda_{(00001)} +
 \Big(-4 p^2 + 15 p^3 - 14 p^4 + 4  p^5\Big) \lambda_{(10001)} \nonumber \\
& +
 \Big(-3 p^2 + 12 p^3 - 14 p^4 + 4  p^5\Big)  \lambda_{(01001)} +
\Big(-p^2  + 8 p^3 - 9 p^4 + 3  p^5\Big) \lambda_{(00101)} \nonumber \\
& +\Big(-p^2 + 6 p^3 - 8 p^4 + 3  p^5\Big) \lambda_{(20001)} +
\Big(5 p^3 - 7 p^4 + 3  p^5\Big) \lambda_{(00011)} \nonumber \\
& +
\Big(3 p^3 - 7 p^4 + 3  p^5\Big)  \lambda_{(00003)} +
\Big(-p^2 + 5 p^3 - 8 p^4 + 4  p^5\Big) \lambda_{(11001)}\;,
\end{align}}}
respectively.
 \end{proposition}
\begin{remark}{\rm
Of course, setting $p=1$, we obtain the coefficients of the Hilbert-Poincaré $U$-series at power $t$ and $t^3$:
\begin{align*}
\label{tEAB}
P_{\Lambda}(C^{1,\bullet}(\mathfrak{m},\mathfrak p),u=t^2,t)
&= - \lambda_{(00001)}+
\lambda_{(11001)}\;,\\
P_{\Lambda}(C^{3,\bullet}(\mathfrak{m},\mathfrak p),u=t^2,t )
&=
+\lambda_{(00001)}
+\lambda_{(10001)}
-\lambda_{(01001)}
+\lambda_{(00101)} 
+\lambda_{(00011)} 
-\lambda_{(00003)}\;,
\end{align*}
which is in agreement with Theorem \ref{thm:2}.
}
\end{remark}

The only Spencer $2$-cochains at degree $t^3$ schematically read $e eq$ and $e\psi L$ and are both odd and both separately not closed. For this reason, they can only belong to the irreducible representations $(\Lambda ^2 V\otimes S)_o$ and $(\odot^3 S)_o$, as shown in Table \ref{SummarizingHPMW}. From \eqref{tEABB}, we immediately see that the Spencer cochains in the latter representation only appear starting from form number three and can then be excluded.
Therefore, only $(\Lambda ^2 V\otimes S)_o$ remains. 
Since the differential of the second cochain is the only one containing the structure $\psi \psi \psi L$ and since the latter can be shown to be non-vanishing in this representation,
we anticipate that there does not exist a closed combination of two-form Spencer cochains in any irreducible representation. This result suggests that
$H^{3,2}(\mathfrak m,\mathfrak p) = 0
$
and it will be rigorously proven in Theorem \ref{thm:H32}. 

\section{Fermionic Spencer cohomology of maximally supersymmetric subalgebras of the $D=11$ Poincaré superalgebra}
\label{sec:spencer-complex}
The deformations of algebraic structures, such as Lie superalgebras, are
typically governed by some cohomology theory.  For Lie
superalgebras, it is the cohomology of the Chevalley--Eilenberg
complex of the Lie superalgebra with coefficients in the adjoint module
\cite{CE,Lei,Bin}.  In the case of a graded Lie
superalgebra (such as the Poincaré superalgebra
\eqref{eq:Z-grading}), the Chevalley--Eilenberg differential has zero degree
and the complex splits in the direct sum of sub-complexes labelled
by the degree.  In studying filtered deformations of graded Lie
superalgebras, we are interested in deforming the Lie bracket by terms
of positive degree. Furthermore, for graded Lie superalgebras, we may often pass to the Spencer complex corresponding to the negative part and a first step in this deformation process is the calculation of the cohomology of this complex.

In \S\ref{sec:deformation-cohomology}, we seek
for filtered deformations of  maximally supersymmetric graded Lie subalgebras $\mathfrak h$ of $\mathfrak p$.
For this, we pin down in \S\ref{sec:spencer-complex} the relevant cohomology groups:
in \S\ref{sec:poinc-super} we collect first
results on the Spencer cohomology of $\mathfrak h$, and then dedicate \S\ref{sec:proofthm1} to the proof of the main Theorem 
\ref{thm:1}. The proof splits into three subsections: a preliminary section, a section focusing on normalization conditions for  degree $1$ cocycles and the group $H^{1,2}(\fm,\fp)$ and, finally, a section on $H^{3,2}(\fm,\fp)$.

\subsection{The deformation complex}
\label{sec:poinc-super}
Our aim is to consider $\mathbb Z$-graded subalgebras
$\fh=\fh_{-2}\oplus\fh_{-1}\oplus\fh_0$ of the Poincaré superalgebra $\fp=\fp_{-2}\oplus\fp_{-1}\oplus\fp_0$
that are maximally supersymmetric, namely satisfying $\fh_{-1}=\fp_{-1}=S$. Because the supertranslation ideal $\mathfrak m$ of $\mathfrak p$ is generated by $S$ (the Dirac current is a surjective map,
since $\fp_{-2}=V$ is an irreducible $\mathfrak{so}(V)$-module), such subalgebras
in fact differ from $\mathfrak p$ only in zero degree, that is, $\fh \subset \fp$ with
$\fh_0 \subset \fp_0$ and $\fh_j = \fp_j$ for $j<0$.

The cochains of the Spencer complex of $\fh$ are linear maps
$\Lambda^p \mathfrak m \to \fh$, where $\Lambda^\bullet\fm$ is meant here
in the super sense using the $\mathbb Z_2$-grading. 
One extends the degree in $\fh$ to such cochains as usual and,
since the $\ZZ$- and
$\ZZ_2$ gradings are compatible, even (odd) cochains have even
(odd) degree. The $p$-cochains
of highest degree are the maps $\Lambda^p V \to \fh_0$, which have
degree $2p$, while the $p$-cochains of lowest degree are those in
$\odot^p S\to V$, which have degree
$p-2$, and then those in $\odot^p S\to S$ and
$\odot^{p-1} S \otimes V\to V$, which
have degree $p-1$. The Spencer differential
$\partial: C^{d,p}(\fm,\fh) \to
C^{d,p+1}(\fm,\fh)$ 
has zero
degree, so the complex breaks up in the direct of sum of
\emph{finite} complexes for each degree.
The spaces in the complexes for small degree are
in Table~\ref{tab:even-cochains-small}; 
we shall be interested in $p=2$ in the remaining of the paper, which corresponds to
infinitesimal  deformations.

\begin{table}[h!] 
  \centering
 \begin{tabular}{c*{7}{|>{$}c<{$}}}
    \multicolumn{1}{c}{} & \multicolumn{6}{|c}{$p$} \\\hline
    deg & 0 & 1 & 2 & 3 & 4 & 5 & 6\\\hline
    0 & \fh_0 & \begin{tabular}{@{}>{$}c<{$}@{}} S \to S\\ V \to
                    V\end{tabular} & \odot^2 S \to V & & & & \\\hline
										    1 & & \begin{tabular}{@{}>{$}c<{$}@{}} S \to \fh_0\\ V \to
                    S\end{tabular} & \begin{tabular}{@{}>{$}c<{$}@{}}
                         \odot^2 S \to S     \\
                          S \otimes V \to V  \end{tabular}&\odot^3S\to V & & &\\\hline
    2 & & V \to \fh_0 & \begin{tabular}{@{}>{$}c<{$}@{}}
                     \odot^2 S \to \fh_0       \\ S \otimes V \to S \\
                          \Lambda^2 V \to V   \end{tabular}
         & \begin{tabular}{@{}>{$}c<{$}@{}} \odot^3 S \to S\\
             \odot^2 S \otimes V \to V\end{tabular} & \odot^4 S \to V & &
        \\\hline
				    3 & &  &  \begin{tabular}{@{}>{$}c<{$}@{}}
                         S \otimes V \to \fh_0   \\ \Lambda^2 V \to S \end{tabular} & \begin{tabular}{@{}>{$}c<{$}@{}}\odot^3S\to \fh_0\\
   \odot^2 S \otimes V \to S \\
   \Lambda^2 V \otimes S \to V\end{tabular} & \begin{tabular}{@{}>{$}c<{$}@{}}
  \odot^4 S \to S \\
  \odot^3 S \otimes V \to V \end{tabular} & \odot^5 S\to V& \\\hline
				4 & & & \Lambda^2 V \to \fh_0 & \begin{tabular}{@{}>{$}c<{$}@{}}
   \odot^2 S \otimes V \to \fh_0 \\
   \Lambda^2 V \otimes S \to S \\ 
   \Lambda^3 V \to V \end{tabular} & \begin{tabular}{@{}>{$}c<{$}@{}}
  \odot^4 S \to \fh_0 \\
  \odot^3 S \otimes V \to S \\
  \Lambda^2 V\otimes\odot^2 S \to V \end{tabular} &   \begin{tabular}{@{}>{$}c<{$}@{}}
  \odot^5 S \to S \\
  \odot^4 S \otimes V \to V \end{tabular} & \odot^6 S\to V \\\hline
  \end{tabular}
		\caption[]{\label{tab:even-cochains-small} Even and odd $p$-cochains of small degree.}
\end{table}


We shall first relate the groups $H^{d,2}(\fm,\fh)$ to the groups $H^{d,2}(\fm,\fp)$ for small positive degrees, but, in order to do so, we have to remind certain deep results from \cite{Figueroa-OFarrill:2015rfh}. The group $H^{2,2}(\fm,\fp)$ is canonically identified with the collection of all cochains $\alpha+\beta+\gamma\in C^{2,2}(\fm,\fp)$, where $\alpha:\Lambda^2V\to V$, $\beta:V\otimes S\to S$, $\gamma:\odot^2 S\to \fso(V)$, such that $\alpha=0$ and $\partial(\beta+\gamma)=0$, cf. \cite{Figueroa-OFarrill:2015rfh}. Moreover $H^{2,2}(\fm,\fp)\cong\Lambda^4 V$ as an $\mathfrak{so}(V)$-module, with the closure condition expressing $\beta$ and $\gamma$ in terms of $\varphi\in\Lambda^4 V$: we have $\beta+\gamma=\beta^\varphi+\gamma^\varphi$, 
where
\begin{align}
\beta^{\varphi}(v,s)&=\tfrac1{24}(v\cdot\varphi-3\varphi\cdot v)\cdot s\;,\\
\gamma^{\varphi}(s,s)v&=-2\kappa(\beta^\varphi(v,s),s)\;,
\end{align}
for all $s\in S$, $v\in V$ \cite{Figueroa-OFarrill:2015rfh, Santi:2019kpx}. This yields a canonical identification $H^{2,2}(\fm,\fp)\cong\{\beta^\varphi+\gamma^\varphi\mid\varphi\in\Lambda^4V\}$.
\begin{proposition}
\label{prop:cohomology1h}
We have the following natural identifications of Spencer cohomology groups
\begin{equation*}
\label{eq:5}
\begin{aligned}
H^{1,2}(\fm,\fh)&\cong 
H^{1,2}(\fm,\fp)\oplus \frac{\partial\left\{X_S:S\to \fso(V)\right\}}{\partial\left\{X_S:S\to \fh_0\right\}}\\
\!\!\!\!\!\!\!\!\!\! H^{2,2}(\fm,\fh)&\cong \left\{\beta^\varphi + \gamma^\varphi\,|\,\varphi\in\Lambda^4 V\, \text{~with~}
\gamma^\varphi(s,s)\in\fh_0\;\text{for all}\; s\in S\right\}\\
H^{3,2}(\fm,\fh)&\cong H^{3,2}(\fm,\fp)\cap C^{3,2}(\fm,\fh)\\
H^{4,2}(\fm,\fh)&=0
\end{aligned}
\end{equation*}
and the Spencer differential $\partial$ is injective on the spaces of $1$-cochains $C^{1,1}(\fm,\fp)$ and $C^{2,1}(\fm,\fp)$.
\end{proposition}
\begin{proof}
From Table \ref{tab:even-cochains-small} we immediately see that $C^{1,2}(\fm,\fh)=C^{1,2}(\fm,\fp)$, so $Z^{1,2}(\fm,\fh)=Z^{1,2}(\fm,\fp)$. We now use that the component $\partial:C^{1,1}(\fm,\fp)\to C^{1,2}(\fm,\fp)$ of the Spencer operator  is injective. (This is a non-trivial fact, which follows from the classification of maximal prolongations of Poincar\'e superalgebras in \cite{MR3255456}: the first Cartan-Tanaka prolongation is non-zero 
in only a few cases, but in such cases the zero-degree level necessarily has to include the grading element $Z$.
Hence the first prolongation $\fp_{(1)}\cong Z^{1,1}(\fm,\fp)$ of $\fp$ is trivial.) 
Then
\begin{align*}
H^{1,2}(\fm,\fp)&\cong Z^{1,2}(\fm,\fp)/C^{1,1}(\fm,\fp)\;,\\
H^{1,2}(\fm,\fh)&\cong Z^{1,2}(\fm,\fp)/C^{1,1}(\fm,\fh)\;,
\end{align*}
and
\begin{align*}
H^{1,2}(\fm,\fh)&\cong  
H^{1,2}(\fm,\fp)\oplus C^{1,1}(\fm,\fp)/C^{1,1}(\fm,\fh)\\
&\cong H^{1,2}(\fm,\fp)\oplus(\fso(V)/\fh_0)\otimes S^*\;,
\end{align*}
proving the first claim. The injectivity of $\partial$ on $C^{2,1}(\fm,\fp)$ is proved in \cite[Lemma 2]{Figueroa-OFarrill:2015rfh}, while the identification 
\begin{align}
\label{eq:formeridentification}
\!\!\!\!\!\!\!\!\!\! H^{2,2}(\fm,\fh)&\cong \frac{\left\{\beta^\varphi + \gamma^\varphi
+ \partial X_V\,|\,\varphi\in\Lambda^4 V,\,X_V: V \to \fso(V) \text{~with~}
\gamma^\varphi(s,s)-X_V(\kappa(s,s))\in\fh_0\right\}}{\partial\left\{X_V: V \to \fh_0\right\}}
\end{align}
and the claim on $H^{4,2}(\fm,\fh)$ are proved in \cite[Prop. 8]{Figueroa-OFarrill:2015rfh}, thus we omit the details. We now decompose any element $\omega\in\odot^2 S$ into
$\omega=-\tfrac{1}{32}\big(\omega^{(1)}+\omega^{(2)}+\omega^{(5)}\big)$ according to $\odot^2 S\cong\Lambda^1 V\oplus \Lambda^2 V\oplus\Lambda^5 V$,
where $\omega^{(q)}\in\Lambda^q V$ for $q=1,2,5$; the overall factor of $-\tfrac{1}{32}$
is introduced so that $\omega^{(1)}$
coincides exactly
with the Dirac current of $\omega$. 
We may then write 
$$\eta(\gamma^\varphi(\omega)v,w)
=\tfrac{1}{3}\eta(\imath_v\imath_w\varphi,\omega^{(2)})
+\tfrac{1}{6}\eta(\imath_v\imath_w\star\varphi,\omega^{(5)})
$$
for all $v,w\in V$, see \cite[Eq. (9)]{Santi:2019kpx}, where $\star$ is the Hodge star operator on $V$.
The condition 
``$\gamma^\varphi(s,s)-X_V(\kappa(s,s))\in\fh_0$ for all $s\in S$'' decouples then into 
\begin{align*}
X_V(\omega^{(1)})&\in\fh_0\;,\\
\gamma^\varphi(\omega^{(2)}+\omega^{(5)})&\in\fh_0\;,
\end{align*}
for all $\omega\in\odot^2 S$. This and \eqref{eq:formeridentification} give the desired identification on $H^{2,2}(\fm,\fh)$.

The claim on $H^{3,2}(\fm,\fh)$ is straightforward once we note that $H^{3,2}(\fm,\fh)=Z^{3,2}(\fm,\fh)$ for any maximally supersymmetric subalgebra $\fh$ of $\fp$, since there are no coboundaries. 
\end{proof}

\subsection{Proof of Theorem \ref{thm:1} (stated on page 3)}
\label{sec:proofthm1}
\subsubsection{Normalization conditions for $2$-cocycles in $\mathbb Z$-grading $1$}
The collapsed Hilbert-Poincaré series at degree $1$ as determined in \eqref{tEA} of Proposition \ref{prop:CHPSd13}
suggests that the group $H^{1,2}(\fm,\fp)$ is non-zero, including at least an $\fso(V)$-module that is isomorphic to $S$. See also the Table \ref{SummarizingHPMW} of Theorem \ref{thm:2}, with the additional observation that the module $[1,1,0,0,1]$ is not present in form number $q=2$, but only in form number $q=3$. We first show that this is in fact the case and give a simple description of cohomology representatives for this $\fso(V)$-module. 
\begin{proposition}
\label{prop:cohomology1p}
The group $H^{1,2}(\fm,\fp)\supset S^*\cong S$ as an $\fso(V)$-submodule. More precisely, we may choose representatives as follows: any element $\phi\in S^*$ determines uniquely the cocycle
$\varepsilon^\phi+\epsilon^\phi$, where $\varepsilon^\phi:\odot^2S\to S$ and $\epsilon^\phi:S\otimes V\to V$ are given by
\begin{equation}
\label{eq:first-cocycle}
\begin{aligned}
\varepsilon^\phi(s,s)&=-2\phi(s)s\;,\\ 
\epsilon^\phi(s,v)&=-2\phi(s)v\;,
\end{aligned}
\end{equation}
for all $s\in S$ and $v\in V$.
\end{proposition}
\begin{proof}
Let $Z$ be the grading element of $\mathfrak{p}=\mathfrak{p}_0\oplus\fp_{-1}\oplus\fp_{-2}$,
which acts with eigenvalue $k$ on $\fp_k$. It can be identified with the dilation element in $\mathfrak{co}(V)$, in particular it does {\it not} belong $\fp_0\cong\mathfrak{so}(V)$.

First of all, we have 
\begin{align*}
Z^{1,2}(\fm,\fp)&\supset B^{1,2}(\fm,\fp)+\partial(\mathbb R Z\otimes S^*)\\
&=\partial(\fso(V)\otimes S^*)+\partial(S\otimes V^*)+\partial(\mathbb RZ\otimes S^*)
\end{align*} 
and the sum is direct, since the Spencer operator 
$
\partial:C^{1,1}(\fm,\fp\oplus\mathbb R Z)\to C^{1,2}(\fm,\fp\oplus\mathbb R Z)
$
extended with dilations
is injective.  (This follows again from \cite{MR3255456}: the first prolongation of $\fp\oplus\mathbb RZ$ is trivial.) 
Therefore
$B^{1,2}(\fm,\fp)=\partial(\fso(V)\otimes S^*)\oplus\partial(S\otimes V^*)$ and $H^{1,2}(\fm,\fp)\supset \partial(\mathbb RZ\otimes S^*)\cong S^*\cong S$.
\end{proof}
\begin{remark}
\label{rem:conformalextension}
{\rm
As established in the proof of Proposition \ref{prop:cohomology1p}, the representative $\varepsilon^\phi+\epsilon^\phi=\partial(Z\otimes\phi)$, where $Z$ is the grading element of $\fp$. In particular, this cohomology contribution would disappear if we were to consider the conformal extension $\fp\oplus\mathbb R Z\cong \mathfrak{co}(V)\oplus S\oplus V$ of the Poincaré superalgebra $\fp$ instead of $\mathfrak p$ itself.
}
\end{remark}
\begin{remark}\label{cohomologyclassequivalence}
{\rm
In the notation of \S\ref{sec6} the representative \eqref{eq:first-cocycle} would read (up to an overall factor)  as
 $\mathring \omega=\psi(\bar\psi q+X_a e^a)$,
and it does not coincide with the expression \eqref{eq:representative-12} of \S\ref{sec6}, which lives in $\odot^2 S\to S$ and has vanishing $S\otimes V\to V$ component. The contradiction is only apparent:
the space $B^{1,2}(\fm,\fp)$ of coboundaries includes two modules isomorphic to $S$, which can be used to modify cocycles at will. Indeed,
    $\mathring\omega-\tilde\omega= \diff(-2i\chi_1-\chi_2)$, for the coefficients appearing in \eqref{eq:representative-12} given by $a=-\frac{33}{16}$, $b= \frac{9}{32}$, $c = \frac{1}{640}$.
}
\end{remark}
We depart with a technical but useful representation-theoretic observation. With a little abuse of notation, we let 
\begin{equation}
\label{eq:soVequivariantClifford}
\begin{aligned}
\operatorname{Cl}&:V\otimes S\to S\\
&\;\;\;v\otimes s\mapsto v\cdot s\\
\\
\operatorname{Cl}&:\Lambda^2 V\otimes S\to V\otimes S\\
&\;\;\;v\wedge w\otimes s\mapsto \frac{1}{2}\big(v\otimes w\cdot s- w\otimes v\cdot s\big)\;,
\end{aligned}
\end{equation}
be the natural $\fso(V)$-equivariant Clifford multiplications with kernel $(V\otimes S)_o$ and $(\Lambda
^2 V\otimes S)_o$ respectively. By composing them, one also gets the $\fso(V)$-equivariant full Clifford multiplication $\Lambda^2 V\otimes S\to S$ that sends any $v\wedge w\otimes s$ to $v\wedge w\cdot s:=\frac{1}{2}\big(v\cdot w- w\cdot v\big)\cdot s$. Its kernel is isomorphic to the direct sum of $(V\otimes S)_o$ and $(\Lambda^2 V\otimes S)_o$. We now explicitly detail the natural $\fso(V)$-equivariant embeddings that are sections of the projections \eqref{eq:soVequivariantClifford}. To this aim, we fix an orthonormal basis $\{e_i\}_{i=0,\ldots,10}$ of $V$ and let the $\Gamma_i$'s be the associated Gamma matrices acting on the spinor module\footnote{The convention in \S\ref{sec:spencer-complex} is that $\{\Gamma_i,\Gamma_j\}=-2\eta_{ij}$, where $\eta$ is the flat metric on $V\cong \mathbb R^{1,10}$ with mostly minus signature. It differs slightly from those of \S\ref{sec6} because here it is more convenient to work in a purely real framework.}. As usual, we will tacitly use Einstein's summation convention on indices.
\begin{lemma}
\label{lemma:usefulrtobservation}
The maps
\begin{equation}
\label{eq:soVequivariantembeddings}
\begin{aligned}
\imath&:S\to V\otimes S\\
&\;\;\;s\mapsto -\frac{1}{11} e_i\otimes\Gamma^i\cdot s
\\
\\
\imath&:V\otimes S\to \Lambda^2 V\otimes S\\
&\;\;\;v\otimes s\mapsto -\frac{2}{9} \big(v\wedge e_i\otimes\Gamma^i\cdot s\big)+\frac{1}{90}(e_i\wedge e_j\otimes\Gamma^{ij}\cdot v \cdot s)
\end{aligned}
\end{equation}
are $\fso(V)$-equivariant sections of the projections \eqref{eq:soVequivariantClifford}. In particular, they are injective maps and their images are the unique $\fso(V)$-submodules $S$ into $V\otimes S$ and $V\otimes S$ into $\Lambda^2 V\otimes S$, respectively. By composing them, one also gets the $\fso(V)$-equivariant section $\imath:S\to\Lambda^2 V\otimes S$ of the full Clifford multiplication, which sends any $s$ to $-\tfrac{1}{110} e_i\wedge e_j\otimes \Gamma^{ij}\cdot s$.
\end{lemma}
\begin{proof}
Equivariance of the maps is clear by construction. To verify that $\imath:S\to V\otimes S$ is a section is sufficient to note that $\operatorname{Cl}(\imath(s))= -\frac{1}{11} \operatorname{Cl}(e_i\otimes\Gamma^i\cdot s)=-\frac{1}{11} \Gamma_i\cdot\Gamma^i\cdot s=s$, for all $s\in S$. 

Now $\operatorname{Cl}:\Lambda^2 V\otimes S\to V\otimes S$ sends $-\frac{2}{9} \big(v\wedge e_i\otimes\Gamma^i\cdot s\big)$ to
\begin{align*}
-\frac{1}{9} \big(v\otimes \Gamma_i\cdot\Gamma^i\cdot s-e_i\otimes v\cdot\Gamma^i\cdot s\big)&=
-\frac{1}{9} \big(-11 v\otimes s+e_i\otimes \Gamma^i\cdot v\cdot s+2\eta(v,e^i) e_i\otimes s\big)\\
&=-\frac{1}{9} \big(-9 v\otimes s+e_i\otimes \Gamma^i\cdot v\cdot s\big)\\
&=v\otimes s-\frac{1}{9}e_i\otimes \Gamma^i\cdot v\cdot s
\end{align*}
and $\frac{1}{90}e_i\wedge e_j\otimes \Gamma^{ij}\cdot v\cdot s$ to
\begin{align*}
\frac{1}{180} \big(e_i\otimes \Gamma_j\cdot\Gamma^{ij}\cdot v\cdot s-e_j\otimes \Gamma_i\cdot\Gamma^{ij}\cdot v\cdot s\big)&=\frac{1}{180} \big(10 e_i\otimes \Gamma^i\cdot v\cdot s+10e_j\otimes \Gamma^j\cdot v\cdot s\big)
\\
&=\frac{1}{9} e_i\otimes \Gamma^i\cdot v\cdot s \;,
\end{align*}
where we used that $\Gamma_j\Gamma^{ij}=\Gamma_j\Gamma^i\Gamma^j+\Gamma_j\eta^{ij}=9\Gamma^i+\Gamma^i=10\Gamma^i$ and $\Gamma_i\Gamma^{ij}=-\Gamma_i\Gamma^{ji}=-10\Gamma^j$. 
This shows that $\operatorname{Cl}(\imath(v\otimes s))=v\otimes s$. Finally, the last claim of the lemma is straightforward.
\end{proof}
 We remark again that the first prolongation 
$\fp_{(1)}\cong Z^{1,1}(\fm,\fp)$ of $\fp$ is trivial \cite{MR3255456}, so that $\partial:C^{1,1}(\fm,\fp)\to C^{1,2}(\fm,\fp)$ is injective. We now conclude this preliminary section by further strengthening this result, at the same time rephrasing the study of the cohomology group $H^{1,2}(\fm,\fp)$ as the study of the cocycle conditions on certain normalized cochains (in other words, we first use the freedom in coboundaries to normalize cochains).

We recall that $\odot^2 S\cong \Lambda^1V\oplus\Lambda^2 V\oplus\Lambda^5 V$ decomposes in a unique way as an $\fso(V)$-module, since each isotypic component is multiplicity free. This is also recasted in the well-known Fierz Identity in $D=11$ supergravity

\begin{equation}
s\overline s=-\frac{1}{32}\Big((\overline s\Gamma^\ell s)\Gamma_\ell+\tfrac{1}{2}(\overline s\Gamma^{\ell_1\ell_2} s)\Gamma_{\ell_1\ell_2}+\tfrac{1}{5!}(\overline s\Gamma^{\ell_1\cdots\ell_5}s)\Gamma_{\ell_1\cdots\ell_5}\Big)\;,
\label{eq:FierzIdentity}
\end{equation}
which expresses the rank $1$ endomorphism $s\overline s$ of $S$ in terms of Gamma matrices, for any $s\in S$. Here we abbreviated the symplectic dual $\langle s,-\rangle$ of a spinor $s\in S$ simply by $\overline s$. For more details, see, for instance, \cite[Appendix A]{Figueroa-OFarrill:2015rfh}.

\begin{proposition}
\label{prop:7}
The cohomology group $H^{1,2}(\fm,\fp)$ can be identified with the space of cocycles $\varepsilon+\epsilon\in C^{1,2}(\fm,\fp)$, where $\varepsilon:\odot^2S\to S$, $\epsilon:S\otimes V\to V$, that satisfy the normalization conditions
\begin{align}
\label{eq:claim1I}
\varepsilon|_{\Lambda^1 V}&=0\;,\\
\label{eq:claim1II}
\eta(\imath_s\epsilon (v),w)&=\eta(\imath_s\epsilon (w),v)
\end{align}
for all $s\in S$, $v,w\in V$.
\end{proposition}
\begin{proof}
We let
\begin{align*}
\pi&:C^{1,2}(\fm,\fp)\to \Hom(\Lambda^1 V,S)\oplus \Hom(S,\fso(V))\\
&\;\;\;\varepsilon+\epsilon\mapsto\varepsilon|_{\Lambda^1 V}+\operatorname{Skew}_V(\epsilon)
\end{align*}
be the natural projection of $C^{1,2}(\fm,\fp)$ given by restriction of elements $\varepsilon:\odot^2 S\to S$ to $\Lambda^1 V\subset \odot^2 S$ and skew-symmetrization in $V$ of elements $\epsilon:S\otimes V\to V$. Our normalizations \eqref{eq:claim1I}-\eqref{eq:claim1II} can be enforced if the composition $\pi\circ\partial:C^{1,1}(\fm,\fp)\to  \Hom(\Lambda^1 V,S)\oplus \Hom(S,\fso(V))$ is an isomorphism. Since the domain and codomain of $\pi\circ\partial$ are  both abstractly isomorphic to $(\Lambda^2 V\otimes S)_o\oplus 2(V\otimes S)_o\oplus 2S$ as $\fso(V)$-modules, it is enough to show injectivity, and, by $\fso(V)$-equivariance, this can be verified separately for each isotypic component.

The coboundary of an element $c=a+b\in C^{1,1}(\fm,\fp)$, with $a:V\to S$ and $b:S\to\fso(V)$, is given by the formulae
\begin{align*}
\partial c(s,s)&=2b(s)s-a(\kappa(s,s))\;,\\
\partial c(v,s)&=\kappa(s,a(v))-b(s)v\;,
\end{align*}
where $s,s_1,s_2\in S$ and $v\in V$. Here $\kappa(s,s)$ is the usual Dirac current. 
\vskip0.2cm\par\noindent
\underline{The isotypic component $(\Lambda^2 V\otimes S)_o$}
\vskip0.2cm\par\noindent
If $b:S\to \fso(V)$ is an element of $(\Lambda^2 V\otimes S)_o$ then
$$(\pi\circ\partial b)(s,s)=0$$  automatically for all $s\in S$, because $\Hom(\Lambda^1 V,S)$ has no submodule isomorphic to $(\Lambda^2 V\otimes S)_o$. On the other hand $(\pi\circ\partial b)(v,s)=-b(s)v=0$ for all $s\in S$, $v\in V$, directly implies $b=0$.
\vskip0.2cm\par\noindent
\underline{The isotypic component $2(V\otimes S)_o$}
\vskip0.2cm\par\noindent
Thanks to Lemma \ref{lemma:usefulrtobservation}, maps $a:V\to S$ and $b:S\to \fso(V)$ that belong to the component $(V\otimes S)_o$ can be written as
\begin{align*}
a&=s_k\otimes e^k:V\to S\;,\qquad
b=-\frac{2}{9}e^k\wedge e^i\otimes\langle\Gamma_i\cdot t_k,-\rangle:S\to\fso(V)\;,
\end{align*}
where $s_k, t_k\in S$ for all $k=0,\ldots,10$, and for which the Clifford multiplications are vanishing: 
\begin{equation}
\label{eq:traceconditions}
\operatorname{Cl}(e^k\otimes s_k)=\Gamma^k\cdot s_k=0\qquad \operatorname{Cl}(e^k\otimes t_k)=\Gamma^k\cdot t_k=0\;.
\end{equation}
Letting $c=a+b$, we compute
\begin{align*}
\partial c(e_\ell,s)&=\kappa(s,a(e_\ell))-b(s)e_\ell\\
&=\langle s,\Gamma_i\cdot s_\ell\rangle e^i+\frac{2}{9}\langle\Gamma_i\cdot t_\ell,s\rangle e^i-\frac{2}{9}\langle\Gamma_\ell\cdot t_i,s\rangle e^i\\
&=-\langle \Gamma_i\cdot s_\ell, s\rangle e^i+\frac{2}{9}\langle\Gamma_i\cdot t_\ell,s\rangle e^i-\frac{2}{9}\langle\Gamma_\ell\cdot t_i,s\rangle e^i
\end{align*}
and take the scalar product with $e_j$ to get $-\langle \Gamma_j\cdot s_\ell,s\rangle+\frac{2}{9}\langle\Gamma_j\cdot t_\ell,s\rangle-\frac{2}{9}\langle\Gamma_\ell\cdot t_j,s\rangle$.
Skew-symmetrizing in the indices $\ell$ and $j$ and eliminating $s$ finally yields $-\Gamma_j\cdot s_\ell+\Gamma_\ell\cdot s_j
+\frac{4}{9}\Gamma_j\cdot t_\ell
-\frac{4}{9}\Gamma_\ell\cdot t_j$, whose vanishing is 
\begin{align}
\Gamma_j\cdot\big(-s_\ell+\frac{4}{9}t_\ell\big)&=
\Gamma_\ell\cdot\big(-s_j
+\frac{4}{9}t_j\big)\;,
\end{align}
for all $\ell,j=0,\ldots,10$. It is not difficult to check that this is equivalent to the condition $s_\ell=\frac{4}{9}t_\ell$. 

On the other hand, we may also compute
\begin{align*}
\partial c(s,s)&=2b(s)s-a(\kappa(s,s))\\
&=-\frac{2}{9}\Gamma^{ki}\cdot s\otimes\langle\Gamma_i\cdot t_k,s\rangle-\langle s,\Gamma^\ell\cdot s\rangle s_\ell\\
&=\frac{2}{9}\Gamma^{ki}\cdot s\otimes\langle s,\Gamma_i\cdot t_k\rangle-\langle s,\Gamma^\ell\cdot s\rangle s_\ell\\
&=\frac{2}{9}\Gamma^{ki}\cdot \big((s\overline s)\Gamma_i\cdot t_k\big)-(\overline s\Gamma^\ell s) s_\ell\;,
\end{align*}
where we used that the spinorial action of $\fso(V)$ is half the Clifford multiplication of $\Lambda^2 V$, under the natural identification $\fso(V)\cong\Lambda^2 V$. Using the Fierz Identity \eqref{eq:FierzIdentity} on $s\overline s$ and retaining only the contribution in $\Lambda^1V\subset \odot^2 S$, we finally arrive at
\begin{align*}
\partial c|_{\Lambda^1 V}(s,s)&=-(\overline s\Gamma^\ell s)\big(s_\ell+\frac{1}{32}\frac{2}{9}\Gamma^{ki}
\Gamma_\ell
\Gamma_i\cdot t_k\big)\\
&=-(\overline s\Gamma^\ell s)\big(s_\ell+\frac{1}{16}\Gamma^k\Gamma_\ell\cdot t_k\big)\\
&=-(\overline s\Gamma^\ell s)\big(s_\ell-\frac{1}{8}t_\ell\big)\;,
\end{align*}
where we used the identities $\Gamma^{ki}\Gamma_\ell\Gamma_i=\Gamma^k\Gamma^i\Gamma_\ell\Gamma_i+\eta^{ki}\Gamma_\ell\Gamma_i=9\Gamma^k\Gamma_\ell+\Gamma_\ell\Gamma^k$ and $\Gamma_k\Gamma_\ell=-\Gamma_\ell\Gamma_k-2\eta_{\ell k}$, and the second trace condition in \eqref{eq:traceconditions}. The vanishing of this term for all $s\odot s\in\Lambda^1 V\subset \odot^2 S$ is equivalent to the condition $s_\ell=\frac{1}{8}t_\ell$.  

In summary, we arrived at the conditions $s_\ell=\frac{4}{9}t_\ell$ and $s_\ell=\frac{1}{8}t_\ell$, thus $s_\ell=t_\ell=0$ for all $\ell=0,\ldots,10$. In other words, $a=b=0$ and $\pi\circ\partial$ is injective on the isotypic component $2(V\otimes S)_o$.
\vskip0.2cm\par\noindent
\underline{The isotypic component $2S$}
\vskip0.2cm\par\noindent
The strategy is similar to the previous case. By Lemma \ref{lemma:usefulrtobservation}, any $s,\widetilde s\in S$ determine maps $a:V\to S$ and $b:S\to\fso(V)$ given by
\begin{align*}
a&=-\frac{1}{11}\Gamma_i\cdot s\otimes e^i\;,\qquad
b=-\frac{1}{110} e_i\wedge e_j\otimes \langle\Gamma^{ij}\cdot \widetilde s,-\rangle\;.
\end{align*}
Letting $c=a+b$, we compute
\begin{align*}
\partial c(e_\ell,t)&=\kappa(t,a(e_\ell))-b(t)e_\ell\\
&=\frac{1}{11}\langle \Gamma_i\Gamma_\ell\cdot s,t\rangle e^i
+\frac{1}{55} \langle\Gamma_{\ell i}\cdot \widetilde s,t\rangle e^i
\\
&=\frac{1}{11}\langle \Gamma_{i\ell}\cdot s,t\rangle e^i
-\frac{1}{11}\langle s,t\rangle e_\ell
+\frac{1}{55} \langle\Gamma_{\ell i}\cdot \widetilde s,t\rangle e^i
\end{align*}
and take the scalar product with $e_j$ to get $\frac{1}{11}\langle \Gamma_{j\ell}\cdot s,t\rangle 
-\frac{1}{11}\langle s,t\rangle \eta_{\ell j}
+\frac{1}{55} \langle\Gamma_{\ell j}\cdot \widetilde s,t\rangle$.
Skew-symmetrizing in the indices $\ell$ and $j$ and eliminating $t$ gives $\frac{1}{11}\Gamma_{j\ell}\cdot s 
+\frac{1}{55} \Gamma_{\ell j}\cdot \widetilde s$, whose vanishing is $\widetilde s=5s$.

On the other hand,
\begin{align*}
\partial c(t,t)&=2b(t)t-a(\kappa(t,t))\\
&=-\frac{1}{110}\langle\Gamma^{ij}\cdot \widetilde s,t\rangle\Gamma_{ij}\cdot t
+\frac{1}{11}\langle t,\Gamma^j\cdot t\rangle \Gamma_j\cdot s\\
&=\frac{1}{110}\Gamma_{ij}\cdot \big((t\overline t)\Gamma^{ij}\cdot \widetilde s\big)
+\frac{1}{11}(\overline t\Gamma^j t)
\Gamma_j\cdot s\;.
\end{align*}
 Using the Fierz Identity \eqref{eq:FierzIdentity} on $t\overline t$ and retaining only the contribution in $\Lambda^1V\subset \odot^2 S$, we finally arrive at
\begin{align*}
\partial c|_{\Lambda^1 V}(t,t)
&=-\frac{1}{110}\frac{1}{32}(\overline t\Gamma^\ell t)\Gamma_{ij}
\Gamma_\ell
\Gamma^{ij}\cdot \widetilde s
+\frac{1}{11}(\overline t\Gamma^\ell t)
\Gamma_\ell\cdot s\\
&=\frac{1}{110}\frac{70}{32}(\overline t\Gamma^\ell t)
\Gamma_\ell\cdot \widetilde s
+\frac{1}{11}(\overline t\Gamma^\ell t)
\Gamma_\ell\cdot s\;,
\end{align*}
where we used that $\Gamma_{ij}\Gamma_\ell\Gamma^{ij}=-70\Gamma_\ell$. The vanishing of this term for all $t\odot t\in\Lambda^1 V\subset \odot^2 S$ is equivalent to the condition 
$\frac{7}{32}\widetilde s+s=0$.

In summary, we have the conditions $\widetilde s=5s$ and $\frac{7}{32}\widetilde s+s=0$, thus $s=\widetilde s=0$. In other words, $a=b=0$ and $\pi\circ\partial$ is injective on the isotypic component $2S$.
\end{proof}
We shall study the maps $\varepsilon+\epsilon\in C^{1,2}(\fm,\fp)$, where $\varepsilon:\odot^2 S\to S$ and $\epsilon:S\otimes V\to V$, which in addition satisfy the cocycle condition 
\begin{equation}
\kappa(s,\varepsilon(s,s))+\beta(\epsilon(s,s),s)=0\;,
\label{eq:cocycleconditiondegree1}
\end{equation}
for all $s\in S$.  We may assume that $\imath_s\epsilon:V\to V$ is {\it symmetric} in $V$ for any fixed $s\in S$
 and that $\varepsilon|_{\Lambda^1 V}=0$, under the natural identification $\odot^2 S\cong \Lambda^1 V\oplus\Lambda^2 V\oplus\Lambda^5 V$.  In fact, these assumptions are precisely the normalization conditions \eqref{eq:claim1I}-\eqref{eq:claim1II} from Proposition \ref{prop:7}. Furthermore, removing the non-trivial contribution $\varepsilon^\phi+\epsilon^\phi$ to $H^{1,2}(\fm,\fp)$ that we already isolated in Proposition \ref{prop:cohomology1p}, we may also assume w.l.o.g. that $\imath_s\epsilon:V\to V$ is {\it traceless} in $V$ for any fixed $s\in S$. Referring to such cochains as ``normalized'', we have proved the following.
\begin{corollary}
\label{cor:final}
The group $H^{1,2}(\fm,\fp)$ is isomorphic as $\fso(V)$-module to the direct sum of $S$ and the space of normalized cocycles, namely the space of maps $\varepsilon+\epsilon\in C^{1,2}(\fm,\fp)$ that satisfy the system of equations
\begin{align}
\label{eq:claim1Ibis}
\varepsilon|_{\Lambda^1 V}&=0\;,\\
\label{eq:claim1IIbis}
\eta(\imath_s\epsilon (v),w)&=\eta(\imath_s\epsilon (w),v)\;,\\
\eta(\imath_s\epsilon (e_i),e^i)&=0\;,\\
\kappa(s,\varepsilon(s,s))+\beta(\epsilon(s,s),s)&=0\;,
\end{align}
for all $s\in S$, $v,w\in V$.
\end{corollary}

\subsubsection{The group $H^{1,2}(\fm,\fp)$}
\label{subsubsec:final12}
Our goal is to show that the space of the normalized cocycles as detailed in Corollary \ref{cor:final} is trivial. 
We depart by partly polarizing \eqref{eq:cocycleconditiondegree1} 
and taking the scalar product with a generic $v\in V$ to get
\begin{equation}
\label{eq:partlypolarized}
\begin{aligned}
0&=2\langle \varepsilon(s,t),v\cdot s\rangle+\langle\varepsilon(s,s),v\cdot t\rangle+2\eta\big(\epsilon(\kappa(t,s),s),v\big)+\eta\big(\epsilon(\kappa(s,s),t),v\big)\\
&=2\langle \varepsilon(s,t),v\cdot s\rangle+\langle\varepsilon(s,s),v\cdot t\rangle+2\eta\big(\epsilon(v,s),\kappa(t,s)\big)+\eta\big(\epsilon(v,t),\kappa(s,s)\big)\\
&=2\langle \varepsilon(s,t),v\cdot s\rangle+\langle\varepsilon(s,s),v\cdot t\rangle+2\langle t,\epsilon(v,s)\cdot s\rangle+\langle s,\epsilon(v,t)\cdot s\rangle \;,
\end{aligned}
\end{equation}
where we used that $\imath_s\epsilon:V\rightarrow V$ is symmetric in $V$. Fixing an orthonormal basis $\{e_i\}_{i=0,\ldots,10}$ of $V$ with associated Gamma matrices $\Gamma_i$, we may write
\begin{equation}
\label{eq:definitionalphabeta}
\begin{aligned}
\varepsilon(s,t)&=\varepsilon^{\ell_1\ell_2}\frac{1}{2}(\overline s\Gamma_{\ell_1\ell_2}t)+
\varepsilon^{\ell_1\cdots\ell_5}\frac{1}{5!}(\overline s\Gamma_{\ell_1\cdots\ell_5}t)\;,\\
\epsilon(v,s)&=e_i\otimes e^j(v)(\overline{\epsilon^i{}_{j}} s)\;,
\end{aligned}
\end{equation}
for all $s,t\in S$ and $v\in V$. Here each of the elements $\varepsilon^{\ell_1\ell_2}$, $\varepsilon^{\ell_1\cdots\ell_5}$, $\epsilon_{ij}$ is in $S$, for any fixed indices. As already explained, the spinors $\epsilon_{ij}$ are symmetric traceless in the indices $i$ and $j$.
\begin{lemma}
The cocyle condition \eqref{eq:partlypolarized} on normalized cochains is equivalent to the vanishing of
\begin{equation}
\label{eq:reformulationcocycle}
\begin{aligned}
&(\overline{\varepsilon^{m_1m_2}}\Gamma_j s)\overline s\Gamma_{m_1m_2}+
\frac{2}{5!}(\overline{\varepsilon^{m_1\cdots m_5}}\Gamma_j s)\overline s\Gamma_{m_1\cdots m_5}+
\frac12(\overline s\Gamma_{\ell_1\ell_2}s)\overline{\varepsilon^{\ell_1\ell_2}}\Gamma_j+\\
&\frac{1}{5!}(\overline s\Gamma_{\ell_1\cdots\ell_5}s)\overline{\varepsilon^{\ell_1\cdots\ell_5}}\Gamma_j+
2(\overline{\epsilon^i{}_j} s)\overline s\Gamma_i+(\overline s\Gamma_\ell s)\overline{\epsilon^\ell{}_j}
\end{aligned}
\end{equation}
for all $j=0,\ldots,10$, and $s\in S$.
\end{lemma}

\begin{proof}
This is obtained by substituting \eqref{eq:definitionalphabeta} into
\eqref{eq:partlypolarized} with $v=e_j$ and abstracting $t\in S$.
\end{proof}
Using the Fierz Identity 
\eqref{eq:FierzIdentity} in \eqref{eq:reformulationcocycle} and abstracting the independent contributions in 
$\odot^2 S\cong \Lambda^1 V\oplus\Lambda^2 V\oplus\Lambda^5 V$, we arrive at three separate equations, which we now detail.
\vskip0.4cm\par\noindent
\underline{The contribution coming from $s\odot s\in\Lambda^1 V\subset\odot^2 S$}
\vskip0.2cm\par\noindent
This identity reads as
\begin{equation*}
\begin{aligned}
&-\frac{1}{32}\overline{\varepsilon^{m_1m_2}}\Gamma_j 
\Gamma_\ell
\Gamma_{m_1m_2}
-\frac{1}{32}\frac{2}{5!}\overline{\varepsilon^{m_1\cdots m_5}}\Gamma_j 
\Gamma_\ell
\Gamma_{m_1\cdots m_5}
-\frac{1}{32}2\overline{\epsilon^i{}_j}
\Gamma_\ell
\Gamma_i
+\overline{\epsilon_\ell{}_j}=0
\end{aligned}
\end{equation*}
and, upon dualization, it becomes
\begin{equation}
\label{eq:reformulationcocycle1}
\begin{aligned}
\frac{1}{32}\Gamma_{m_1m_2}\Gamma_\ell\Gamma_j\varepsilon^{m_1m_2}
+\frac{1}{32}\frac{2}{5!}\Gamma_{m_1\cdots m_5}\Gamma_\ell\Gamma_j \varepsilon^{m_1\cdots m_5}
-\frac{1}{32}2
\Gamma_i\Gamma_\ell\epsilon^i{}_j
+\epsilon_\ell{}_j&=0\;.
\end{aligned}
\end{equation}
The equation holds for all indices $j,\ell=0,\ldots,10$.
\vskip0.2cm\par\noindent
\underline{The contribution coming from $s\odot s\in\Lambda^2 V\subset\odot^2 S$}
\vskip0.2cm\par\noindent
This identity reads as
\begin{equation*}
\begin{aligned}
&-\frac{1}{32}\overline{\varepsilon^{m_1m_2}}\Gamma_j 
\Gamma_{\ell_1\ell_2}
\Gamma_{m_1m_2}
-\frac{1}{32}\frac{2}{5!}\overline{\varepsilon^{m_1\cdots m_5}}\Gamma_j 
\Gamma_{\ell_1\ell_2}
\Gamma_{m_1\cdots m_5}+
\overline{\varepsilon_{\ell_1\ell_2}}\Gamma_j
-\frac{1}{32}2\overline{\epsilon^i{}_j} 
\Gamma_{\ell_1\ell_2}
\Gamma_i=0
\end{aligned}
\end{equation*}
and, upon dualization, it becomes
\begin{equation}
\label{eq:reformulationcocycle2}
\begin{aligned}
\!\!\!\!\!\!\frac{1}{32}\Gamma_{m_1m_2}\Gamma_{\ell_1\ell_2}\Gamma_j\varepsilon^{m_1m_2}
+\frac{1}{32}\frac{2}{5!}\Gamma_{m_1\cdots m_5}\Gamma_{\ell_1\ell_2}\Gamma_j\varepsilon^{m_1\cdots m_5}
-\Gamma_j\varepsilon_{\ell_1\ell_2}
-\frac{1}{32}2
\Gamma_i\Gamma_{\ell_1\ell_2}\epsilon^i{}_j
=0\;.
\end{aligned}
\end{equation}
The equation holds for all indices $j,\ell_1,\ell_2=0,\ldots,10$.

\vskip0.2cm\par\noindent
\underline{The contribution coming from $s\odot s\in\Lambda^5 V\subset\odot^2 S$}
\vskip0.2cm\par\noindent
This identity reads as
\begin{equation*}
\begin{aligned}
&-\frac{1}{32}\overline{\varepsilon^{m_1m_2}}\Gamma_j 
\Gamma_{\ell_1\cdots\ell_5}
\Gamma_{m_1m_2}-\frac{1}{32}\frac{2}{5!}\overline{\varepsilon^{m_1\cdots m_5}}\Gamma_j 
\Gamma_{\ell_1\cdots\ell_5}
\Gamma_{m_1\cdots m_5}
+\overline{\varepsilon_{\ell_1\cdots\ell_5}}\Gamma_j
-\frac{1}{32}2\overline{\epsilon^i{}_j} 
\Gamma_{\ell_1\cdots\ell_5}
\Gamma_i=0
\end{aligned}
\end{equation*}
and, upon dualization, it becomes
\begin{equation}
\label{eq:reformulationcocycle3}
\begin{aligned}
\frac{1}{32}\Gamma_{m_1m_2}\Gamma_{\ell_1\cdots\ell_5}\Gamma_j\varepsilon^{m_1m_2} 
+\frac{1}{32}\frac{2}{5!}\Gamma_{m_1\cdots m_5}\Gamma_{\ell_1\cdots\ell_5}\Gamma_j\varepsilon^{m_1\cdots m_5} 
-\Gamma_j\varepsilon_{\ell_1\cdots\ell_5}
-\frac{1}{32}2\Gamma_i\Gamma_{\ell_1\cdots\ell_5}\epsilon^i{}_j 
=0\;.
\end{aligned}
\end{equation}
The equation holds for all indices $j,\ell_1,\ldots,\ell_5=0,\ldots,10$.
\vskip0.3cm\par\noindent
The system of equations \eqref{eq:reformulationcocycle1}-\eqref{eq:reformulationcocycle3} looks rather beautiful and challenging. Although there are several ways to simplify the system from the representation-theoretic point of view, we haven't been able to find a sufficiently clear and complete proof that avoids discussing too many subcases. It is therefore a matter of calculating the resulting expressions using our favorite explicit realization of the Clifford algebra and see that $\varepsilon=\epsilon=0$ is the only solution to the system; the explicit verification can be found in the Mathematica supplement accompanying the arXiv posting of this article. 
This proves:

\begin{theorem}
\label{thm:persoilconto}
The group $H^{1,2}(\fm,\fp)\cong S$ as an $\fso(V)$-module.
\end{theorem}
As already advertised at the beginning of \S\ref{sec:proofthm1}, this is coherent with the result on the Euler characteristic $\chi(H^{1,\bullet}(\mathfrak m, \mathfrak p)\otimes S^*)^G=+1$ obtained in \S\ref{sec6} using the Molien-Weyl formula, rigorously setting the result suggested in \S\ref{sec3.3}.
\subsubsection{The group $H^{3,2}(\fm,\fp)$}
\label{subsubH32}
We here deal with Spencer $2$-cochains $\sigma+\tau$ of degree $3$, where $\sigma:V \otimes S \to \fso(V)$ and $\tau:\Lambda^2 V\to\ S$. Since there are no Spencer $1$-cochains of degree $3$ (see \eqref{tEABB} and Table \ref{tab:even-cochains-small}), the cohomology group $H^{3,2}(\fm,\fp)$ coincides with the space of Spencer $2$-cochains satisfying the cocycle conditions
\begin{equation}
\label{eq:cocycleeqsdegree3}
\begin{aligned}
\sigma(\kappa(s,s),s)&=0\;,\\
\kappa(s,\tau(v,w))&=\sigma(v,s)w-\sigma(w,s)v\;,\\
\tau(\kappa(s,s),v)&=-2\sigma(v,s)s\;,
\end{aligned}
\end{equation} 
for all $s\in S$, $v,w\in V$. We note that the space where the component $\tau$ lives is isomorphic as an $\fso(V)$-module to 
$\Lambda^2 V\otimes S\cong (\Lambda^2 V\otimes S)_o\oplus (V\otimes S)_o\oplus S$, thus it consists of $3$ irreducible components, while the space where $\sigma$ lives is much bigger and it consists of $10$ irreducible components. However the following result cuts down a lot of the freedom.
\begin{proposition}
\label{prop:tausigma}
If $\sigma+\tau$ is a Spencer $2$-cocycle of degree $3$, then $\tau$ uniquely determines $\sigma$ via any of the last two equations in \eqref{eq:cocycleeqsdegree3}. Explicitly, if $f:=\imath_s\sigma:V\to\fso(V)$ and $g:=\imath_v\sigma:S\to\fso(V)$, then
\begin{align}
\label{eq:explicitexpressionsI}
2\eta\bigg(f(v)w,u\bigg)
&=-\biggl\langle \tau(u,v),w\cdot s\biggr\rangle-\biggl\langle\tau(u,w),v\cdot s\biggr\rangle-\biggl\langle \tau(v,w),u\cdot s \biggr\rangle\;,\\
\label{eq:explicitexpressionsII}
2\biggl\langle g(r)s,t\biggr\rangle&=-\biggl\langle\tau(\kappa(t,r),v),s\biggr\rangle+\biggl\langle\tau(\kappa(s,t),v),r\biggr\rangle-\biggl\langle\tau(\kappa(r,s),v),t\biggr\rangle\;,
\end{align}
for all $r,s,t\in S$ and $u,v,w\in V$.
In particular the group $H^{3,2}(\fm,\fp)$ is isomorphic to a submodule of $\Lambda^2 V\otimes S\cong (\Lambda^2 V\otimes S)_o\oplus (V\otimes S)_o\oplus S$.
\end{proposition}
\begin{proof}
The last two equations in \eqref{eq:cocycleeqsdegree3} can be rewritten as
\begin{align*}
\tau(\kappa(s,s),v)&=-\slashed{\partial}(\imath_v\sigma)(s,s)\;,\\
\kappa(s,\tau(v,w))&=-\slashed{\partial}(\imath_s\sigma)(v,w)\;,
\end{align*}
for all $s\in S$, $v,w\in V$. Here $\slashed{\partial}$ is the Spencer operator of the linear Lie algebra $\fso(V)$ acting on the purely odd $S$ (i.e., symmetrization) and on the purely even $V$ (i.e., skew-symmetrization), respectively. It is well-known that the first prolongation of $\fso(V)$ is trivial in both cases (for the first case, note that $\fso(V)\subset \mathfrak{sp}(S)$ and that the first prolongation of  $\mathfrak{sp}(S)$ on the purely odd $S$ is trivial by \cite[Thm. 5.1]{Ga}). Then $\tau$ determines $f=\imath_v\sigma:S\to\fso(V)$ and $g=\imath_s\sigma:V\to\fso(V)$, respectively, and any of the two suffices for our first claim. The expressions \eqref{eq:explicitexpressionsI}-\eqref{eq:explicitexpressionsII} can be verified by checking the last two equations in \eqref{eq:cocycleeqsdegree3}, due to uniqueness of $f$ and $g$ (alternatively, they are obtained by direct combinatorial arguments, but we will not do it). The rest is clear.
\end{proof}
\begin{remark}
{\rm
Although we will not need this fact, it is worth to note that the first cocycle condition in \eqref{eq:cocycleeqsdegree3} is redundant, as it follows directly from the other two:
\begin{align*}
-2\kappa\Big(s,\sigma(w,s)s\Big)=\kappa\Big(s,\tau(\kappa(s,s),w)\Big)&=\sigma(\kappa(s,s),s)w-\sigma(w,s)\kappa(s,s)\\
&=\sigma(\kappa(s,s),s)w-2\kappa\Big(\sigma(w,s)s,s\Big)\;,
\end{align*}
from which $\sigma(\kappa(s,s),s)w=0$ for all $s\in S$ and $w\in V$, i.e., $\sigma(\kappa(s,s),s)=0$ for all $s\in S$.
}
\end{remark}
Fixing an orthonormal basis $\{e_i\}_{i=0,\ldots,10}$ of $V$ and a basis $\{q_\alpha\}_{\alpha=1,\ldots, 32}$ of $S$ as usual, we write
\begin{align*}
\sigma&=\frac{1}{2}e_\ell\wedge e_m\sigma^{\ell m}{}_{k\beta}\otimes e^k\otimes q^\beta\;,\\
\tau&=\frac{1}{2}q_\alpha\tau^\alpha{}_{ij}\otimes e^i\wedge e^j\;,
\end{align*}
and rewrite \eqref{eq:explicitexpressionsI}-\eqref{eq:explicitexpressionsII} choosing $s=q_\alpha$, $v=e_i$, $w=e_j$, $u=e_k$, and, respectively, $r=q_\gamma$, $v=e_p$, $s=t=s^\delta q_\delta$ into
\begin{align}
\label{eq:explicitexpressionsIbis}
2\sigma_{kji\alpha}&=-\tau^\beta{}_{ki}\overline{q_\beta}\Gamma_jq_\alpha+\tau^\beta{}_{jk}\overline{q_\beta}\Gamma_iq_\alpha-\tau^\beta{}_{ij}\overline{q_\beta}\Gamma_kq_\alpha\;,\\
\label{eq:explicitexpressionsIIbis}
-\frac{1}{2}\sigma_{\ell m p\gamma}(\overline s\Gamma^{\ell m}s)&=2(\overline s\Gamma^qq_\gamma)\tau_{\delta qp}s^\delta-(\overline s\Gamma^q s)\tau_{\gamma qp}\;.
\end{align}
Substituting \eqref{eq:explicitexpressionsIbis} into \eqref{eq:explicitexpressionsIIbis} yields
\begin{align*}
-\frac{1}{2}\Big(-\tau^\beta{}_{\ell p}(\overline{q_\beta}\Gamma_mq_\gamma)+\tau^\beta{}_{m\ell}(\overline{q_\beta}\Gamma_p q_\gamma)-\tau^\beta{}_{pm}(\overline{q_\beta}\Gamma_\ell q_\gamma)\Big)(\overline s\Gamma^{\ell m}s)&=4s^\delta\tau_{\delta qp}(\overline s\Gamma^qq_\gamma)-2\tau_{\gamma qp}(\overline s\Gamma^q s)\\
&=-4\tau^\alpha{}_{qp}\overline{q_\alpha} (s\overline s)\Gamma^qq_\gamma-2\tau_{\gamma qp}(\overline s\Gamma^q s)
\end{align*}
which is an equation on $\tau$ only, for all $s\in S$ and indices $p=0,\ldots,10$, $\gamma=1,\ldots,32$. Substituting the Fierz Identity \eqref{eq:FierzIdentity} and abstracting the independent contributions in $\odot^2 S\cong \Lambda^1 V\oplus\Lambda^2 V\oplus\Lambda^5 V$, we arrive at three separate equations (since we won't make any use of the contribution coming from $s\odot s\in\Lambda^1 V\subset\odot^2 S$, we omit it).
\vskip0.2cm\par\noindent
\underline{The contribution coming from $s\odot s\in\Lambda^5 V\subset\odot^2 S$}
\vskip0.2cm\par\noindent
This identity reads as
\begin{align*}
\tau^\alpha{}_{qp}\overline{q_\alpha}\Gamma_{\mu_1\cdots\mu_5}\Gamma^q&=0\;,
\end{align*}
for all indices $p=0,\ldots 10$ and $\mu_1,\ldots,\mu_5=1,\ldots,32$. In particular, we may multiply by $\Gamma^{\mu_1\cdots\mu_5}$ from the right and, using that $\Gamma_{\mu_1\cdots\mu_5}\Gamma^q\Gamma^{\mu_1\cdots\mu_5}=5040\Gamma^q$, arrive at 
$\tau^\alpha{}_{qp}\Gamma^qq_\alpha=0$. This exactly means that $\tau:\Lambda^2 V\to S$ is in the kernel of the Clifford multiplication \eqref{eq:soVequivariantClifford}, thus $\tau\in (\Lambda^2 V\otimes S)_o$.
\vskip0.2cm\par\noindent
\underline{The contribution coming from $s\odot s\in\Lambda^2 V\subset\odot^2 S$}
\vskip0.2cm\par\noindent
After abstracting $q_\gamma$ and dualizing, we get
\begin{align*}
\tau^\beta{}_{\ell p}\Gamma_mq_\beta-\tau^\beta{}_{m\ell}\Gamma_pq_\beta+\tau^\beta{}_{pm}\Gamma_\ell q_\beta
&=-\frac{1}{8}\tau^{\alpha q}{}_{p}\Gamma_q\Gamma_{\ell m}q_\alpha\\
&=-\frac{1}{8}\tau^{\alpha q}{}_{p}\big(\Gamma_{\ell m}\Gamma_q+2\eta_{mq}\Gamma_\ell-2\eta_{\ell q}\Gamma_m\big)q_\alpha\\
&=\frac{1}{4}\tau^{\beta}{}_{pm}\Gamma_\ell q_\beta
+\frac{1}{4}\tau^{\beta}{}_{\ell p}\Gamma_m q_\beta\;,
\end{align*}
where we used that $\tau\in (\Lambda^2 V\otimes S)_o$. In other words 
$
\tfrac34\tau^\beta{}_{\ell p}\Gamma_mq_\beta-\tau^\beta{}_{m\ell}\Gamma_pq_\beta+\tfrac34\tau^\beta{}_{pm}\Gamma_\ell q_\beta
=0
$
and multiplying by $\Gamma^p$ from the left we get
\begin{align*}
0&=\tfrac34\tau^\beta{}_{\ell}{}^{p}\Gamma_p\Gamma_mq_\beta+11\tau^\beta{}_{m\ell}q_\beta+\tfrac34\tau^{\beta p}{}_{m}\Gamma_p\Gamma_\ell q_\beta\\
&=-\tfrac32\tau^\beta{}_{\ell}{}_{m}q_\beta+11\tau^\beta{}_{m\ell}q_\beta-\tfrac32\tau^{\beta}{}_{\ell m}q_\beta\;,
\end{align*}
where we used that $\Gamma_p\Gamma_k=-\Gamma_k\Gamma_p-2\eta_{pk}$ and $\tau\in (\Lambda^2 V\otimes S)_o$. Hence 
\begin{align*}
-14 \tau^\beta{}_{\ell}{}_{m}q_\beta=0&\Longrightarrow\tau=0\\
&\Longrightarrow\sigma=0
\end{align*} 
thanks to Proposition \ref{prop:tausigma}. We thus proved the following:
\begin{theorem}
\label{thm:H32}
The group $H^{3,2}(\fm,\fp)$ is trivial.
\end{theorem}
\section{Maximally supersymmetric even/odd filtered subdeformations of the $D=11$ Poincaré superalgebra}
\label{sec:deformation-cohomology}

In this section, we study maximally supersymmetric filtered subdeformations of $\mathfrak p$. This involves, at first-order, the cohomology of the Spencer complex of $\fh$, for which the classification of the relevant groups has now been completed: see the combination of Theorem \ref{thm:1} and Proposition \ref{prop:cohomology1h}.

\subsection{Preliminary definitions}
\label{subsec:definitiondeformation}
We shall here seek
for filtered deformations of maximally supersymmetric $\mathbb Z$-graded subalgebras
$\fh=\fh_{-2}\oplus\fh_{-1}\oplus\fh_0$ of the Poincaré superalgebra $\fp=\fp_{-2}\oplus\fp_{-1}\oplus\fp_0$
as of \S\ref{sec:spencer-complex}. By \cite{Cheng-Kac}, these are the Lie superalgebras $F$
with an associated compatible filtration
$F^\bullet=\cdots= F^{-2} \supset F^{-1} \supset F^0\supset
0=\cdots$
such that the corresponding $\mathbb Z$-graded Lie superalgebra agrees
with $\fh$.  Any such
filtration $F^\bullet$ is isomorphic \emph{as a vector space} to the
canonical filtration of $\fh$ given by $F^{i} = \fh$ for all $i<-2$,
$F^{i} = 0$ for all $i>0$ and
\begin{equation*}
  F^{-2} =\fh= \fh_{-2} \oplus \fh_{-1} \oplus \fh_{0}~,\qquad
  F^{-1} = \fh_{-1} \oplus \fh_{0}~,\qquad F^0 = \fh_0~.
\end{equation*}
The Lie superalgebra structure on $F$ satisfies $[F^i,F^j] \subset
F^{i+j}$ and the
components of the Lie brackets of zero filtration degree have to coincide with
the Lie brackets of $\fh$. This is the classical approach to filtered deformations that covers the standard (i.e., even) infinitesimal deformations.
\vskip0.2cm\par

For our purposes, given any real Lie superalgebra $\fg=\fg_{\bar 0}\oplus\fg_{\bar 1}$, we consider the tensor product 
$
\fg^\Lambda:=\fg\otimes_{\mathbb R}\Lambda^\bullet\cong\Lambda^\bullet\otimes_{\mathbb R}\fg
$ of $\fg$ with an auxiliary finite-dimensional exterior algebra $\Lambda^\bullet :=\Lambda^\bullet (W)$, endowed
with its natural structure of Lie superalgebra  given by the decomposition
\begin{multline*}
\fg^\Lambda=\fg^\Lambda_{\bar 0}\oplus\fg^\Lambda_{\bar 1}\;,\quad
\fg^\Lambda_{\bar 0}=(\Lambda^\bullet_{\bar 0}\otimes\fg_{\bar 0})\oplus (\Lambda^\bullet_{\bar 1}\otimes\fg_{\bar 1})\;,\quad
\fg^\Lambda_{\bar 1}=(\Lambda^\bullet_{\bar 1}\otimes\fg_{\bar 0})\oplus (\Lambda^\bullet_{\bar 0}\otimes\fg_{\bar 1})
\end{multline*}
and the Lie bracket 
$[\mathpzc t_1 X,\mathpzc t_2 Y]=(-1)^{|X||\mathpzc t_2|}\mathpzc t_1\mathpzc t_2 [X,Y]$, for all homogeneous $X,Y\in \fg$, and $\mathpzc t_1,\mathpzc t_2\in\Lambda^\bullet$. 
Note that this is a Lie superalgebra over $\Lambda^\bullet$, in the sense that the Lie bracket is $\Lambda^\bullet$-linear, with the usual rule of signs w.r.t. $\mathbb Z_2$-grading at hand: $[\mathpzc t_1 X,\mathpzc t_2 Y]=(-1)^{|X||\mathpzc t_2|}\mathpzc t_1\mathpzc t_2 [X,Y]$ for all homogeneous $X,Y\in \fg^\Lambda$ -- not only belonging to $\fg$ -- and  $\mathpzc t_1,\mathpzc t_2\in\Lambda^\bullet$. Nonetheless, the even and odd components of $\fg^\Lambda$ are {\it not} modules over $\Lambda^\bullet$, but only over $\Lambda^\bullet_{\bar 0}$. For any fixed finite-dimensional vector space $W$ and associated exterior algebra $\Lambda^\bullet :=\Lambda^\bullet (W)$, we then give the following. 
\begin{definition}
\label{def:filtereddeformation}
A filtered deformation of $\fh$ (parametrized by $W$) is the datum of a Lie superalgebra $F$ supported on the vector superspace $\fh^\Lambda=\fh\otimes_{\mathbb R}\Lambda^\bullet$ such that:
\begin{itemize}
    \item[(i)]  the bracket is $\Lambda^\bullet$-linear,
    \item[(ii)] the bracket preserves the grading on $\fh^\Lambda=\fh\otimes_{\mathbb R}\Lambda^\bullet$ inherited from the grading of $\fh$ and the natural non-positive grading of $\Lambda^\bullet$ (i.e., the elements of $W$ have degree $-1$),
    \item[(iii)] the components of the bracket with coefficients in $\mathbb R=\Lambda^0$
coincide with those of $\fh$.
    \end{itemize}
\end{definition}
We can therefore
describe a filtered deformation according to Definition \ref{def:filtereddeformation} by the
following brackets
\begin{equation}
\label{eq:brackets}
  \begin{aligned}
	[\fh_0,\fh_0] &\subset \fh_0\;\\
	  [\fh_{0}, V] & \subset  V\oplus (\Lambda^1\otimes S)\oplus(\Lambda^2\otimes\fh_{0})\;\\
			  [\fh_{0}, S] & \subset  S\oplus (\Lambda^1\otimes\fh_{0})\;\\
    [S, S] & \subset  V\oplus (\Lambda^1\otimes S)\oplus  (\Lambda^2\otimes\fh_{0})\;\\
    [V, S] & \subset (\Lambda^1\otimes V)\oplus (\Lambda^2\otimes S)\oplus (\Lambda^3\otimes \fh_0)\;\\
    [V, V] & \subset (\Lambda^2\otimes V) \oplus (\Lambda^3\otimes S)\oplus(\Lambda^4\otimes \fh_{0})\;
  \end{aligned}
\end{equation}
where the bracket
components with coefficients in $\mathbb R=\Lambda^0$
should
not be modified from the ones in $\fh$. Note that
the brackets are in fact compatible with the filtration 
 $F^\bullet=\cdots \supset F^{-2} \supset F^{-1} \supset F^0\supset
\cdots$ of $\fh^\Lambda$ by $\Lambda^\bullet$-modules given by $F^{i} = \fh^\Lambda$ for all $i<-2$,
$F^{i} = 0$ for all $i>0$,
\begin{equation*}
  F^{-2} =\fh^\Lambda= \fh_{-2}^\Lambda \oplus \fh_{-1}^\Lambda \oplus \fh_{0}^\Lambda~,\qquad
  F^{-1} = \fh_{-1}^\Lambda \oplus \fh_{0}^\Lambda,\qquad F^0 = \fh_0^\Lambda\;,
\end{equation*}	
and that the associated $\mathbb Z$-graded Lie superalgebra over $\Lambda^\bullet$ agrees
with $\fh^\Lambda$. This filtered structure (together with the fact that the powers of the "parameters" in $W$ keep track of the amount by which the filtration degree fails to be preserved) is intrinsic and should be preserved by isomorphisms.
\vskip0.2cm\par

We here list
the components of the Lie brackets of non-zero filtration degree: 
\vskip0.2cm\par\noindent
\begin{itemize}
\item
the degree +$1$ components are elements 
  \begin{equation} 
	\label{eq:comp1}
 \begin{split}
\widehat\varepsilon&\in\Lambda^1\otimes\Hom(\odot^2 S, S),\;\widehat\epsilon\in\Lambda^1\otimes\Hom(S\otimes V,V),\\
\widehat\mu&\in\Lambda^1\otimes\Hom(\fh_0\otimes V,S),\;\widehat\theta\in\Lambda^1\otimes\Hom(\fh_0\otimes S, \fh_0);
  \end{split}
  \end{equation}
\item the degree +$2$ components are elements 
  \begin{equation}
		\label{eq:comp2}
 \begin{split}
      \widehat\alpha&\in\Lambda^2\otimes\Hom(\Lambda^2 V, V),\;\widehat\beta\in\Lambda^2\otimes\Hom(V\otimes S,S),\\
      \widehat\gamma &\in\Lambda^2\otimes\Hom(\odot^2 S,\fh_0),\;\widehat\delta\in\Lambda^2\otimes\Hom(\fh_0\otimes V,\fh_0);
    \end{split}
  \end{equation}
\item the degree +$3$ components are elements 
  \begin{equation} 
		\label{eq:comp3}
 \begin{split}
      \widehat\sigma&\in\Lambda^3\otimes\Hom(V\otimes S, \fh_0),\;\widehat\tau\in\Lambda^3\otimes\Hom(\Lambda^2 V,S);\\
    \end{split}
  \end{equation}
\item the degree +$4$ component is an element
  \begin{equation} 
		\label{eq:comp4}
 \begin{split}
      \widehat\rho&\in\Lambda^4\otimes\Hom(\Lambda^2 V, \fh_0).\\
    \end{split}
  \end{equation}
\end{itemize}

It is convenient to adopt the following Sweedler-like short-cut notation. We write $\widehat\epsilon=\mathpzc t\epsilon$, $\widehat\alpha=\mathpzc t^2\alpha$, $\widehat\sigma=\mathpzc t^3\sigma$, $\widehat\rho=\mathpzc t^4\rho$, and similarly for the other components \eqref{eq:comp1}-\eqref{eq:comp4}. Here $\mathpzc t^k$ simply indicates homogeneous elements of degree $k$ of the exterior algebra, but the components of the Lie bracket are not necessarily decomposable elements of $\Lambda^\bullet\otimes\Hom(\Lambda^2\fh,\fh)$. For instance, $\widehat\sigma=\mathpzc t^3\sigma=\tfrac{1}{3!}\sum \mathpzc t_{ijk}\otimes\sigma^{ijk}$, for basis elements $\mathpzc t_{ijk}$ of $\Lambda^3$, and components $\sigma^{ijk}\in\Hom(V\otimes S,\fh_0)$.
\subsection{The Jacobi identities}
\label{subsec:JI}
The Lie brackets of a filtered deformation $F$ of $\fh$ are
given by equation \eqref{eq:brackets} in terms of 	\eqref{eq:comp1}-\eqref{eq:comp4}. 
The only additional conditions come from demanding that the Lie
brackets \eqref{eq:brackets} do define a Lie superalgebra, i.e., they come from imposing the Jacobi identities for $F$.  There are ten
such identities and to go through them
systematically, we use the notation $[ijk]$, $i,j,k=0,1,2$,
for the identity involving $X \in \fh_{-i}$, $Y \in \fh_{-j}$,
$Z \in \fh_{-k}$.

\subsubsection*{The $[000]$ Jacobi}
This is automatically satisfied because $\fh_0$ is a Lie subalgebra of
$\fso(V)$.
\subsubsection*{The $[001]$ Jacobi}
Using that the action of $\fh_0$ on $S$ is
the restriction to $\fh_0$ of the spinor representation
of $\fso(V)$, we are left with the equation
\begin{align}
\label{eq:001h}
\widehat\theta([A,B],s)
-\widehat\theta(A,Bs)+\widehat\theta(B,As)-[A,\widehat\theta(B,s)]+[B,\widehat\theta(A,s)]&=0
\end{align}
for all $A,B\in\fh_0$ and $s\in S$.
\subsubsection*{The $[002]$ Jacobi}
Using that the action of $\fh_0$ on $V$ is
the restriction to $\fh_0$ of the vector representation
of $\fso(V)$, we are left with the equations
\begin{align}
\label{eq:002S}
\widehat\mu([A,B],v)-\widehat\mu(A,Bv)+\widehat\mu(B,Av)-A(\widehat\mu(B,v))+B(\widehat\mu(A,v))&=0\;,\\
\notag
\widehat\delta([A,B],v)
-\widehat\delta(A,Bv)+\widehat\delta(B,Av)-[A,\widehat\delta(B,v)]+[B,\widehat\delta(A,v)]&=\widehat\theta(A,\widehat\mu(B,v))\\
\label{eq:002h}
&\;\;\;-\widehat\theta(B,\widehat\mu(A,v))\;,
\end{align}
for all $A,B\in\fh_0$ and $v\in V$.  
\subsubsection*{The $[011]$ Jacobi}
Using that the Dirac current $\kappa:\odot^2S\to V$ is $\fso(V)$-equivariant, hence $\fh_0$-equivariant, 
we are left with the equations
\begin{align}
\label{eq:011S}
2\widehat\theta(A,s)s-\widehat\mu(A,\kappa(s,s))
&=(A\cdot\widehat\varepsilon)(s,s)\;,\\
\label{eq:011h}
\widehat\delta(A,\kappa(s,s))
+(A\cdot\widehat\gamma)(s,s)
&=2\widehat\theta(\widehat\theta(A,s),s)-\widehat\theta(A,\widehat\varepsilon(s,s))\;,
\end{align}
for all $A\in\fh_0$, $s\in S$.
\subsubsection*{The $[111]$ Jacobi}
The Jacobi identity says that $[[s,s],s] = 0$ for all $s \in S$, and
it expands to
\begin{align}
\label{eq:111V}
\widehat\epsilon(\kappa(s,s),s)+\kappa(\widehat\varepsilon(s,s),s)&=0\;,
\\
\label{eq:111S}
  \widehat\beta(\kappa(s,s),s)+\widehat\gamma(s,s)s&=-\widehat\varepsilon(\widehat\varepsilon(s,s),s)\;,\\
	\label{eq:111h}
	\widehat\sigma(\kappa(s,s),s)&=-\widehat\theta(\widehat\gamma(s,s),s)-\widehat\gamma(\widehat\varepsilon(s,s),s)\;,
\end{align}
for all $s\in S$.
\subsubsection*{The $[112]$ Jacobi}
After a somewhat lengthy calculation, this Jacobi identity reduces to
\begin{align}
  \label{112V}
\widehat\alpha(\kappa(s,s),v)+2\kappa(s,\widehat\beta(v,s))+\widehat\gamma(s,s)v&=2\widehat\epsilon(s,\widehat\epsilon(s,v))-\widehat\epsilon(\widehat\varepsilon(s,s),v)\;,\\
\notag
\widehat\tau(\kappa(s,s),v)+2\widehat\sigma(v,s)s&=-\widehat\mu(\widehat\gamma(s,s),v)-\widehat\beta(\widehat\varepsilon(s,s),v)\\
\label{112S}
&\;\;\;\;+2\widehat\beta(s,\widehat\epsilon(s,v))-2\widehat\varepsilon(s,\widehat\beta(v,s))\;,\\
\notag
\widehat\rho(\kappa(s,s),v)+\widehat\delta(\widehat\gamma(s,s),v)+2\widehat\gamma(s,\widehat\beta(v,s))&=-\widehat\sigma(\widehat\varepsilon(s,s),v)+2\widehat\sigma(s,\widehat\epsilon(s,v))\\
\label{112h}
&\;\;\;\;-2\widehat\theta(\widehat\sigma(v,s),s)
\end{align}
for all $s \in S$, $v \in V$. 
\subsubsection*{The $[012]$ Jacobi}
In this case, we have
\begin{align}	 
\label{012V}
	(A\cdot \widehat\epsilon)(v,s)
	&=\kappa(\widehat\mu(A,v),s)-\widehat\theta(A,s)v\;,\\
\notag
	(A\cdot\widehat\beta)(v,s)-\widehat\delta(A,v)s
	&=\widehat\varepsilon(\widehat\mu(A,v),s)-\widehat\mu(A,\widehat\epsilon(v,s))\\	
	\label{012S}
	&\;\;\;\;-\widehat\mu(\widehat\theta(A,s),v)\;,\\
	\notag
	(A\cdot\widehat\sigma)(v,s)
	&=\widehat\theta(\widehat\delta(A,v),s)-\widehat\delta(\widehat\theta(A,s),v)
-\widehat\delta(A,\widehat\epsilon(v,s))\\
		\label{012h}
	&\;\;\;\;+\widehat\gamma(\widehat\mu(A,v),s)-\widehat\theta(A,\widehat\beta(v,s))\;,
	\end{align}
for all $A \in \fh_0$, $v \in V$ and $s \in S$.
\subsubsection*{The $[022]$ Jacobi}
Here
\begin{align}
\label{022V}
(A\cdot\widehat\alpha)(v,w)-\widehat\delta(A,v)w+\widehat\delta(A,w)v
&=\widehat\epsilon(\widehat\mu(A,v),w)-\widehat\epsilon(\widehat\mu(A,w),v)\;,\\
\notag
(A\cdot\widehat\tau)(v,w)+\widehat\mu(A,\widehat\alpha(v,w))
&=\widehat\beta(v,\widehat\mu(A,w))-\widehat\beta(w,\widehat\mu(A,v))\\
\label{022S}
&\;\;\;\;+\widehat\mu(\widehat\delta(A,v),w)-\widehat\mu(\widehat\delta(A,w),v)\;,\\
\notag
(A\cdot\widehat\rho)(v,w)+\widehat\delta(A,\widehat\alpha(v,w))+\widehat\theta(A,\widehat\tau(v,w))
&=\widehat\sigma(v,\widehat\mu(A,w))-\widehat\sigma(w,\widehat\mu(A,v))\\
&\;\;\;\;+\widehat\delta(\widehat\delta(A,v),w)
-\widehat\delta(\widehat\delta(A,w),v)\;,
\label{022h}
\end{align}
for all $A \in \fh_0$ and $v,w \in V$.
\subsubsection*{The $[122]$ Jacobi}
This Jacobi identity expands to the following conditions
\begin{align}
\notag
\kappa(s,\widehat\tau(v,w))-\widehat\sigma(w,s)v+\widehat\sigma(v,s)w+\widehat\epsilon(s,\widehat\alpha(v,w))&=\widehat\alpha(\widehat\epsilon(s,v),w)-\widehat\alpha(\widehat\epsilon(s,w),v)\\
\label{122V}
&\;\;\;+\widehat\epsilon(\widehat\beta(w,s),v)-\widehat\epsilon(\widehat\beta(v,s),w)\;,\\
\notag
\widehat\rho(v,w)s+\widehat\beta(w,\widehat\beta(v,s))-\widehat\beta(v,\widehat\beta(w,s))&=-\widehat\beta(\widehat\alpha(v,w),s)+\widehat\varepsilon(s,\widehat\tau(v,w))
\\
\notag
&\;\;\;\;-\widehat\tau(\widehat\epsilon(s,v),w)+\widehat\tau(\widehat\epsilon(s,w),v)\\
\label{122S}
&\;\;\;\;+\widehat\mu(\widehat\sigma(v,s),w)-\widehat\mu(\widehat\sigma(w,s),v)\;,\\
\notag
\widehat\gamma(s,\widehat\tau(v,w))-\widehat\sigma(\widehat\alpha(v,w),s)-\widehat\theta(\widehat\rho(v,w),s)&=\widehat\rho(\widehat\epsilon(s,v),w)-\widehat\rho(\widehat\epsilon(s,w),v)\\
\notag
&
\;\;\;+\widehat\sigma(w,\widehat\beta(v,s))-\widehat\sigma(v,\widehat\beta(w,s))
\\
\label{122h}
&\;\;\;+\widehat\delta(\widehat\sigma(w,s),v)-\widehat\delta(\widehat\sigma(v,s),w)\;,
\end{align}
for all $s \in S$ and $v,w \in V$.
\subsubsection*{The $[222]$ Jacobi}
Finally the last component of the Jacobi identity expands to the following three
Bianchi-like identities:
\begin{align}
  \label{222V}
\mathfrak S\big(\widehat\rho(w,u)v-\widehat\alpha(v,\widehat\alpha(w,u))\big)&=\mathfrak S\big(\widehat\epsilon(v,\widehat\tau(w,u))\big)\;,\\
\label{222S}
\mathfrak S\big(\widehat\mu(\widehat\rho(w,u),v)-\widehat\tau(v,\widehat\alpha(w,u))\big)&=\mathfrak S\big(\widehat\beta(v,\widehat\tau(w,u))\big)\;,\\
\label{222h}
\mathfrak S\big(\widehat\delta(\widehat\rho(w,u),v)-\widehat\rho(v,\widehat\alpha(w,u))\big)&=\mathfrak S\big(\widehat\sigma(v,\widehat\tau(w,u))\big)\;,
\end{align}
for all $u,v,w\in V$. (The symbol $\mathfrak S$ denotes the cyclic sum over $u,v,w$.)
\vskip0.3cm\par\noindent

\subsection{Analysis at filtration degrees $1$ and $2$}
\label{subsec:5.3}

Now, it is well-known that the restriction to $\fm$ of the first non-zero contribution of a filtered
deformation is a cohomology class in positive degree which is $\fh_0$-invariant (see \cite[Prop. 2.2]{Cheng-Kac}).
This is true also in our more general framework, as we now explain. 

Equation \eqref{eq:111V} is separately satisfied for the different components of $\widehat\epsilon=\mathpzc t\epsilon$ and $\widehat\varepsilon=\mathpzc t\varepsilon$, thus yielding the identity 
$$\epsilon(\kappa(s,s),s)+\kappa(\varepsilon(s,s),s)=0\;,$$ for all $s\in S$. This is the condition that $\epsilon+\varepsilon\in C^{1,2}(\fm,\fh)$ is a Spencer cocycle. Similarly, equations \eqref{eq:011S} and \eqref{012V} are satisfied separately for all the components  and say that $A\cdot(\epsilon+\varepsilon)=\partial(\imath_A\theta+\imath_A\mu)$, which is a Spencer coboundary for all $A\in\fh_0$.
In the case of a coboundary, the contribution can be absorbed component by component via a redefinition of the complementary
subspaces in the chain of filtrands (the proof of \cite[Prop. 2.3]{Cheng-Kac} extends verbatim to our case). 
All in all, we are interested in the cohomology class of $\epsilon+\varepsilon\in C^{1,2}(\fm,\fh)$, which follows from Propositions \ref{prop:cohomology1h}- \ref{prop:cohomology1p}:

\begin{corollary}
\label{cor:form-inf-def}
Let $F$ be a filtered deformation of $\fh$. Then its infinitesimal deformation $\varepsilon+\epsilon=\varepsilon^\phi+\epsilon^\phi+\partial X_S=\partial(Z\otimes\phi)+\partial X_S$, where $Z$ is the grading element of $\fh$, $\phi\in S^*$, and $X_S:S\to\fso(V)$ is defined up to elements in the subalgebra $\fh_0$ of $\fso(V)$. Explicitly
\begin{align*}
\varepsilon(s,s)&=-2\phi(s)s+2X_s(s)\;,\\
\epsilon(s,v)&=-2\phi(s)v+X_s(v)\;,
\end{align*}
for all $s\in S$, $v\in V$.
\end{corollary}

We collect some further important consequences of these results.
\begin{proposition}
\label{prop:stabilizingphi}
Let $F$ be a filtered deformation of $\fh$. Then:
\begin{itemize}
	\item[(i)] $\fh_0\subset\stab_{\fso(V)}(\phi)$, 
	\item[(ii)] the component $\mu$ of filtration degree $1$ vanishes,
	\item[(iii)] the component $\theta$ of filtration degree $1$ satisfies $\imath_A\theta=A\cdot X_S$ for all $A\in\fh_0$.
\end{itemize}
\end{proposition}
\begin{proof}
As already observed at the beginning of \S\ref{subsec:5.3}, the infinitesimal deformation is $\fh_0$-invariant as a cohomology class: $A\cdot(\epsilon+\varepsilon)=\partial(\imath_A\theta)+\partial(\imath_A\mu)$ for all $A\in\fh_0$. Using Corollary \ref{cor:form-inf-def} and the fact that the operator $\partial:C^{1,1}(\fm,\fp\oplus\mathbb R Z)\to C^{1,2}(\fm,\fp\oplus\mathbb R Z)$ is $\fso(V)$-equivariant and injective (as we already explained in \S\ref{sec:spencer-complex}, injectivity follows from the results of \cite{MR3255456}), we then infer
\begin{align*}
\partial(\imath_A\theta)+\partial(\imath_A\mu)&=A\cdot(\epsilon+\varepsilon)=A\cdot\big(\partial(Z\otimes\phi)+\partial X_S\big)\\
&=\partial\big(Z\otimes(A\cdot\phi)\big)+\partial \big(A\cdot X_S\big)\\
&\Rightarrow
\begin{cases}
\imath_A\theta&=A\cdot X_S,\\
\imath_A\mu&=0,\\
Z\otimes (A\cdot\phi)&=0,
\end{cases}
\end{align*}
for all $A\in\fh_0$. This proves the three claims.
\end{proof}
Our infinitesimal analysis on the equations \eqref{eq:011S}, \eqref{eq:111V}, \eqref{012V} has the following consequences: the Jacobi identities \eqref{eq:001h} and \eqref{eq:002S} are automatically satisfied, while \eqref{eq:002h} and \eqref{022V} simplify to
\begin{equation}
\label{eq:uffpioggia}
\begin{aligned}
\delta([A,B],v)
-\delta(A,Bv)+\delta(B,Av)&-[A,\delta(B,v)]+[B,\delta(A,v)]=0\;,\\
(A\cdot\alpha)(v,w)-\delta(A,v)w+\delta(A,w)v
&=0\;.
\end{aligned}
\end{equation}
\begin{proposition}
Let $F$ be a filtered deformation of $\fh$. Then there is a linear map $X_V:V\rightarrow\fso(V)$ such that
\begin{align}
\label{uffpioggia2}
\alpha(v,w)&=X_vw-X_w v\;,\\
\label{uffpioggia3}
\delta(A,v)&=[A,X_v]-X_{Av}\;,
\end{align}
for all $v,w\in V$ and $A\in\fh_0$.
\end{proposition}
\begin{proof}
It is well-known that any $\alpha:\Lambda^2 V\to V$ can be expressed as $\alpha=\slashed{\partial} X_V$ for a unique $X_V:V\to\fso(V)$, where here $\slashed{\partial} $ is a component of the classical Spencer operator in pseudo-Riemannian geometry (from a more geometric perspective, the map $X_V$ encodes
a linear choice of Killing vector fields acting simply transitively at a fixed point, see e.g. \cite[\S 3.2]{Figueroa-OFarrill:2016khp}). This is \eqref{uffpioggia2}.

The second equation in \eqref{eq:uffpioggia} then reads $A\cdot\slashed{\partial}  X_V=\slashed{\partial} (\imath_A\delta)$, and by $\fso(V)$-equivariancy and injectivity of $\slashed{\partial} $, we arrive at $A\cdot X_V=\imath_A\delta$. This is \eqref{uffpioggia3}. It is now a simply check to verify that the first equation in \eqref{eq:uffpioggia} is automatically satisfied.
\end{proof}

It is now a good place to list all the remaining Jacobi identities from \S\ref{subsec:JI}:
\begin{align}
\label{eq:32}
\widehat\delta(A,\kappa(s,s))
+(A\cdot\widehat\gamma)(s,s)
&=2\widehat\theta(\widehat\theta(A,s),s)-\widehat\theta(A,\widehat\varepsilon(s,s))\;,\\
\label{eq:33}
\widehat\beta(\kappa(s,s),s)+\widehat\gamma(s,s)s&=-\widehat\varepsilon(\widehat\varepsilon(s,s),s)\;,\\
\label{eq:34}
	\widehat\sigma(\kappa(s,s),s)&=-\widehat\theta(\widehat\gamma(s,s),s)-\widehat\gamma(\widehat\varepsilon(s,s),s)\;,\\
	\label{eq:35}
\widehat\alpha(\kappa(s,s),v)+2\kappa(s,\widehat\beta(v,s))+\widehat\gamma(s,s)v&=2\widehat\epsilon(s,\widehat\epsilon(s,v))-\widehat\epsilon(\widehat\varepsilon(s,s),v)\;,\\
\notag
\widehat\tau(\kappa(s,s),v)+2\widehat\sigma(v,s)s&=-\widehat\beta(\widehat\varepsilon(s,s),v)+2\widehat\beta(s,\widehat\epsilon(s,v))\\
\label{eq:36}
&\;\;\;\;-2\widehat\varepsilon(s,\widehat\beta(v,s))\;,\\
\notag
\widehat\rho(\kappa(s,s),v)+\widehat\delta(\widehat\gamma(s,s),v)+2\widehat\gamma(s,\widehat\beta(v,s))&=-\widehat\sigma(\widehat\varepsilon(s,s),v)+2\widehat\sigma(s,\widehat\epsilon(s,v))\\
\label{eq:37}
&\;\;\;\;-2\widehat\theta(\widehat\sigma(v,s),s)\;,\\
\label{eq:38}
	(A\cdot\widehat\beta)(v,s)-\widehat\delta(A,v)s
	&=0\;,\\
\notag
	(A\cdot\widehat\sigma)(v,s)
	&=\widehat\theta(\widehat\delta(A,v),s)-\widehat\delta(\widehat\theta(A,s),v)
\\
		\label{eq:39}
	&\;\;\;\;-\widehat\delta(A,\widehat\epsilon(v,s))-\widehat\theta(A,\widehat\beta(v,s))\;,\\
	\label{eq:40}
(A\cdot\widehat\tau)(v,w)
&=0\;,\\
	\label{eq:41}
(A\cdot\widehat\rho)(v,w)+\widehat\delta(A,\widehat\alpha(v,w))+\widehat\theta(A,\widehat\tau(v,w))
&=\widehat\delta(\widehat\delta(A,v),w)
-\widehat\delta(\widehat\delta(A,w),v)\;,\\
\notag
\kappa(s,\widehat\tau(v,w))-\widehat\sigma(w,s)v+\widehat\sigma(v,s)w&=\widehat\alpha(\widehat\epsilon(s,v),w)-\widehat\alpha(\widehat\epsilon(s,w),v)\\
\notag
&\;\;\;\;+\widehat\epsilon(\widehat\beta(w,s),v)-\widehat\epsilon(\widehat\beta(v,s),w)\\
	\label{eq:42}
&\;\;\;\;+\widehat\epsilon(s,\widehat\alpha(v,w))
\;,
\\
\notag
\widehat\rho(v,w)s+\widehat\beta(w,\widehat\beta(v,s))-\widehat\beta(v,\widehat\beta(w,s))&=-\widehat\beta(\widehat\alpha(v,w),s)+\widehat\varepsilon(s,\widehat\tau(v,w))
\\
	\label{eq:43}
&\;\;\;\;-\widehat\tau(\widehat\epsilon(s,v),w)+\widehat\tau(\widehat\epsilon(s,w),v)\;,\\
\notag
\widehat\gamma(s,\widehat\tau(v,w))-\widehat\sigma(\widehat\alpha(v,w),s)-\widehat\theta(\widehat\rho(v,w),s)&=\widehat\rho(\widehat\epsilon(s,v),w)-\widehat\rho(\widehat\epsilon(s,w),v)\\
\notag
&
\;\;\;\;+\widehat\sigma(w,\widehat\beta(v,s))-\widehat\sigma(v,\widehat\beta(w,s))
\\
&\;\;\;\;+\widehat\delta(\widehat\sigma(w,s),v)-\widehat\delta(\widehat\sigma(v,s),w)\;,\\
	\label{eq:45}
\mathfrak S\big(\widehat\rho(w,u)v-\widehat\alpha(v,\widehat\alpha(w,u))\big)&=\mathfrak S\big(\widehat\epsilon(v,\widehat\tau(w,u))\big)\;,\\
	\label{eq:46}
\mathfrak S\big(\widehat\tau(v,\widehat\alpha(w,u))\big)&=-\mathfrak S\big(\widehat\beta(v,\widehat\tau(w,u))\big)\;,\\
	\label{eq:47}
\mathfrak S\big(\widehat\delta(\widehat\rho(w,u),v)-\widehat\rho(v,\widehat\alpha(w,u))\big)&=\mathfrak S\big(\widehat\sigma(v,\widehat\tau(w,u))\big)\;,
\end{align}
for all $A,B\in\fh_0$, $s\in S$ and $v,w\in V$. 
\begin{definition}
\label{def:6}
The first-order infinitesimal direction of a filtered deformation $F$ of $\fh$ is called {\it nilpotent} if the right hand sides of the equations \eqref{eq:32}, \eqref{eq:33} and \eqref{eq:35} vanish.
\end{definition}
\begin{remark}
\label{rem:10}
\hfill
{\rm
\begin{itemize}
	\item[(i)] This definition is of cohomological nature. It says that the bracket \'a la Nijenhuis--Richardson of the first-order deformation components $\epsilon,\varepsilon,\mu,\theta$, which is always cohomologous to zero in $H^{2,3}(\fh,\fh)$ (see, e.g., \cite{Fia}), it is actually required to vanish. Whereas this is not a generic assumption, it is still satisfied in many natural cases, as we now explain.
	\item[(ii)] It is not difficult to see that the right hand sides of the equations \eqref{eq:32}, \eqref{eq:33} and \eqref{eq:35} 
depend quadratically on $\phi$ and $X_S$, but with the contribution quadratic in $\phi$ that is absent. In particular, if $X_S$ vanish, the first-order infinitesimal deformation is automatically nilpotent. It is so also when the deformation is decomposable, i.e., the Lie brackets are decomposable elements of $\Lambda^\bullet\otimes\Hom(\Lambda^2\fh,\fh)$. This follows from the expression of the components $\epsilon,\varepsilon,\mu,\theta$ of degree $1$ in terms of $\phi$ and $X_S$ given in Corollary \ref{cor:form-inf-def}  and Proposition \ref{prop:stabilizingphi}: one has to choose $\widehat\phi=\mathpzc t\phi$ and $\widehat{X_S}=\mathpzc tX_S$ for the same $\mathpzc t\in\Lambda^1$.
\end{itemize}

}
\end{remark}

\begin{proposition}
\label{prop:eleven}
Let $F$ be a filtered deformation of $\fh$. If the first-order infinitesimal deformation is nilpotent, then:
\begin{itemize}
	\item[(i)] $\alpha+\beta+\gamma=\beta^\varphi + \gamma^\varphi
+ \partial X_V$ 
where $\beta^\varphi$ and $\gamma^\varphi$ are as explained in \S\ref{sec:poinc-super};
\item[(ii)]	$\fh_0\subset\stab_{\fso(V)}(\varphi)$.
\end{itemize}
By absorbing exact contributions in the chain of filtrands, we may assume w.l.o.g. that $X_V=0$, hence $\alpha=\delta=0$ as well, and $\beta=\beta^\varphi$, $\gamma=\gamma^\varphi$.
\end{proposition}
\begin{proof}
Equations \eqref{eq:33} and \eqref{eq:35} reduce to the cocycle conditions on $\alpha+\beta+\gamma$ studied in \cite{Figueroa-OFarrill:2015rfh}, so $\alpha+\beta+\gamma=\beta^\varphi + \gamma^\varphi
+ \partial X_V$ as first claimed. Using then that 
$A\cdot X_V=\imath_A\delta$ for all $A\in\fh_0$,
the equations \eqref{eq:32} and \eqref{eq:38} read $A\cdot\gamma^\varphi=A\cdot\beta^\varphi=0$. In other words $\gamma^{A\cdot \varphi}=\beta^{A\cdot\varphi}=0$ for all $A\in\fh_0$ and, using the explicit expression of $\beta^\varphi$ given in \cite{Figueroa-OFarrill:2015rfh}, we finally get $A\cdot\varphi=0$.

Finally, we note that $X_V:V\to\fh_0$ thanks to Proposition \ref{prop:cohomology1h}, so the contribution $\partial X_V$ can be absorbed via a redefinition of the complementary
subspaces in the chain of filtrands.
\end{proof}

This subsection concludes the analysis of all the Jacobi identities of filtration degree $1$ and $2$. In a nutshell, for deformations with nilpotent first-order infinitesimal direction, the identities are equivalent to the cocycle and $\fh_0$-invariance conditions for elements of the Spencer cohomology groups $H^{d,2}(\fm,\fh)$ with $d=1,2$.
 It remains to study equations \eqref{eq:34}, \eqref{eq:36}, \eqref{eq:37}, \eqref{eq:39}-\eqref{eq:47} at filtration degrees $\geq 3$.
\subsection{A no-go theorem}
To move forward, we will restrict to the case where the Lie subalgebra $\fh_0\subset \fso(V)$ is compact. 
We refer to such filtered subdeformations as with ``compact stabilizer''. In this case, we have the $\fh_0$-stable orthogonal  decomposition $V=U\oplus\mathbb R\xi$ and any $\varphi\in\Lambda^4 V$ decomposes into $\varphi=\xi\wedge\Phi^3+\Phi^4$, where $\Phi^3\in\Lambda^3 U$ and $\Phi^4\in\Lambda^4 U$. Furthermore $\fso(V)=\fh_0\oplus\fh_0^\perp$ in a canonical way as $\fh_0$-modules.

We note that, in the context of the present paper, this is an assumption of {\it genericity type}: if the infinitesimal first-order deformation $\phi\in H^{1,2}(\fm,\fh)\cong S$ is in the timelike orbit of the projectivized action of $G=\Spin^\circ(V)$ on $\mathbb P(S)$ (namely, the unique open orbit), then 
$\fh_0$ is compact: it stabilizes $\phi$ by Proposition \ref{prop:stabilizingphi} and the timelike Dirac current $\xi=\kappa(\phi,\phi)$.
\begin{proposition}
\label{prop:siamoallafine}
Let $F$ be a filtered deformation of $\fh$ that is with compact stabilizer and with nilpotent  first-order infinitesimal direction. Then $\varphi=0$, and $\beta=\gamma=0$.
\end{proposition}
\begin{proof}
We have that $\gamma^\varphi(s,s)\in\fh_0\subset\stab_{\fso(V)}(\varphi)\cap\stab_{\fso(V)}(\phi)$ for all $s\in S$ by Proposition \ref{prop:eleven} and that
$
\eta(\gamma^\varphi(\omega)v,w)
=\tfrac{1}{3}\eta(\imath_v\imath_w\varphi,\omega^{(2)})
+\tfrac{1}{6}\eta(\imath_v\imath_w\star\varphi,\omega^{(5)})
$
for all $\omega\in\odot^2 S\cong \Lambda^1 V\oplus\Lambda^2V\oplus\Lambda^5 V$ and $v,w\in V$ by \cite[Eq. (9)]{Santi:2019kpx} (see also the proof of Proposition \ref{prop:cohomology1h}). Consider the so-called ``Dirac kernel'' $\mathfrak D:=\Lambda^2 V\oplus\Lambda^5 V\subset\odot^2 S$ and note that $\gamma^\varphi(\mathfrak D)$ is a Lie subalgebra of $\fh_0$ \cite[Lemma 18]{Figueroa-OFarrill:2016khp}. Hence
\begin{align*}
0=\eta(\gamma^\varphi(\omega)v,\xi)
&=\tfrac{1}{3}\eta(\imath_v\imath_\xi\varphi,\omega^{(2)})
+\tfrac{1}{6}\eta(\imath_v\imath_\xi\star\varphi,\omega^{(5)})
\\&
=\eta(\xi,\xi)\Big(\tfrac{1}{3}\eta(\imath_v\Phi^3,\omega^{(2)})
+\tfrac{1}{6}\eta(\imath_v\bigstar\Phi^4,\omega^{(5)})\Big)
\end{align*}
for all $\omega\in\mathfrak D$ and $v\in U$. Here $\bigstar$ is the Hodge star operator on the Euclidean vector space $U$. Choosing $\omega\in \Lambda^2 U\subset\mathfrak D$ and $\omega\in\Lambda^5 U\subset\mathfrak D$, we arrive at $\Phi^3=0$ and $\Phi^4=0$, respectively. 
\end{proof}
It remains to deal with the components $\sigma,\tau$ of degree $3$ and $\rho$ of degree $4$, and the
analysis at filtration degrees $\geq 3$.

\begin{theorem}
\label{thm:6}
Let $F$ be a filtered deformation of $\fh$ (parametrized by $W$). If $F$ is with compact stabilizer and with nilpotent first-order infinitesimal direction, then there exist $\varphi^i\in S^*$
and $\fh_0$-equivariant $X^i_S:S\to\fh_0^\perp$, $i=1,\ldots,\dim W$,
such that $F$ is isomorphic to a filtered deformation whose Lie brackets are of the form
\begin{equation}
\label{eq:bracketsfinal}
  \begin{aligned}
	[A,B] &=AB-BA\\
	  [A,v] &= Av\\
			  [A,s] &= As\\
    [s, s] & =  \kappa(s,s)-2\sum \mathpzc t_{i}\otimes\phi^i(s)s+2\sum \mathpzc t_{i}\otimes X^i_s(s)
\\
    [s, v] & =-2\sum \mathpzc t_{i}\otimes\phi^i(s)v+\sum \mathpzc t_{i}\otimes X^i_s(v)\\
    [v, w] & =0
  \end{aligned}
\end{equation}
for all $A,B\in \mathfrak h_0$, $v,w\in V$, $s\in S$. In particular, $F$ is isomorphic to an odd deformation of the first-order.
\end{theorem}
\begin{proof}
Since $\alpha=\beta=\gamma=0$ up to isomorphisms thanks to Propositions \ref{prop:eleven} and \ref{prop:siamoallafine}, the equations \eqref{eq:34}, \eqref{eq:36}, \eqref{eq:42} of filtration degree $3$ reduce to the cocyle conditions on the maps $\tau$ and $\sigma$ studied in \S\ref{subsubH32}. We have $Z^{3,2}(\fm,\fh)=H^{3,2}(\fm,\fh)=0$ due to Proposition \ref{prop:cohomology1h} and Theorem \ref{thm:H32}, thus $\tau=\sigma=0$. Finally $\rho=0$ from either	\eqref{eq:37} or \eqref{eq:43}. This shows that, up to isomorphism, the higher degree components \eqref{eq:comp1}-\eqref{eq:comp4} of the Lie brackets of $F$ all vanish, except for the degree $1$ components $\widehat\epsilon$, $\widehat\varepsilon$, $\widehat\theta$, which are as in Corollary \ref{cor:form-inf-def} and Proposition \ref{prop:stabilizingphi}
for $(\phi,X_S)\in H^{1,2}(\fm,\fh)$. Since $\fh_0$ is compact, we may assume $X_S:S\to\fh_0^\perp$ and the identity $\imath_A\theta=A\cdot X_S$ for all $A\in\fh_0$ decouples, yielding $\theta=0$ and the $\fh_0$-equivariancy of $X_S$.
\end{proof}

\section{Conclusions}\label{Conclusions}

In Theorem \ref{thm:1} we provided the classification of the Spencer cohomology groups $H^{d,2}(\mathfrak{m},\mathfrak{p})$ in all positive $\mathbb Z$-gradings, where $\mathfrak p$ is the $D=11$ Poincaré superalgebra and $\mathfrak m$ its supertranslation ideal. We then studied  the collapsed Hilbert-Poincaré $U$-series of $\mathfrak p$ in detail: they are computed by means of the Molien-Weyl formula and the main result is summarized in Theorem \ref{thm:2}. Theorem \ref{thm:3} is an independent mathematical result regarding the applications of the Molien-Weyl formula for general graded Lie superalgebras. We then studied the integrability of maximally supersymmetric filtered subdeformations of $\fp$, establishing the no-go Theorem \ref{thm:6}
for those subdeformations whose infinitesimal odd direction is generic and nilpotent. Our result can also be regarded as an extension of the fact that all the $D=11$ {\it bosonic} supergravity backgrounds with $32$ Killing spinors feature a non-compact stabiliser (see, e.g., \cite[\S 4.3]{Figueroa-OFarrill:2015rfh}) and underlines the importance of generalizing the analysis performed in \S\ref{sec:deformation-cohomology} to the case where the infinitesimal odd direction is not generic, i.e., it is in the lightlike orbit on $\mathbb P(S)$. Another interesting research line would be to investigate the existence of such kind of odd deformations in lower dimensions, which, to our knowledge, has not been considered in the literature. For example, we plan to perform our analysis for $\mathcal N=1$ $D=4$ supergravity. In that case, the analysis of the Hilbert-Poincaré series would be much easier and the integrability of the corresponding odd cocycles more tractable too.
Table \ref{SummarizingHPMW} provides some information also on cocycles at form number $q>2$. One could expect Spencer cocycles of this kind, which necessarily exist, to deform the Poincaré superalgebra into strongly homotopy Lie algebras. The relation between these and the supergravity backgrounds is yet to be studied.

Another open question, which requires a better understanding, is the possible relation between the Spencer cocycles in $H^{1,2}(\mathfrak m,\mathfrak p)$ and the spinorial 1-forms in the spectrum of $D=11$ supergravity, appearing when strictifying the L$_3$-algebra to an ordinary Lie superalgebra (named D'Auria-Frè algebra, see \cite{DFd11} for the original derivation and \cite{FDAnew1,FDAnew2,Ravera:2021sly,FDAnew3} for the group-theoretical role played by the extra 1-form spinors). 
In particular, it would be interesting to understand the relationship with the FDAs approach and compare the normalization conditions on Cartan superconnections recently prescribed from Spencer cohomology in \cite{KST-TG} with those traditionally obtained in the rheonomic approach via Lagrangian principles.
Similar additional spinorial contributions also appeared in
\cite{Howe, CanLech}, where it was argued that they can be eliminated in a conformal framework. 
This claim has a very neat interpretation in terms of Spencer cohomology (see our Remark \ref{rem:conformalextension}), but we shall stress that those contributions are not trivial in a non-conformal framework and thus further investigations are desirable.

At last, on a different note, the techniques developed in \cite{CGNR} in the supergravity context and refined here in \S\ref{sec2} could be useful within the gauged supergravity context and, more specifically, in the duality covariant approach \cite{Samtleben:2008pe,Trigiante:2016mnt}. The latter indeed requires the introduction of higher forms and an associated \emph{tensor hierarchy} \cite{deWit:2005hv,deWit:2005ub,deWit:2007kvg,deWit:2008ta}, naturally encoded in a FDA, which could more naturally be understood in superspace through the Molien-Weyl integration formula.

\section*{Acknowledgments}
{We are grateful to the Galileo Galilei Institute for Theoretical Physics and its director Stefania De Curtis for the very kind hospitality where part of this work has been completed.} Andrea Santi would like to thank José Figueroa-O'Farrill for his very kind help with Mathematica in \S \ref{subsubsec:final12}. The work of C.A.C. is supported by the Excellence Cluster Origins of the DFG under Germany’s Excellence Strategy EXC-2094 390783311. Pietro Antonio Grassi is partially supported by research funds of Universit\`a del Piemonte Orientale. Andrea Santi is supported by a tenure-track RTDB position and acknowledges the MIUR Excellence Department Project MatMod@TOV
awarded to the Department of Mathematics, University of Rome Tor Vergata, CUP E83C23000330006.  Ruggero Noris acknowledges
GAČR grant EXPRO 20-25775X for financial support. Lucrezia Ravera is supported by an RTDA position at Politecnico di Torino. This article/publication was also supported by the ``National Group for
Algebraic and Geometric Structures, and their Applications'' GNSAGA-INdAM (Italy) and it is
based upon work from COST Action CaLISTA CA21109 supported by COST
(European Cooperation in Science and Technology), 
{\scriptsize\url{https://www.cost.eu}}.

\appendix

\section{Characters and Plethystic Exponentials of \S \ref{sec6}}\label{appchar}

Here we collect the key ingredients to implement the Molien-Weyl formula for the case of the $D=11$ Poincaré superalgebra. The relevant characters are
\begin{align}
\label{11dAa}
\chi_{V}(z) &=1 + \frac{z_5^2}{z_4}+z_1+\frac{z_2}{z_1}+\frac{z_3}{z_2}+\frac{z_4}{z_3}+\frac{1}{z_1}+\frac{z_1}{z_2}+\frac{z_2}{z_3}+\frac{z_3}{z_4}+\frac{z_4}{z_5^2} \,, \\
\chi_{S}(z)&= \frac{z_5 z_1}{z_2}+\frac{z_5 z_1}{z_3}+\frac{z_3 z_5 z_1}{z_2 z_4}+\frac{z_5 z_1}{z_4}+\frac{z_3 z_1}{z_2 z_5}+\frac{z_4 z_1}{z_2 z_5}+\frac{z_4 z_1}{z_3 z_5}+\frac{z_1}{z_5} 
\nonumber \\ &+\frac{z_5}{z_2}+\frac{z_2 z_5}{z_3}+\frac{z_5}{z_3}+\frac{z_2 z_5}{z_4}+\frac{z_3 z_5}{z_2 z_4}+\frac{z_3 z_5}{z_4}+\frac{z_5}{z_4}+z_5+\frac{z_2}{z_5}+\frac{z_3}{z_2 z_5}
\nonumber \\ &+
\frac{z_3}{z_5}+\frac{z_4}{z_2 z_5}+\frac{z_2 z_4}{z_3 z_5}+\frac{z_4}{z_3 z_5}+\frac{z_4}{z_5}
+\frac{1}{z_5}+\frac{z_2 z_5}{z_3 z_1}+\frac{z_2 z_5}{z_4 z_1}+\frac{z_3 z_5}{z_4 z_1}
\nonumber \\ &+
\frac{z_5}{z_1}+\frac{z_2}{z_5 z_1}+\frac{z_3}{z_5 z_1}+\frac{z_2 z_4}{z_3 z_5 z_1}+\frac{z_4}{z_5 z_1} \,, \label{11dAb}\\
\chi_{\mathfrak{so}(V)}(z)&= \frac{z_1^2}{z_2}+\frac{z_5^2 z_1}{z_2 z_4}+\frac{z_5^2 z_1}{z_4}+\frac{z_3 z_1}{z_2}+\frac{z_3 z_1}{z_2^2}+\frac{z_4 z_1}{z_2 z_3}+\frac{z_4 z_1}{z_3}+\frac{z_1}{z_2}+\frac{z_2 z_1}{z_3}+\frac{z_1}{z_3}+\frac{z_3 z_1}{z_2 z_4}+\frac{z_3 z_1}{z_4}\nonumber \\
&+\frac{z_4 z_1}{z_2 z_5^2}+\frac{z_4 z_1}{z_5^2}+z_1+\frac{z_5^2}{z_3}+\frac{z_3 z_5^2}{z_2 z_4}+\frac{z_2 z_5^2}{z_3 z_4}+\frac{z_5^2}{z_4}+\frac{z_3 z_5^2}{z_4^2}+z_2+\frac{z_3}{z_2}+\frac{z_4}{z_2}+\frac{z_4}{z_3}+\frac{z_2 z_4}{z_3^2} \nonumber \\
&+\frac{1}{z_2}+\frac{z_2}{z_3}+\frac{z_3^2}{z_2 z_4}+\frac{z_2}{z_4}+\frac{z_3}{z_4}+\frac{z_4^2}{z_3 z_5^2}+\frac{z_3}{z_5^2}+\frac{z_3 z_4}{z_2 z_5^2}+\frac{z_2 z_4}{z_3 z_5^2}+\frac{z_4}{z_5^2}+5+\frac{z_2 z_5^2}{z_4 z_1}+\frac{z_5^2}{z_4 z_1}+\frac{z_2}{z_1}
\nonumber \\
&+\frac{z_3}{z_2 z_1}+
\frac{z_3}{z_1}+\frac{z_2 z_4}{z_3 z_1}+\frac{z_4}{z_3 z_1}+\frac{z_2^2}{z_3 z_1}+\frac{z_2}{z_3 z_1}+\frac{z_2 z_3}{z_4 z_1}+\frac{z_3}{z_4 z_1}+\frac{z_2 z_4}{z_5^2 z_1}+\frac{z_4}{z_5^2 z_1}+\frac{1}{z_1}+\frac{z_2}{z_1^2}\,.\label{11dAc}
\end{align}
Using \eqref{MW_G}-\eqref{MW_J} and \eqref{11dAa}-\eqref{11dAb}, one can compute the associated plethystic exponentials
\begin{eqnarray*}
\label{11dB}
PE[\chi_{V}(z) u] &=& \prod_{\lambda \in \Delta_V} (1 - u z^{\lambda}) =
\left(1-u\right) \left(1-\frac{u}{z_1}\right) \left(1-u z_1\right) \left(1-\frac{u z_1}{z_2}\right) \left(1-\frac{u z_2}{z_1}\right) \nonumber \\
&&
 \left(1-\frac{u z_2}{z_3}\right) \left(1-\frac{u z_3}{z_2}\right) \left(1-\frac{u z_3}{z_4}\right) \left(1-\frac{u z_4}{z_3}\right) 
 \left(1-\frac{u z_4}{z_5^2}\right) \left(1-\frac{u z_5^2}{z_4}\right) \nonumber 
\,,\\ 
\end{eqnarray*}
\begin{eqnarray*}
\label{11dAa2}
PE[\chi_{S}(z) t] &=& \frac{1}{\prod_{\lambda \in \Delta_S} (1 - t z^\lambda)} = 
\frac{z_1^8 z_2^8 z_3^8 z_4^8 z_5^{16}}{{\rm Den}} \nonumber \\
{\rm Den} &=& 
\left(t-z_5\right) \left(t z_1-z_5\right) \left(t z_2-z_5\right) \left(t z_3-z_5\right) \left(t z_4-z_5\right) \left(t z_5-1\right) \left(t z_5-z_1\right) 
\nonumber \\
&&
\left(t z_5-z_2\right) \left(t z_5-z_3\right) \left(t z_5-z_4\right) \left(t z_2-z_1 z_5\right) \left(t z_3-z_1 z_5\right) \left(t z_4-z_1 z_5\right) 
\nonumber \\
&&
\left(t z_1 z_5-z_2\right) \left(t z_1 z_5-z_3\right) \left(t z_1 z_5-z_4\right) \left(t z_3-z_2 z_5\right) \left(t z_1 z_3-z_2 z_5\right) 
\nonumber \\
&&
\left(t z_4-z_2 z_5\right) \left(t z_1 z_4-z_2 z_5\right) \left(t z_2 z_5-z_3\right) \left(t z_2 z_5-z_1 z_3\right) \left(t z_2 z_5-z_4\right) 
\nonumber \\
&&
\left(t z_2 z_5-z_1 z_4\right) \left(t z_4-z_3 z_5\right) \left(t z_1 z_4-z_3 z_5\right) \left(t z_2 z_4-z_3 z_5\right) \left(t z_3 z_5-z_4\right) 
\nonumber \\
&&
\left(t z_3 z_5-z_1 z_4\right) \left(t z_3 z_5-z_2 z_4\right) \left(t z_2 z_4-z_1 z_3 z_5\right) \left(t z_1 z_3 z_5-z_2 z_4\right)
\,.~~~~~~
\end{eqnarray*}
Another ingredient needed for the implementation of the Molien-Weyl formula is the Haar measure $\diff\mu|_T=\phi\diff\nu$, associated with 
the maximal compact subgroup $K$ of $G$ as in \eqref{MW_D}-\eqref{MW_E}:
\begin{align*}
\label{11dC}
\diff \mu|_{T}& = \frac{1}{z_1^8 z_2^8 z_3^8 z_4^8 z_5^9}
\left(1-z_1\right) \left(1-z_2\right) \left(z_1-z_2\right) \left(z_1^2-z_2\right) \left(z_1-z_3\right) \left(z_2-z_3\right) \left(z_1 z_2-z_3\right) \left(z_2-z_1 z_3\right)  \nonumber \\
&\times 
\left(z_2^2-z_1 z_3\right) \left(z_2-z_4\right) \left(z_3-z_4\right) \left(z_1 z_3-z_4\right) \left(z_3-z_1 z_4\right) \left(z_2 z_3-z_1 z_4\right) \left(z_1 z_3-z_2 z_4\right) 
 \nonumber \\
&\times 
\left(z_3^2-z_2 z_4\right) \left(z_3-z_5^2\right) \left(z_4-z_5^2\right) \left(z_1 z_4-z_5^2\right) \left(z_4-z_1 z_5^2\right) \left(z_2 z_4-z_1 z_5^2\right) \left(z_1 z_4-z_2 z_5^2\right) 
 \nonumber \\
&\times 
\left(z_3 z_4-z_2 z_5^2\right) \left(z_2 z_4-z_3 z_5^2\right) \left(z_4^2-z_3 z_5^2\right) {\frac{1}{(2\pi i)^5} \diff z_1 \diff z_2 \diff z_3 \diff z_4 \diff z_5} \,,
\end{align*}
where the factor $\frac{1}{z_1^8 z_2^8 z_3^8 z_4^8 z_5^9}$ is an overall  contribution arising from the Weyl weight function $\phi$ and the normalized Haar measure $\displaystyle\diff\nu$ on $T$.


\begin{thebibliography}{99}

\bibitem{deMedeiros:2016srz}
P.~de Medeiros, J.~Figueroa-O'Farrill and A.~Santi,
{\it Killing superalgebras for Lorentzian four-manifolds},
JHEP \textbf{06}, 106 (2016).



\bibitem{deMedeiros:2018ooy}
P.~de Medeiros, J.~Figueroa-O'Farrill and A.~Santi,
{\it Killing superalgebras for Lorentzian six-manifolds},
J. Geom. Phys. \textbf{132}, 13-44 (2018).


\bibitem{Figueroa-OFarrill:2015rfh}
J.~Figueroa-O'Farrill and A.~Santi,
{\it Spencer cohomology and 11-dimensional supergravity},
Commun. Math. Phys. \textbf{349}, 627--660 (2017).

\bibitem{Figueroa-OFarrill:2015tan}
J.~Figueroa-O'Farrill and A.~Santi,
{\it Eleven-dimensional supergravity from filtered subdeformations of the Poincar\'e superalgebra},
J. Phys. A \textbf{49}, 295204 (2016).

\bibitem{Figueroa-OFarrill:2016khp}
J.~Figueroa-O'Farrill and A.~Santi,
{\it On the algebraic structure of Killing superalgebras},
Adv. Theor. Math. Phys. \textbf{21}, 1115-1160 (2017).


\bibitem{Santi:2019kpx}
A.~Santi,
{\it Remarks on highly supersymmetric backgrounds of 11-dimensional supergravity},
 Proceedings of the Abel
 Symposium 2019 ``Geometry, Lie Theory and Applications'', Springer series ``Abel Symposia'', 253--277 (2022).



\bibitem{Santi:2010kb}
A.~Santi and A.~Spiro,
{\it Super-Poincare algebras, space-times and supergravities (I)},
Adv. Theor. Math. Phys. \textbf{16}, 1411-1441 (2012).

\bibitem{Santi:2011mc}
A.~Santi and A.~Spiro,
{\it Super-Poincare algebras, space-times and supergravities (II)},
J. Math. Phys. \textbf{53}, 032505 (2012).



\bibitem{MR4316462}
A. Beckett, J. Figueroa-O'Farrill, {\it Killing superalgebras for lorentzian five-manifolds},
JHEP {\bf 7} (2021).

\bibitem{MR4506436}
J. Figueroa-O'Farrill, G. Franchetti, 
{\it Kaluza-Klein reductions of maximally supersymmetric five-dimensional Lorentzian spacetimes},
Classical Quantum Gravity {\bf 39}, 40 pp (2022).

\bibitem{Cheng-Kac}
S.~J.~Cheng and V.~G.~Kac,
{\it Generalized Spencer cohomology and filtered deformations of $\mathbb Z$-graded Lie superalgebras},
Adv. Theor. Math. Phys. \textbf{2}, 1141-1182 (1998).

\bibitem{CE}
C. Chevalley and S. Eilenberg,
{\it Cohomology Theory of Lie Groups and Lie Algebras}, 
Trans. Amer. Math. Soc. \textbf{63}, 85-124 (1948).

\bibitem{Nahm}
W. Nahm, {\it Supersymmetries and their representations}, Nucl. Phys.
{\bf B135}, 149--166 (1978).

\bibitem{Cremmer:1978km}
E.~Cremmer, B.~Julia and J.~Scherk,
{\it Supergravity Theory in Eleven-Dimensions},
Phys. Lett. B \textbf{76}, 409-412 (1978).

\bibitem{KST-TG}
B.\ Kruglikov, A.\ Santi, D.\ The
{\it Symmetries of supergeometries related to nonholonomic superdistributions},
 Transform. Groups {\bf 29} (2024), 179--229.


\bibitem{PVN-Z2}
P.\ van\ Nieuwenhuizen, {\it Unitarity of supergravity and $\mathbb Z_2$ or $\mathbb Z_2\times \mathbb Z_2$ or $\mathbb Z_2\times \mathbb Z_2\times\mathbb Z_2$ gradings of gauge and ghost fields}, J. Math. Phys. {\bf 21}, 2562--2566 (1980).


\bibitem{DerksenKemper}
H.~Derksen and G.~Kemper 
{\it Computational Invariant Theory}, New York: Springer (2002).

\bibitem{procesi}
C. Procesi, 
{\it Lie groups: An approach through invariants and representations},
Universitext. New York, NY: Springer. (2007). 

\bibitem{CGNR}
C.\ A.\ Cremonini, P.\ A.\ Grassi, R.\ Noris, L.\ Ravera, {\it Supergravities and Branes from Hilbert-Poincaré Series},  JHEP {\bf 12},  088 (2023).

\bibitem{sullivan}
D. Sullivan, {\it Infinitesimal computations in topology}, 
Publications Math\'ematiques de l'IHES {\bf 47}, 269-331 (1977).

\bibitem{CDF}
L.~Castellani, R.~D'Auria and P.~Frè,
{\it Supergravity and superstrings: A Geometric perspective}, Vol. 1-3, Published in Singapore: World Scientific, 1991.

\bibitem{FDAdual1}
L.~Castellani and A.~Perotto,
{\it Free differential algebras: Their use in field theory and dual formulation},
Lett.\ Math.\ Phys.\  {\bf 38}, 321 (1996).








\bibitem{gm3}
L.~Castellani, P.~Frè, F.~Giani, K.~Pilch and P.~van Nieuwenhuizen,
{\it Beyond $d=11$ Supergravity and Cartan Integrable Systems},
Phys.\ Rev.\ D {\bf 26}, 1481 (1982).

\bibitem{gm13}
R.~D'Auria, P.~Frè and T.~Regge,
{\it Graded Lie Algebra Cohomology and Supergravity},
Riv.\ Nuovo Cim.\  {\bf 3N12}, 1 (1980).

\bibitem{DFd11}
R.~D'Auria and P.~Frè,
{\it Geometric Supergravity in d = 11 and Its Hidden Supergroup},
Nucl. Phys. B \textbf{201} (1982), 101-140
[erratum: Nucl. Phys. B \textbf{206} (1982), 496]

\bibitem{Fia}
A.\ Fialowski, {\it An example of formal deformations of Lie algebras}, in Deformation Theory of Algebras and Structures and Applications, Kluwer 1988, 375--401.

\bibitem{Pouliot:1998yv}
P.~Pouliot,
{\it Molien function for duality},
JHEP \textbf{01}, 021 (1999).

\bibitem{Weyl} H.~Weyl, 
{\it Classical Groups: their Invariants and Representations}, 
Princeton University Press (1997).

\bibitem{Feng:2007ur}
B.~Feng, A.~Hanany and Y.~H.~He,
{\it Counting gauge invariants: The Plethystic program},
JHEP \textbf{03}, 090 (2007).

\bibitem{Hanany:2008sb}
A.\ Hanany, N.\ Mekareeya and G.\ Torri,
{\it The Hilbert Series of Adjoint SQCD},
Nucl. Phys. B \textbf{825}, 52-97 (2010).

\bibitem{Bin}
B.\ Binegar, {\it Cohomology and deformations of Lie superalgebras}, Lett. Math. Phys. {\bf 12} (1986), no. 4, 301--308.

\bibitem{Lei}
D.\ A.\ Leites, {\it Cohomology of Lie superalgebras}, Funkcional. Anal. i Prilozen. {\bf 9} (1975), no. 4, 75--76.

\bibitem{MR3255456}
A.~Altomani and A.~Santi, {\it Classification of maximal transitive prolongations of super-{P}oincaré algebras}, Adv. Math. {\bf 265}, 60--96 (2014). 

\bibitem{Ga}
A.\ S.\ Galaev, {\it Irreducible complex skew-Berger algebras}, Differential Geom. Appl. {\bf 27}, 743--754, 2009.




\bibitem{FDAnew1} 
L.~Andrianopoli, R.~D'Auria and L.~Ravera,
{\it Hidden Gauge Structure of Supersymmetric Free Differential Algebras},
JHEP {\bf 1608}, 095 (2016).

\bibitem{FDAnew2}
L.~Andrianopoli, R.~D'Auria and L.~Ravera,
{\it More on the Hidden Symmetries of 11D Supergravity},
Phys.\ Lett.\ B {\bf 772}, 578 (2017).



\bibitem{Ravera:2021sly}
L.~Ravera,
{\it On the~Hidden Symmetries of~ \(D=11\) Supergravity},
Springer Proc. Math. Stat. \textbf{396} (2022), 211-222.

\bibitem{FDAnew3}
L.~Ravera,
{\it Hidden role of Maxwell superalgebras in the free differential algebras of D = 4 and D = 11 supergravity},
Eur. Phys. J. C \textbf{78}, 211 (2018). 

\bibitem{CanLech}
A.\ Candiello, K.\ Lechner, {\it Duality in supergravity theories},
Nuclear Physics B {\bf 412}, 479--501 (1994).

\bibitem{Howe}
P.\ S.\ Howe, {\it Weyl superspace},
Phys. Lett. B {\bf 415}, 149--155 (1997).






\bibitem{Samtleben:2008pe}
H.~Samtleben,
{\it Lectures on Gauged Supergravity and Flux Compactifications},
Class. Quant. Grav. \textbf{25} (2008), 214002.

\bibitem{Trigiante:2016mnt}
M.~Trigiante,
{\it Gauged Supergravities},
Phys. Rept. \textbf{680} (2017), 1-175.




\bibitem{deWit:2005hv}
B.~de Wit and H.~Samtleben,
{\it Gauged maximal supergravities and hierarchies of nonAbelian vector-tensor systems},
Fortsch. Phys. \textbf{53} (2005), 442-449.

\bibitem{deWit:2005ub}
B.~de Wit, H.~Samtleben and M.~Trigiante,
{\it Magnetic charges in local field theory},
JHEP \textbf{09} (2005), 016.

\bibitem{deWit:2007kvg}
B.~de Wit, H.~Samtleben and M.~Trigiante,
{\it The Maximal D=4 supergravities},
JHEP \textbf{06} (2007), 049.


\bibitem{deWit:2008ta}
B.~de Wit, H.~Nicolai and H.~Samtleben,
{\it Gauged Supergravities, Tensor Hierarchies, and M-Theory},
JHEP \textbf{02} (2008), 044.


\end{thebibliography}
\end{document}